\documentclass[10pt]{amsart}
\pdfoutput=1
\usepackage[english]{babel}
\usepackage{graphicx,amsmath}
\usepackage{enumitem}
\usepackage{MnSymbol}
\usepackage{subfig}
\usepackage{float}
\usepackage{floatflt,young,youngtab}
\usepackage{tikz}
\usepackage{blkarray}
\usepackage{multirow}
\theoremstyle{plain}
\numberwithin{equation}{section}

\newtheorem{theorem}{Theorem}[section]
\newtheorem{corollary}[theorem]{Corollary}
\newtheorem{lemma}[theorem]{Lemma}
\newtheorem{proposition}[theorem]{Proposition}
\newtheorem{remark}[theorem]{Remark}

\newtheorem{definition}[theorem]{Definition}
\newtheorem{example}[theorem]{Example}
\newtheorem{conjecture}[theorem]{Conjecture}
\newtheorem{construction}[theorem]{Construction}

\def \GTNN {{Gr^{\mbox{\tiny TNN}} (k,n)}}
\def \GTP {{Gr^{\mbox{\tiny TP}} (k,n)}}

\def \DKP {{\mathcal D}_{\textup{\scriptsize KP},\Gamma}}

\def \DS {{\mathcal D}_{\textup{\scriptsize S},\Gamma_0}}

\def \S {{\mathcal S}_{\mathcal M}^{\mbox{\tiny TNN}}}
\def \Sprime {{\mathcal S}_{ \overline{\mathcal M}}^{\mbox{\tiny TNN}}}

\title[Kasteleyn theorem, geometric signatures and KP-II divisors]{Kasteleyn theorem, geometric signatures and KP-II divisors on planar bipartite networks in the disk}
\author{Simonetta Abenda}
\address{Dipartimento di Matematica and Alma Mater Research Center on Applied Mathematics, Universit\`a di Bologna, ITALY; 
INFN, sez. di Bologna, Italy
}
\email{simonetta.abenda@unibo.it
}

\thanks{
This research has been partially supported by GNFM-INDAM and RFO University of Bologna.}

\begin{document}

\begin{abstract}
{Maximal minors of Kasteleyn sign matrices on planar bipartite graphs in the disk count dimer configurations with prescribed boundary conditions, and the weighted version of such matrices provides a natural parametrization of the totally non--negative part of real Grassmannians \cite{PSW,Lam1, Lam2,Sp, AGPR}. In this paper we provide a geometric interpretation of such variant of Kasteleyn theorem: a signature is Kasteleyn if and only if it is geometric in the sense of \cite{AG4}. We apply this geometric characterization to explicitly solve the associated system of relations and provide a new proof that the parametrization of positroid cells induced by Kasteleyn weighted matrices coincides with that of Postnikov boundary measurement map.
Finally we use Kasteleyn system of relations to associate algebraic geometric data to KP multi-soliton solutions. Indeed the KP wave function solves such system of relations at the nodes of the spectral curve if the dual graph of the latter represents the soliton data. Therefore the construction of the divisor is automatically invariant, and finally it coincides with that in \cite{AG3,AG5} for the present class of graphs.}

\medskip \noindent {\sc{2010 MSC.}} 05C90, 14H70, 14M15, 37K40.

\noindent {\sc{Keywords.}} Totally non-negative Grassmannians, positroid cells, planar bipartite networks in the disk, duality, almost perfect matching, Kasteleyn signatures, M--curves, KP hierarchy, real soliton and finite-gap solutions.
\end{abstract}
\maketitle

\tableofcontents
\section{Introduction}
Kasteleyn \cite{Kas1} and Temperley--Fisher \cite{TF} started the study of the dimer model by computing the number of dimer configurations on a rectangular grid. This result was then generalized by Kasteleyn \cite{Kas2} who related the number of perfect matchings on a finite planar graph to the square root of the determinant of a matrix. More recently dimer models on planar bipartite periodic graphs have appeared in mathematical literature because of their relation to combinatorics, algebraic geometry and quantum integrable systems \cite{EKLP, CEP, CKP, KO, KOS, KG}. 

Also dimer models on planar bipartite graphs in the disk present interesting features because of their connection to total positivity, toric geometry, theoretical physics and integrable systems \cite{PSW, Sp, Lam1, Lam2, AGPR}. A dimer configuration on such a graph $\mathcal G$ is an almost perfect matching, {\sl i.e.} a subset of edges such that each internal vertex is used exactly once whereas the boundary vertices may or may not be used. In \cite{PSW} it was pointed out the existence of a bijection between almost perfect matchings and perfect orientations of $\mathcal G$, and toric geometry was used to investigate the topology of totally non--negative Grassmannians.  The connection between dimer models on planar bipartite graphs in the disk and totally non--negative Grassmannians was then studied in \cite{Lam1, Lam2}.

The variant of Kasteleyn theorem relevant in such setting is the following one \cite{Sp}: for a planar bipartite graph in the disk with boundary vertices of the same color, there exists a sign matrix such that its
maximal minors give the number of almost perfect matchings with prescribed boundary conditions. Moreover, if one assigns positive weights to the edges of the graph, the maximal minors of the weighted version of the sign matrix are the Pl\"ucker coordinates of the point in the totally non-negative Grassmannian obtained from Postnikov boundary measurement map \cite{Sp}. 

The proof in \cite{Sp} is topological, whereas in this paper we focus on the explicit representation and the geometric nature of Kasteleyn matrices, since we are interested in their application to integrable systems. We prove that the total signature of a face depends only on the number of edges bounding it, and that they are realized by the geometric signatures introduced in \cite{AG4} on directed plabic graphs.
This geometric representation of Kasteleyn signatures is relevant both to solve Kasteleyn system of relations using the Talaska flows \cite{Tal2} and to provide an alternative proof that the parametrization of positroid cells in terms of the maximal minors of Kasteleyn sign matrices coincides with that of the image of Postnikov boundary measurement map.

Then, we use Kasteleyn system of relations to assign algebraic geometric data to the family of real regular multi--line solutions of the Kadomtsev--Petviashvili (KP) integrable hierarchy introduced in \cite{Mat, FN} and studied in \cite{Mal,BK, BPPP, CK, KW2014, AG1, A1, AG3, AG5}. At this aim we use the approach to degenerate finite--gap spectral theory proposed in \cite{Kr3} and applied 
to Toda and KP integrable systems in \cite{KV, AG1, A1, AG3, Nak1, AG5, BEN}. We remark that our final goal is the search of effective methods to detropicalize spectral curves and obtain KP solutions fulfilling \cite{DN} as done in \cite{AG2} for soliton data in $Gr^{\mbox{\tiny TP}}(2,4)$. Moreover, the inverse problem of constructing soliton KP solutions from real regular spectral data on reducible nodal curves was solved in \cite{A1} in a special case (see also \cite{Nak1}), but remains an open problem in general. In \cite{AFMS} they propose an algebraic method based on Dubrovin threefold for the KP equation to study tropicalization of algebraic curves and construct KP soliton solutions from spectral data on reducible nodal curves: therefore it would be relevant to relate the two approaches when reality and regularity conditions hold.

\smallskip

Finally, systems of relations have been introduced in \cite{Lam2} to provide a mathematical framework for the construction of scattering amplitudes in $N=4$ SYM theory \cite{AGP2}, since such systems explicitly realize the totally non--negative part of any positroid cell via the amalgamation of the little positive Grassmannians, $Gr^{\mbox{\tiny TP}}(1,3)$ and $Gr^{\mbox{\tiny TP}}(2,3)$. We are convinced that a geometric approach to the problem may give a new insight to this matter, and that systems of relations may be fruitfully applied also for other problems in mathematical or theoretical physics related to totally non--negative Grassmannians.

\smallskip

\textbf{Outline of the main results}
Let $\mathcal G = (\mathcal B \cup \mathcal W, \mathcal E)$ be a reduced planar bipartite graph in the disk  with boundary vertices of equal color, where $\mathcal{B}, \mathcal{W}$ are the sets of black, white vertices of the graph respectively, and $\mathcal{E}$ is the set of edges. In this paper, we call Kasteleyn an edge signature $\sigma: \mathcal E \mapsto \{ \pm 1 \}$ such that the face signature $\sigma(\Omega)$ fulfills Kasteleyn condition for any finite face $\Omega$:
\begin{equation}\label{eq:1}
\sigma(\Omega) \equiv  \prod_{e \in \partial \Omega} \sigma(e)= (-1)^{\frac{|\Omega|}{2}+1}, 
\end{equation}
where $|\Omega|$ denotes the number of edges bounding $\Omega$.  Then, in Theorem \ref{theo:kast} we show that the $|\mathcal B|\times |\mathcal W|$ Kasteleyn sign matrix $K^{\sigma}$ associated to such data fulfills Speyer variant \cite{Sp} of Kasteleyn theorem:
\begin{itemize}
\item The maximal minors of $K^{\sigma}$ indexed by the boundary dimer configurations share the same sign and their absolute value is the number of almost perfect matchings for the given boundary conditions;
\item If one fixes a Kasteleyn signature as in (\ref{eq:1}) and assigns positive edge weights to the graph, the maximal minors of the weighted Kasteleyn matrix $K^{\sigma,wt}$ defined in (\ref{eq:kas_wt_entries}) are the Pl\"ucker coordinates of the point in the totally non--negative Grassmannian given by the boundary measurement map introduced in \cite{Pos}.
\end{itemize}

\smallskip

Here we choose black boundary vertices on the reduced bipartite graph, and characterize Kasteleyn system of relations associated to a signature fulfilling (\ref{eq:1}) (see Section \ref{sec:kas_rel}) motivated by Speyer representation of the Kasteleyn weighted matrix (see Formula (\ref{eq:mat_A})). However, we remark that the equal color of the boundary vertices is just a technical assumption which simplifies both the representation of Kasteleyn matrices and the solutions to the induced system of relations. 

For the systems of relations on bipartite graphs we use the representation introduced in \cite{AGPR}. More precisely, 
for any signature fulfilling (\ref{eq:1}) and for an edge weighting $t_{bw} : \mathcal E \mapsto \mathbb{C}^*$ we call Kasteleyn the system $(v^{(k)}= \{v_b^{(k)}: b \in \mathcal{B}\}, R_w)$, where:
\begin{itemize}
\item $v^{(k)}_b$ is an element in some vector space $V$ assigned to the black vertex $b\in 	\mathcal B$;
\item $R_w$ is the linear relation at the white vertex $w\in \mathcal W$ represented by the $w$--th column of the weighted Kasteleyn matrix $K^{\sigma,wt}$: $R_w (v^{(k)}) \equiv \sum_{b\in \mathcal B} K^{\sigma,wt}_{bw} v^{(k)}_{b} =0$, where $K^{\sigma,wt}_{bw}=\sigma (\overline{bw}) t_{bw}$, and $\overline{bw}$ denotes an edge.
\end{itemize}
Theorem \ref{theo:kast} implies that the system has maximal rank equal to the number of white vertices, and its kernel is $(n-k)$--dimensional.
In particular, if $V=\mathbb{C}^{n-k}$, the system induces an isomorphism between dual positroid varieties (Theorem \ref{theo:sol_kas_sys_1}). If $V=\mathbb{R}^n$ and one fixes natural boundary conditions at $(n-k)$ boundary vertices, then the solution of the linear system at the remaining $k$ boundary vertices provides a representative matrix of the network $(\mathcal G, t_{bw})$ (Theorem \ref{theo:sol_kas_sys_2}).

We remark that the Kasteleyn signatures introduced in \cite{AGPR} differ from those fulfilling (\ref{eq:1}) at the external faces of the graph. A natural connection to total positivity also holds in their approach and has applications to discrete integrable dynamical systems such as the pentagram map \cite{Schw}, $Q$--nets \cite{DS, BS} and discrete Darboux maps \cite{Schi}. 
In Section \ref{sec:AGPR} we review the main results in \cite{AGPR} and explicitly describe the transformation which relates the two types of Kasteleyn signatures.

\smallskip

Then in Section \ref{sec:geom_sign} we provide a geometric construction of Kasteleyn signatures. Indeed we show that Kasteleyn signatures are equivalent to the geometric signatures introduced in \cite{AG4} for the class of graphs under consideration. In \cite{AG4} a geometric signature is uniquely and explicitly assigned to the edges of a planar bicolored (plabic) graph using two geometric indices: the local winding number and the intersection number. These indices are ruled by perfect orientations and gauge ray directions which behave as gauge transformations; therefore there exists a unique equivalence class of geometric signatures on $\mathcal G$ \cite{AG4}.
The explicit relation between the equivalence class of Kasteleyn and that of geometric signatures is given in Theorem \ref{theo:main}:
If $\mathcal G$ is a reduced bipartite graph, $\sigma$ a Kasteleyn signature on $\mathcal G$ satisfying (\ref{eq:1}), and $\epsilon^{(g)}$ a geometric signature on $\mathcal G$ as in Definition \ref{def:geo_sign_expl}, then
the Kasteleyn face signature $\sigma(\Omega)$, 
\[
\sigma(\Omega)= \prod_{e \in \partial \Omega} \sigma(e),
\] 
and the geometric face signature $\epsilon^{(g)}(\Omega)$,
\begin{equation}\label{eq:2}
\epsilon^{(g)}(\Omega) = \sum_{e\in\partial\Omega} \epsilon^{(g)}(e).
\end{equation}
are related as follows on each finite face $\Omega$:
\begin{equation}\label{eq:3}
\sigma(\Omega) = 
(-1)^{\epsilon^{(g)}(\Omega)}.
\end{equation}

Such geometric characterization is then used to explicitly solve Kasteleyn system of relations $(v^{(k)}, R_w)$ using Talaska flows.
Flows and conservative flows on directed graphs, originally introduced in \cite{Tal2} to compute the boundary measurement map, were used in \cite{AG4} to construct the explicit solution to the system of relations for geometric signatures. Therefore they also provide the explicit solution to Kasteleyn system of relations at all internal vertices (Theorem \ref{theo:kas_geo_sys}).

In \cite{AG4} geometric signatures are defined on the more general class of perfectly oriented plabic graphs such that every edge belongs to at least one directed path starting and ending at the boundary of the disk. We remark that this condition is exactly the one for which the gauge freedom in the definition of the geometric signature is fixed by the gauge transformations at the vertices of the graph. Since plabic graphs are not necessarily bipartite, in the setting of \cite{AG4, AG5}, (\ref{eq:3}) becomes the definition of Kasteleyn signature, and (\ref{eq:1}) is replaced by the following relation
\begin{equation}\label{eq:4}
\sigma(\Omega) \equiv  \prod_{e \in \partial \Omega} \sigma(e)= (-1)^{\frac{n_w(\Omega)}{2}+1}, \quad \quad \mbox{ for any finite face } \Omega,
\end{equation}
where $n_w(\Omega)$ is the number of internal white vertices bounding $\Omega$. We conjecture that signatures on planar non--bipartite graphs in the disk satisfying (\ref{eq:4}) may be also given a statistical mechanical interpretation. 

\smallskip

The maximal minors of the Kasteleyn weighted sign matrices are the Pl\"ucker coordinates of points in totally non--negative Grassmannians $Gr^{\mbox{\tiny TNN}} (k,n)$, and such parametrization is equivalent to that of Postnikov boundary measurement map \cite{Sp}. 
In the present setting Theorems \ref{theo:main} and \ref{theo:kas_geo_sys} imply an alternative proof of such equivalence (Corollary \ref{cor:pos_sp_equiv}).

\smallskip

Then in Section \ref{sec:comb} we apply Kasteleyn system of relations to the spectral problem for the KP--II real regular multi--line solitons. In such setting: 
\begin{itemize}
\item The point $[A]$ represented by the network $(\mathcal G, t_{bw})$ and the real ordered phases $\mathcal K =\{ \kappa_1< \cdots <\kappa_n\}$ uniquely identify a real regular multi--line soliton solution of the KP--II equation (see \cite{CK, KW2014});
\item The solution to the direct spectral problem provides a rational curve $\Gamma_0$, with marked points $\kappa_1,\dots,\kappa_n$, $P_0$, a spectral coordinate $\zeta$ such that $\zeta^{-1} (P_0)=0$, and a degree $k$ non--special divisor $\DS =\{P_1,\dots,P_k\}$ such that $\zeta(P_j) \in [\kappa_1,\kappa_n]$, $\forall j\in [k]$ \cite{Mal};
\item In a natural normalization, the KP wave function $ \psi (P, \vec x)$, where $P\in \Gamma\backslash \{P_0\}$ and $\vec x$ is a finite set of KP times, has the following property: the vector 
$({ \psi } (\kappa_1, \vec x), \dots, { \psi } (\kappa_n, \vec x) )$ defines un untrivial flow in the plane orthogonal to $[A]$ as the KP times $\vec x$ evolve (Lemma \ref{lem:psi_kappa}).
\end{itemize}
As pointed out in \cite{AG1} the mismatch between the dimension of the soliton variety and that of the divisor data on $\Gamma_0$ does not allow to reconstruct the soliton solution from the divisor $\DS$ (inverse spectral problem). Then,
to complete the KP divisor, we follow the approach in \cite{AG3, AG5} based on the degenerate finite gap theory on reducible curves introduced in \cite{Kr3}: we prove that it is possible to fix an initial time $\vec x_0$ and extend the spectral data from $(\Gamma_0, P_0, \DS(\vec x_0))$ to $(\Gamma, P_0, \DKP(\vec x_0))$, with $\Gamma$ a reducible $\mathtt M$--curve, in such a way that:
\begin{itemize}
\item The divisor $\DKP(\vec x_0)$ is non--special, and its restriction to $\Gamma_0$ is $\DS(\vec x_0)$;
\item $\DKP$ satisfies Dubrovin--Natanzon reality and regularity conditions \cite{ DN}: there is exactly one divisor point in each oval except in the one containing $P_0$.
\end{itemize}

The key problem is to choose the spectral curve $\Gamma$ so that the wave function takes equal values at pairs of double points for all KP times. Lemma \ref{lem:psi_kappa} implies that Kasteleyn system of relations gives the answer to such question provided that
\begin{itemize}
\item The spectral curve $\Gamma$ has dual graph $\mathcal G$ and we identify $\Gamma_0$ with the boundary of the disk;
\item At the boundary vertices $b_j$, we choose boundary conditions $v^{(k)}_{b_j}\equiv  \psi (\kappa_j, \vec x)$ in Kasteleyn system of relations, and then we assign the solution $v^{(k)}_{b} =v^{(k)}_{b} (\vec x) $ at the internal black vertex $b$ to the wave function $ \psi$ at the corresponding double points of $\Gamma$ at time $\vec x$.
\end{itemize} 

The normalized wave function is then extended to the whole $\Gamma \backslash \{P_0\}$ meromorphically in the spectral parameter. The reality and regularity property of the divisor on $\Gamma$ follows from the total non--negativity of the soliton data (Theorem \ref{theo:pos_div}). 
Finally the KP divisor on $\Gamma$ constructed using Kasteleyn system of relations is equivalent to that in \cite{AG5}. In that paper the geometric system of relation was used on a more general class of graphs. However the present construction is technically simpler, and automatically ensures the independence of the KP divisor on the gauge freedoms of the graph since it is done on undirected networks.

\smallskip

\textbf{Acknowledgements} I thank Pavlo Pylyavskyy for stimulating my interest in Kasteleyn theorem, and Petr Grinevich for several useful discussions. I also thank the referee of the paper for valuable remarks.

\section{Totally non-negative Grassmannians and almost perfect matchings on planar bipartite graphs in the disk }\label{sec:GTNN}

\begin{figure}
\centering{\includegraphics[width=0.32\textwidth]{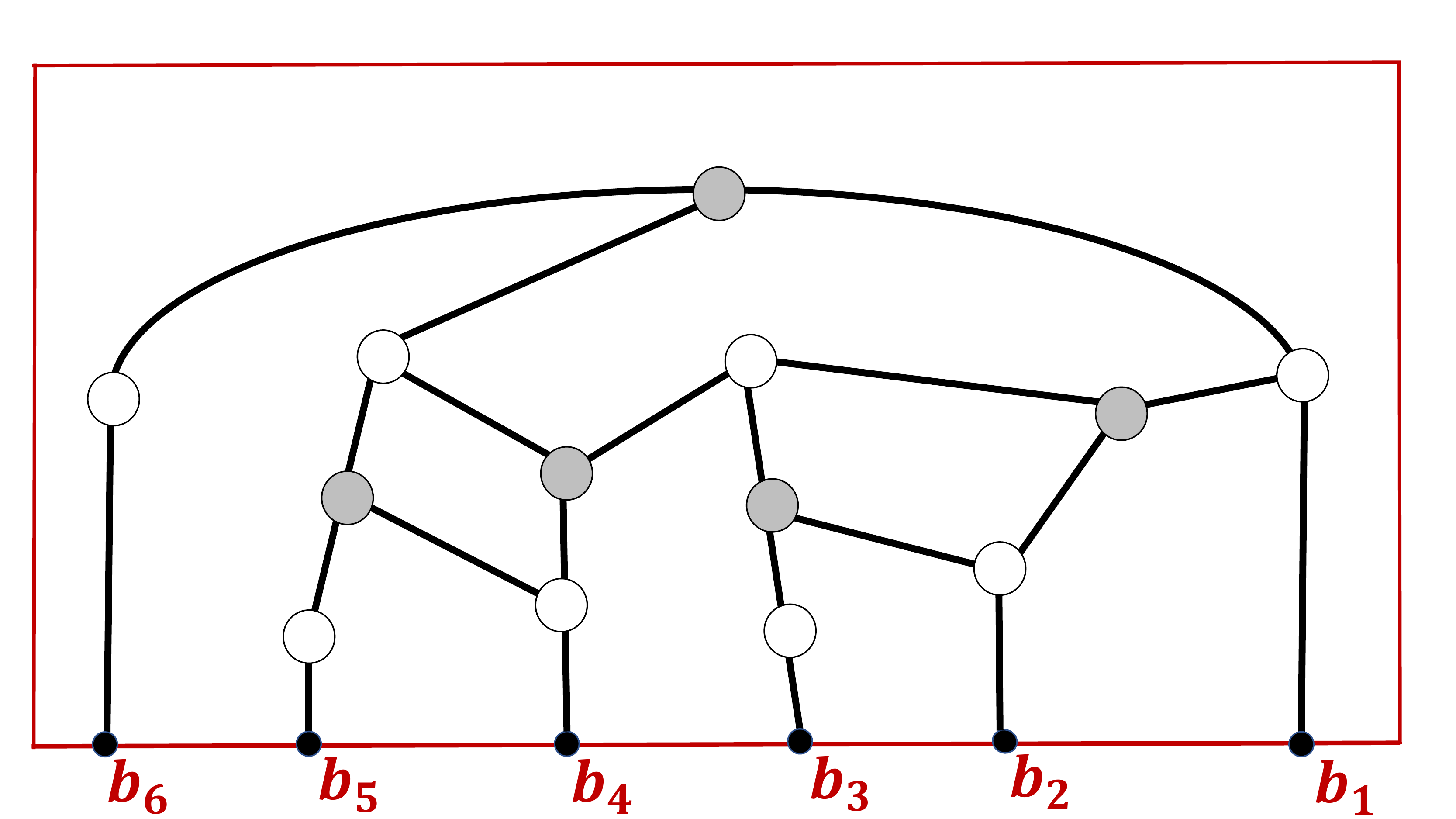}
\centering\includegraphics[width=0.32\textwidth]{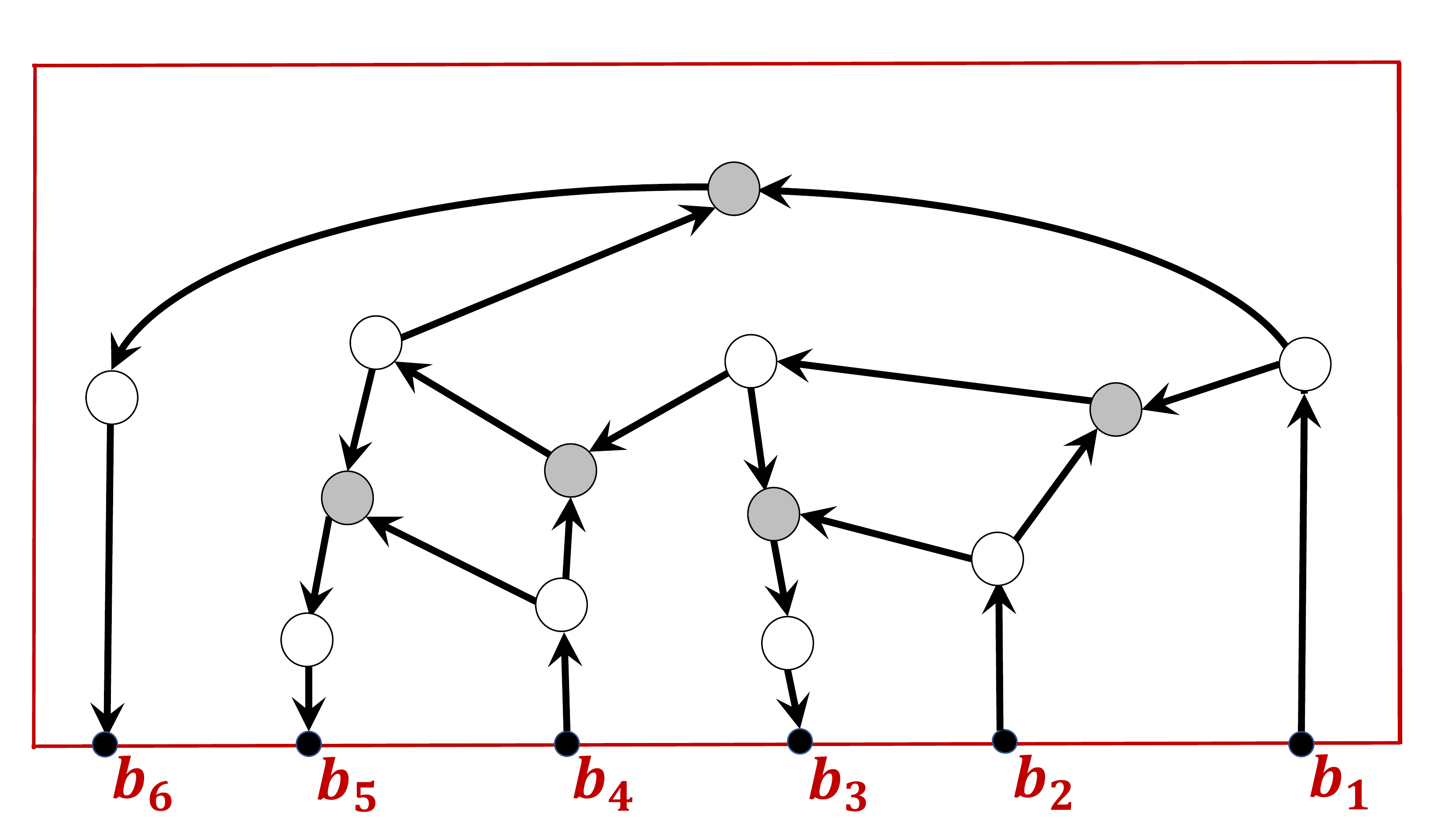}
\centering\includegraphics[width=0.32\textwidth]{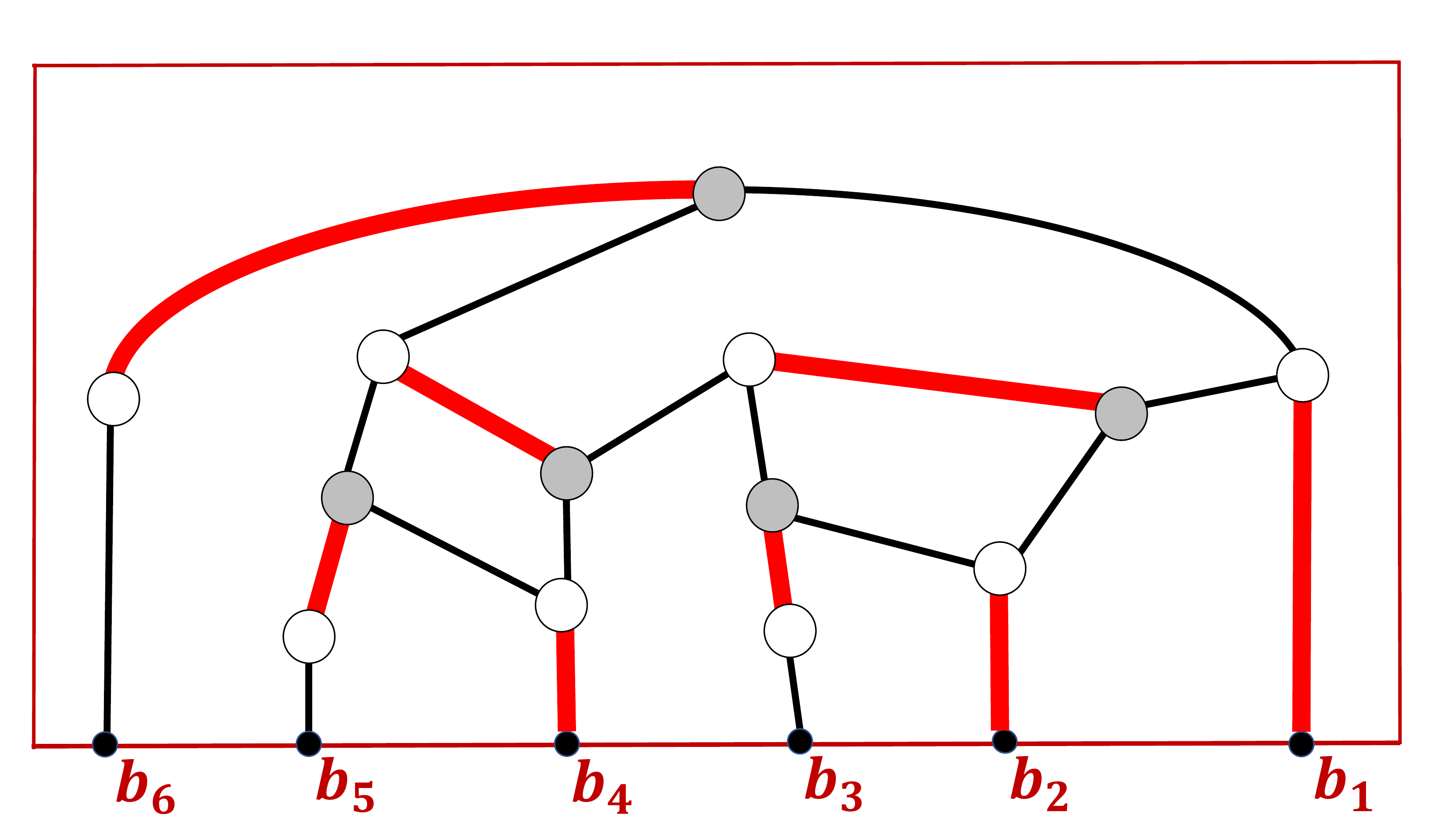}
\vspace{-.3 truecm}
\caption{\small{\sl A planar bipartite graph in the disk [left], a perfect orientation on it [center] and an almost perfect matching [right]. The boundary of the disk is colored red in all Figures.}}\label{fig:graph}}
\end{figure}

In this Section we briefly review the properties of totally non--negative Grassmannians necessary in the following Sections and, in particular, their relation to almost perfect matchings on planar bipartite graphs in the disk. We mainly follow \cite{Pos, Lam2} and \cite{PSW}. We remark that totally non negative Grassmannians are a special case of the generalization of classical positivity to generalized partial flag varieties by Lusztig \cite{Lus}.

\begin{definition}\textbf{Totally non-negative Grassmannian \cite{Pos}.}
Let $Mat^{\mbox{\tiny TNN}}_{k,n}$ denote the set of real $k\times n$ matrices of maximal rank $k$ with non--negative maximal minors $\Delta_I (A)$. Let $GL_k^+$ be the group of $k\times k$ matrices with positive determinants. Then the totally non-negative Grassmannian $\GTNN$ is 
\[
\GTNN = GL_k^+ \backslash Mat^{\mbox{\tiny TNN}}_{k,n}.
\]
\end{definition}

In the theory of totally non-negative Grassmannians an important role is played by the positroid stratification. Each cell in this stratification is defined as the intersection of a Gelfand-Serganova stratum \cite{GS,GGMS} with the totally non-negative part of the Grassmannian. More precisely:
\begin{definition}\textbf{Positroid stratification \cite{Pos}.} Let $\mathcal M$ be a matroid i.e. a collection of $k$-element ordered subsets $I$ in $[n]$, satisfying the exchange axiom (see, for example \cite{GS,GGMS}). Then the positroid cell $\S$ is defined as
$$
\S=\{[A]\in \GTNN\ | \ \Delta_{I}(A) >0 \ \mbox{if}\ I\in{\mathcal M} \ \mbox{and} \  \Delta_{I}(A) = 0 \ \mbox{if} \ I\not\in{\mathcal M}  \}.
$$
If $\S$ is not empty, $\mathcal M$ is called a positroid (totally non--negative matroid).
The positroid variety of $\mathcal M$ is
$$
\Pi (\mathcal M) =\{[A]\in Gr(k,n)\ |  \  \Delta_{I}(A) = 0 \ \mbox{if} \ I\not\in{\mathcal M}  \}. 
$$
\end{definition}

The definition means that a positroid $\mathcal M$ is a realizable matroid represented by $k\times n$ real matrices with positive maximal minors when the ordered column set $I=\{ 1\le i_1 < \cdots < i_k \le n\} \in\mathcal M$, and zero maximal minors otherwise. 

Every positroid cell is homeomorphic to an open ball of finite dimension \cite{Pos}. 
The combinatorial classification of all non-empty positroid cells was obtained in \cite{Pos}, where in particular the equivalence classes of planar bicolored (plabic) graphs representing positroid cells were classified, and an explicit and relevant minimal parametrization of each positroid cell $\S$ was obtained using reduced plabic networks in the disk.

The class of graphs considered throughout this paper are reduced planar bipartite graphs in the disk on which there exists at least one almost perfect matching; such condition is equivalent to the requirement that the graph possesses a perfect orientation \cite{PSW}. We recall that a \textbf{perfect orientation} $\mathcal O$ of a bicolored graph $\mathcal G$ is a choice of directions of its edges such that each black internal vertex $b$ is incident to exactly one edge directed away from $b$, and each white internal vertex $w$ is incident to exactly one edge directed towards $w$. A graph is perfectly orientable if it possesses a perfect orientation.

\begin{definition}\textbf{Planar bipartite graph in the disk}\label{def:graph}
We call planar bipartite graph in the disk an undirected planar graph drawn inside a disk, $\mathcal G =(\mathcal V,\mathcal E)$, with finite vertex set $\mathcal V$ and finite edge set $\mathcal E$, such that:
\begin{enumerate}
\item Internal vertices are strictly inside the disk and are either black or white;
\item There are $n>0$ boundary vertices on the boundary of the disk labeled $b_1,\dots, b_n$, in clockwise order. Boundary vertices share the same color and have degree one;
\item Each edge in $\mathcal G$ joins two vertices of different color;
\item $\mathcal G$ is perfectly orientable.
\end{enumerate}
In the following we denote $\mathcal B$ and $\mathcal W$ respectively the set of black and of white vertices.
\end{definition}
In Figure \ref{fig:graph} we show a planar bipartite graph in the disk. 

Let $(\mathcal G, \mathcal O)$ denote the directed graph $\mathcal G$ with a perfect orientation $\mathcal O$. Then the \textbf{source set} $I=I(\mathcal O)\subset [n]$ is the set of $i$ such that $b_i$ is a boundary source of $(\mathcal G, \mathcal O)$, similarly $b_j$ is a boundary sink for all $j\in \bar I$.

\begin{proposition}\cite{Pos}
Let $\mathcal G$ be as in Definition \ref{def:graph}. Then all of its perfect orientations have source sets of equal cardinality and $\mathcal G$ is called of type $(k,n)$ if source sets have cardinality $k$.

Moreover, let $\mathcal M_{\mathcal G}$ be the collection of the $k$--subsets $I(\mathcal O)$ of its perfect orientations:
\begin{equation}\label{def:matroid_orient}
\mathcal M_{\mathcal G} \, = \, \{ I(\mathcal O) \, | \,\mathcal O \mbox{ is a perfect orientation of } \mathcal G \, \}.
\end{equation} 
Then $\mathcal M_{\mathcal G}$ is a positroid.
\end{proposition}

The graph in Figure \ref{fig:graph} represents the positroid cell $\S\subset Gr^{\mbox{	\tiny TNN}} (3,6)$ such that $I \not \in {\mathcal M}$ if and only if $I=\{ 1,2,3\}$. Indeed it is easy to check that there is no perfect orientation only for the source set $b_1,b_2,b_3$.

\begin{figure}
\centering{\includegraphics[width=0.33\textwidth]{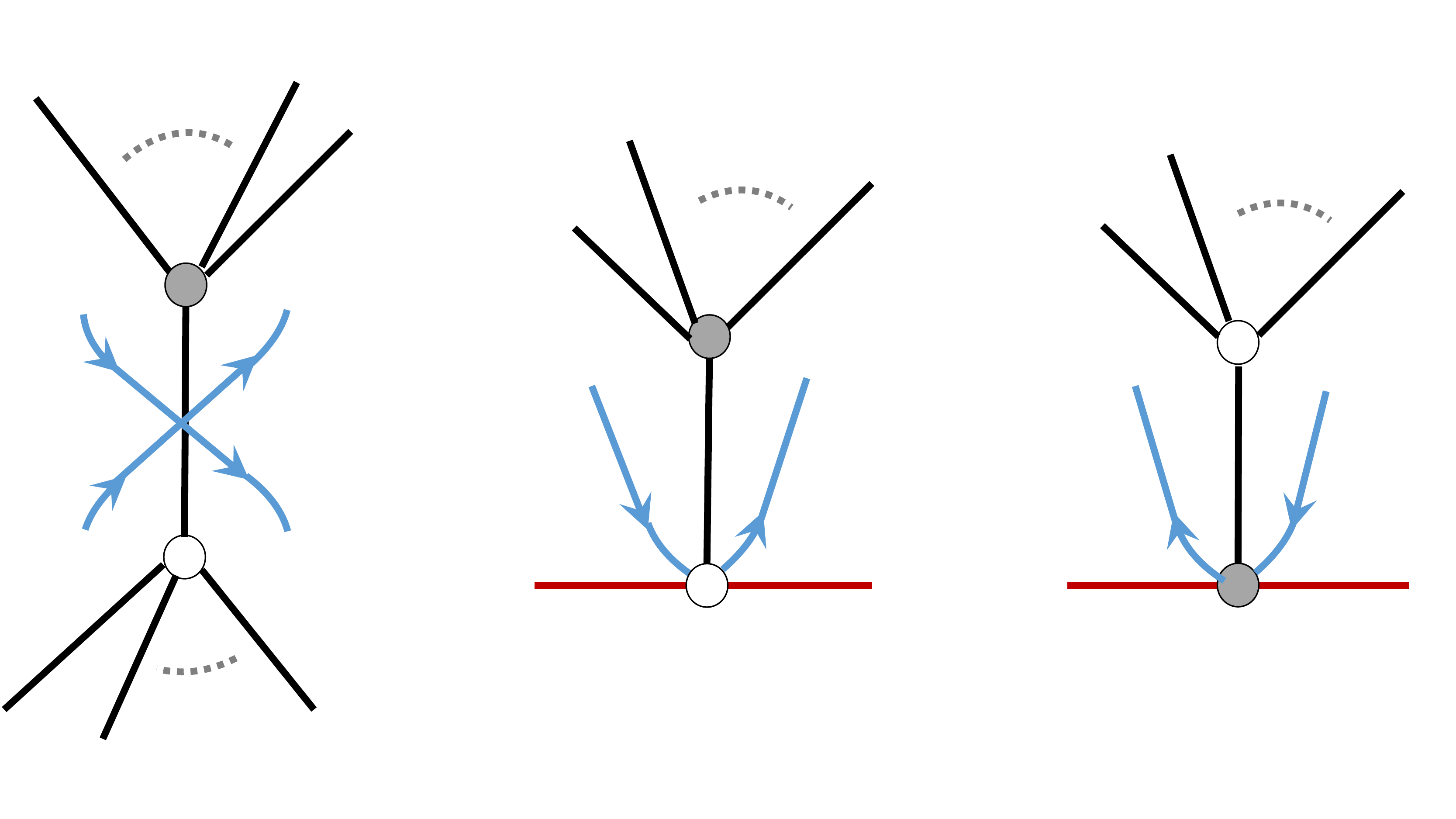}
\vspace{-.5 truecm}
\caption{\small{\sl Strands at internal and at boundary edges.}}\label{fig:strand_rule}}
\end{figure}

\smallskip

\begin{definition}\textbf{Reduced planar bipartite graph in the disk}\label{def:graph_red}
A graph $\mathcal G$ as in Definition \ref{def:graph} is \textbf{reduced} if moreover
\begin{enumerate}
\item Every component of $\mathcal G$ contains at least one boundary vertex;
\item Every internal vertex of degree 1 is adjacent to a boundary vertex:
\item There is at most one edge sharing a pair of vertices $b$, $w$;
\item The number of faces of $\mathcal G$ is minimal among graphs with the same positroid.
\end{enumerate}
\end{definition}

\begin{figure}
  \centering{\includegraphics[width=0.33\textwidth]{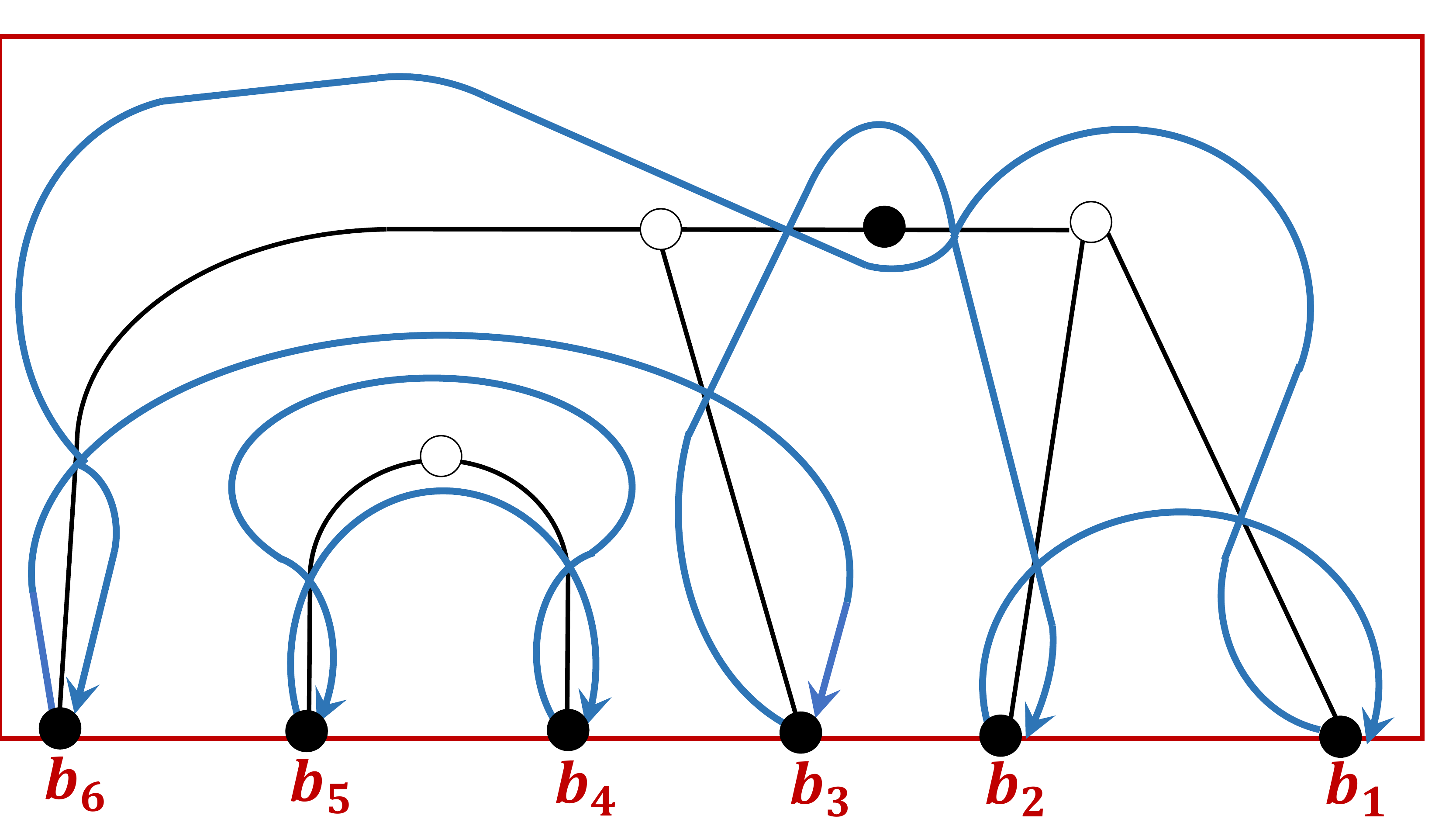}
	\includegraphics[width=0.4\textwidth]{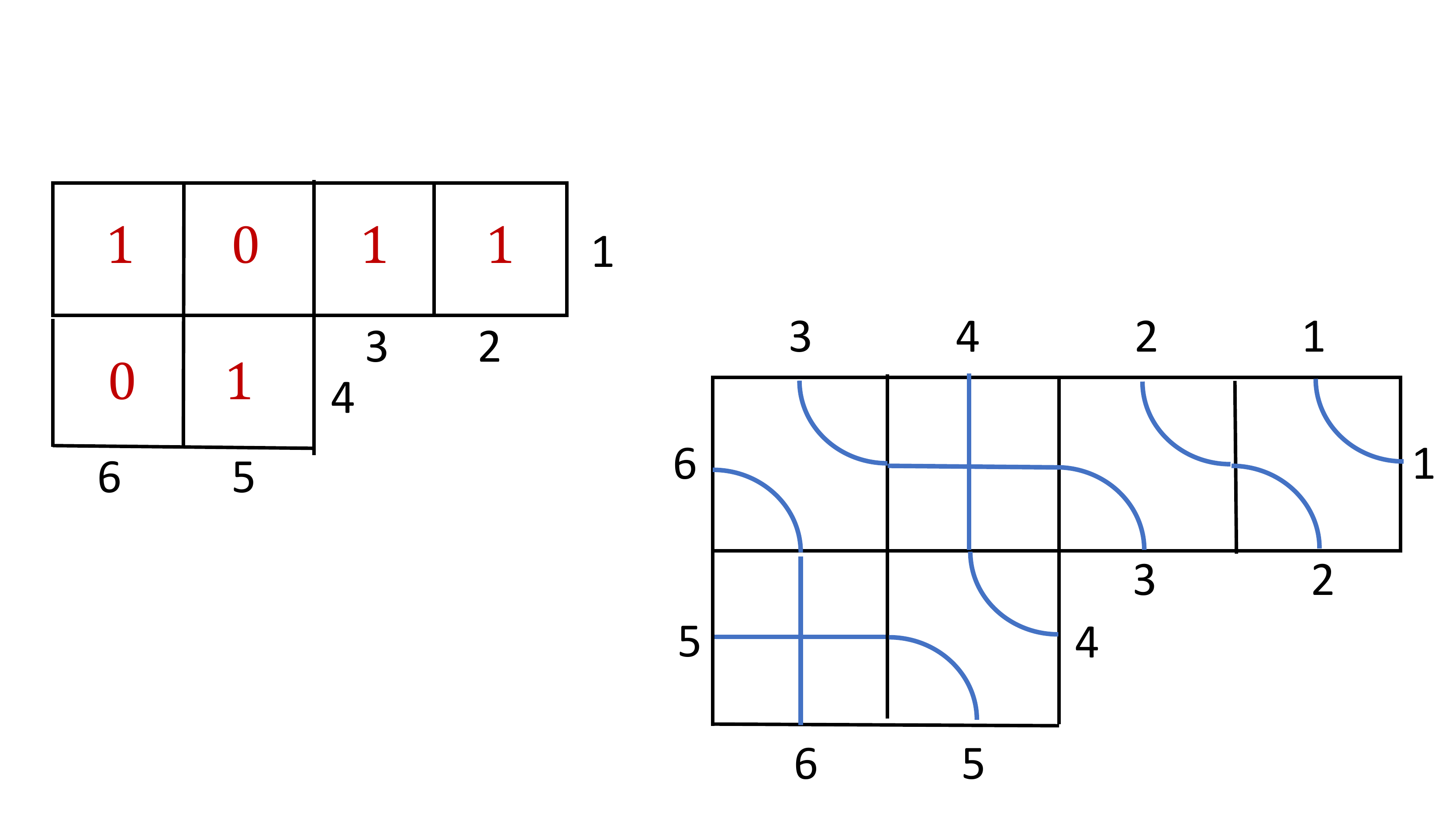}
	\vspace{-.3 truecm}
  \caption{\small{\sl The graph and the Le--diagram both represent the same positroid cell in $Gr^{\mbox{\tiny TNN}}(2,6)$ since they share the same value of the decorated permutation $\pi=(6,1,2,5,4,3)$.  }\label{fig:strand_Gr26}}}
\end{figure}

Positroids in $\GTNN$ are in bijection with the following objects \cite{Pos}: \textbf{decorated permutations} on $n$ letters with $k$ weak excedances and with \textbf{Le--diagrams} of type $(k,n)$. 

The decorated permutation $\pi=\pi(\mathcal M)$ may be computed using oriented strands on a reduced bipartite graph $\mathcal G= \mathcal G(\mathcal M)$ \cite{Pos}: each internal edge intersects transversely two strands at the midpoint, whereas at the boundary edges the strands terminate at the boundary vertex as shown in Figure \ref{fig:strand_rule}. Then $\pi(i)$ is the boundary destination of the strand starting at the boundary vertex $i$. In Figure \ref{fig:strand_Gr26} [left] the reduced graph corresponds to $\pi=(6,1,2,5,4,3)$. The excedances $\pi(1)$ and $\pi(4)$ imply that the lexicographically minimal base in the corresponding matroid $\mathcal M$ is $\{1, 4\}$, therefore $\mathcal G$ represents a positroid cell in $Gr^{\mbox{\tiny TNN}}(2,6)$. 

A Le--diagram $L$ of type $(k,n)$ is a Young diagram in the $k\times (n-k)$ rectangle together with a filling of $0,1$s such that there is no $0$ which has a 1 above it in the same column and a 1 to the left in the same row \cite{Pos}. The total number of $1$s is the dimension of the corresponding positroid cell $\S$ \cite{Pos}. If one labels the southeast border of the Le--diagram $L$ with the numbers $1,2,\dots,n$ starting from the northeast corner, then the labels of the vertical edges of the border form the lexicographically minimal base of the positroid $\mathcal M$ represented by $L$. In Figure \ref{fig:strand_Gr26} [center] we show an example: the Le-diagram represents a four-dimensional positroid cell in $Gr^{\mbox{\tiny TNN}} (2,6)$ with lexicograhically minimal base $\{1, 4\}$.

One may compute $\pi$ also using Le-diagrams \cite{KW2014}: first one replaces each 1 with an elbow and each 0 with a cross. Then one labels the northeast destination of each pipe with the same label of its southeast starting point. $\pi$ is then computed as follows: if $i$ labels a vertical edge on the southeast border then $\pi(i)$ is the label on the same row on the west border; if $i$ labels an horizontal edge on the southeast border then $\pi(i)$ is the label on the same column on the northern border (see Figure \ref{fig:strand_Gr26} [right]). 

\begin{figure}
  \centering{\includegraphics[width=0.37\textwidth]{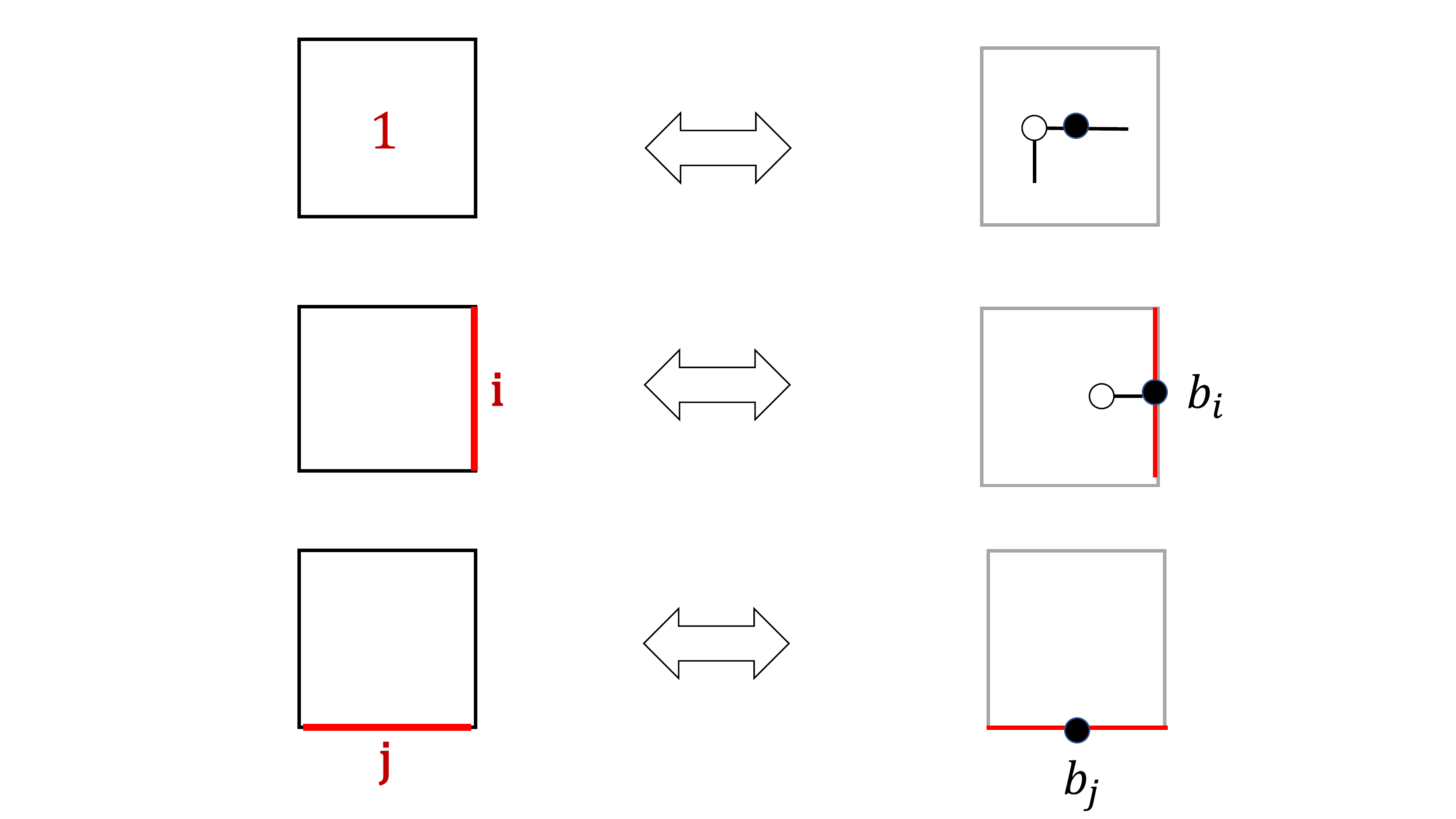}
	\includegraphics[width=0.45\textwidth]{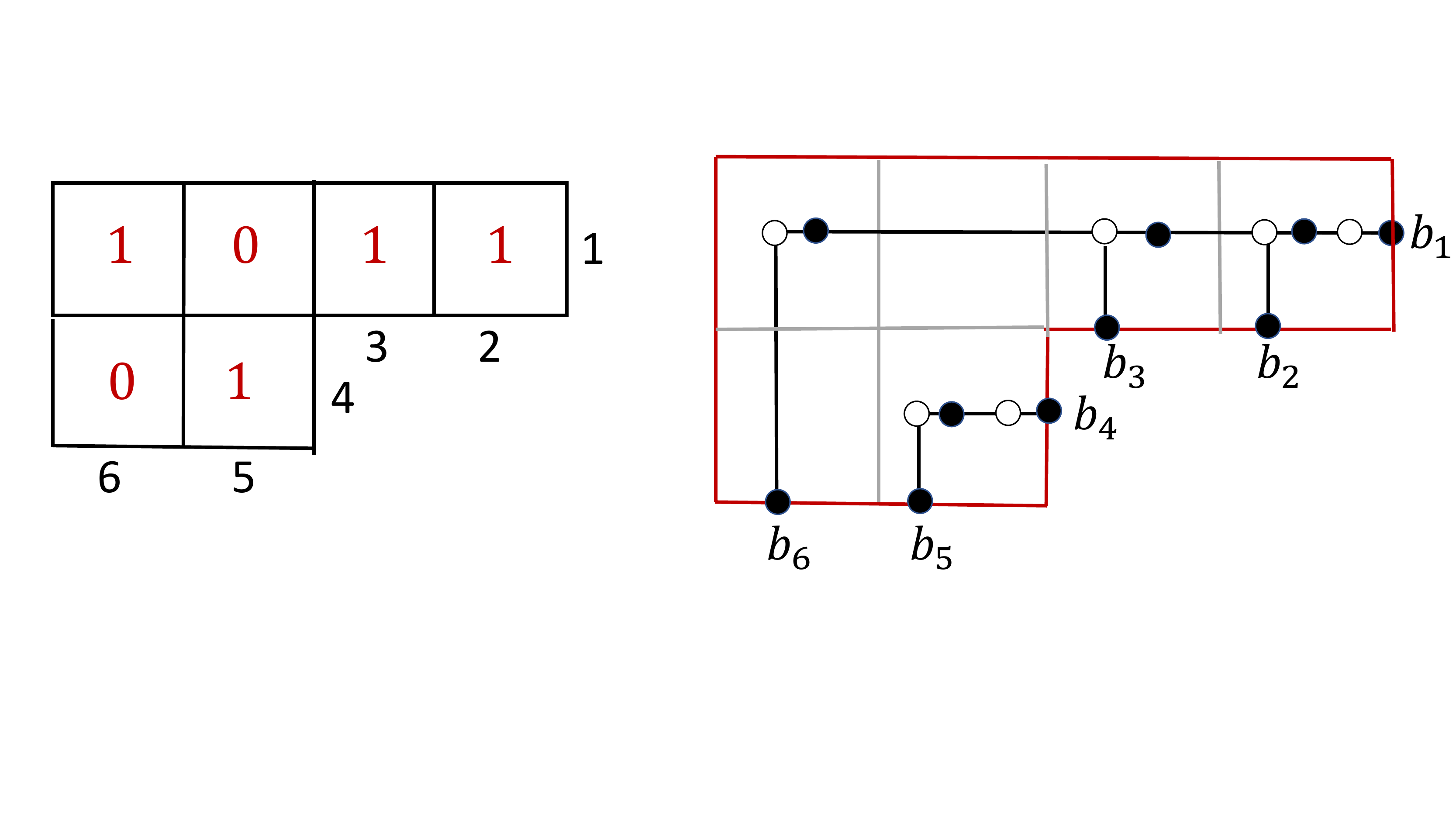}
	\vspace{-.3 truecm}
  \caption{\small{\sl Left: the rules to construct the Le--graph from the Le--diagram. Right: Le--diagram and Le--graph for the Example of Figure \ref{fig:strand_Gr26}.}\label{fig:Le_graph}}}
\end{figure}

\smallskip

\textbf{The Le--graph} Each positroid cell $\S$ is represented by at least one reduced graph \cite{Pos}. 
The bipartite Le-graph representing the positroid cell $\S$ is a reduced graph constructed directly from the Le-diagram $L=L(\S)$ as follows \cite{Pos}. It is obtained putting a black vertex in the middle of each segment of the southeast border of $L$; if the border segment is vertical, one also adds a white vertex next to it and a horizontal edge; finally, in the middle of each Le-box filled by 1  one inserts a hook with a white vertex on the left and a black vertex on the right (see Figure \ref{fig:Le_graph} [left]). Then, for each box filled by 1, the horizontal half--edge starting at the black vertex is prolonged to the nearest white vertex on the right, and the vertical half--edge starting at the white vertex is prolonged to its nearest black vertex downwards. The boundary of the Young diagram is the boundary of the disk. In Figure \ref{fig:Le_graph} [right] we construct the Le--graph for the positroid cell of Figure \ref{fig:strand_Gr26}.

\smallskip

\textbf{Irreducible positroid cells} An irreducible positroid cell $\S$ in $Gr^{\mbox{\tiny TNN}}(k,n)$ corresponds to a positroid $\mathcal M$ with the following additional property: for any $j\in [n]$, there exist $I, J\in \mathcal M$ such that $j\in I$ and $j\not\in J$.  Then the Le--diagram $L(\S)$ does not contain either rows or columns filled by $0$s, whereas $\pi(\S)$ is a derangement, {\sl i.e.} a permutation of $n$ letters with $k$ excedances and no fixed points. Bipartite graphs representing an irreducible cell $\S$ do not possess isolated boundary vertices. The reduced graphs representing $\S$ possess $g+1$ faces where $g$ is the dimension of $\S$.

\smallskip

Finally let us recall the natural duality transformation of positroids. 
\begin{definition}\textbf{Duality transformations between positroids and positroid cells}\label{def:dual_cell}
Given a positroid $\mathcal M$ of $k$--element subsets in $[n]$, its dual is the positroid $\overline{\mathcal M}$ of $(n-k)$-- element subsets in $[n]$ such that 
\begin{equation}\label{eq:dual_matroid}
I\in \mathcal M \quad\quad \iff \quad\quad \bar I \in \overline{\mathcal M}.
\end{equation}
If $\S\subset \GTNN$ is the positroid cell represented by $\mathcal M$, then we denote $\Sprime\subset Gr^{\mbox{\tiny TNN}}(n-k,n)$ its dual positroid cell. Similarly if $\Pi_{\mathcal M}\subset Gr(k,n)$ is the positroid variety represented by $\mathcal M$, then we denote $\Pi_{\overline{\mathcal M}} \subset Gr(n-k,n)$ the dual positroid variety represented by $\overline{\mathcal M}$.
\end{definition}

If $\pi$ is the derangement representing $\mathcal M$, then the derangement $\bar \pi$ representing $\overline{\mathcal M}$ is $\bar \pi = \pi^{-1}$ \cite{Pos}. 

\begin{figure}
  \centering{\includegraphics[width=0.33\textwidth]{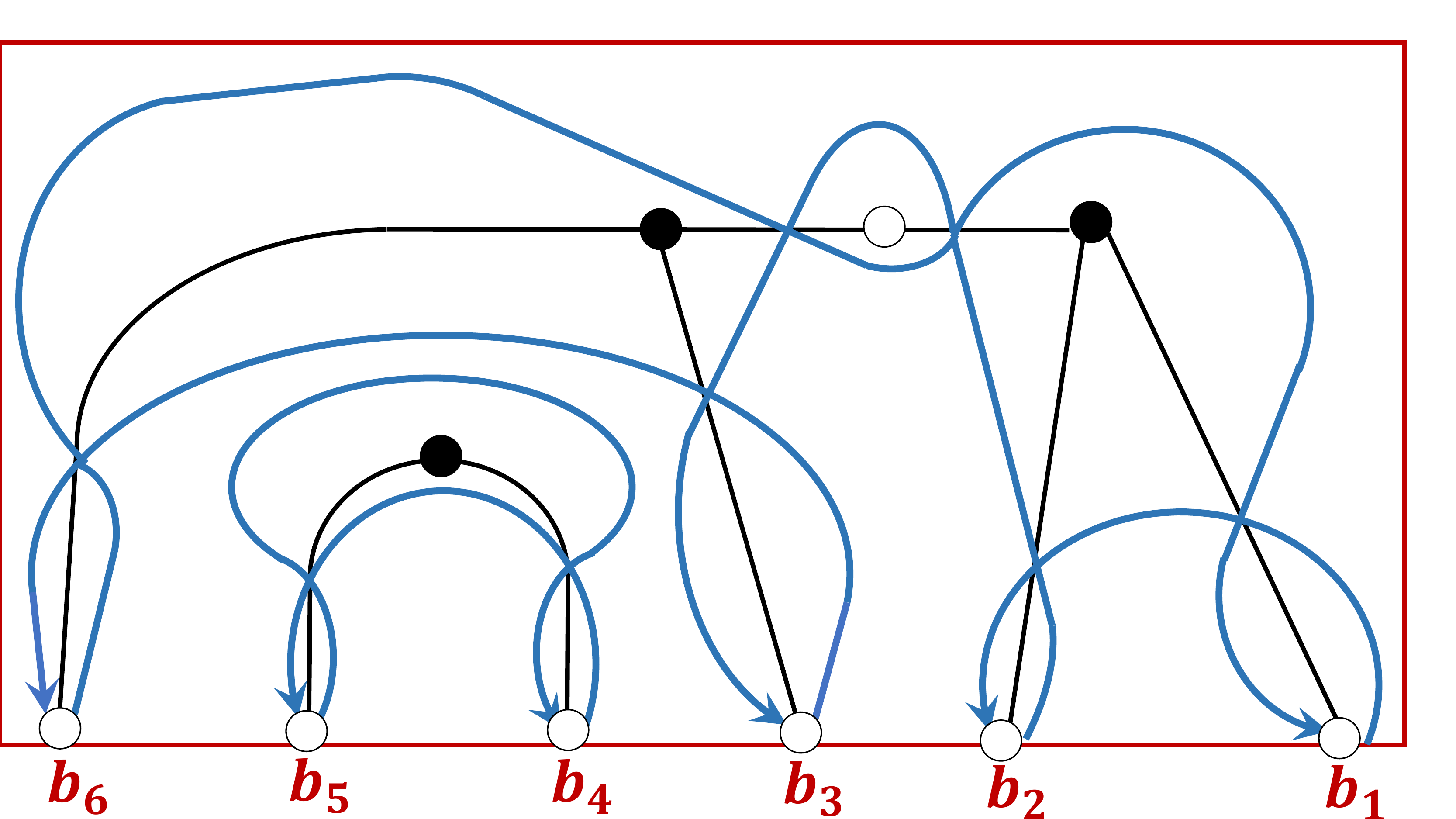}
	\includegraphics[width=0.35\textwidth]{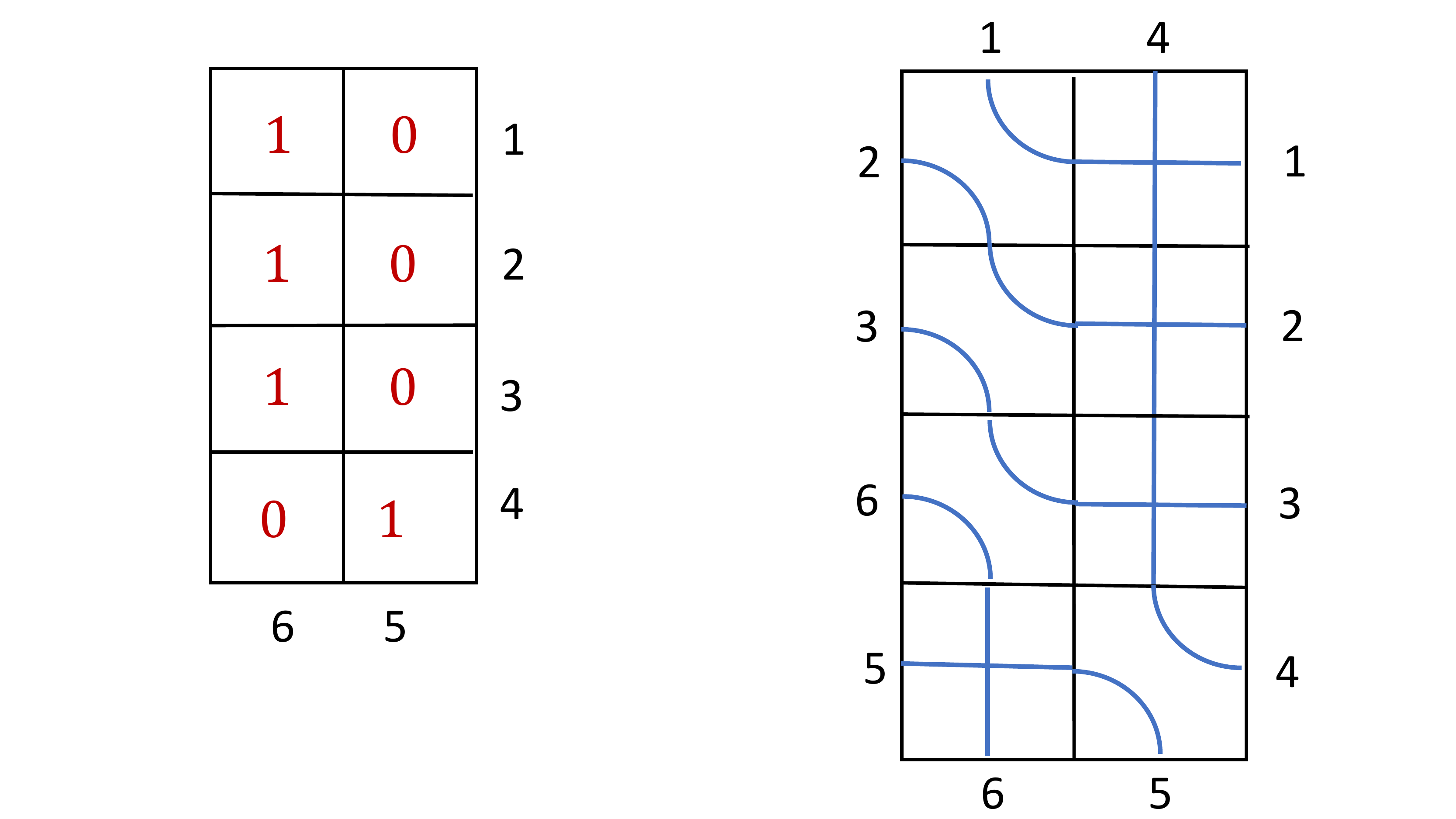}
	\vspace{-.3 truecm}
  \caption{\small{\sl The graph, the derangement and the Le-diagram of the positroid cell in $Gr^{\mbox{\tiny TNN}}(4,6)$ dual to that in Figure \ref{fig:strand_Gr26}.}\label{fig:strand_Gr26_dual}}}
\end{figure}
If $\mathcal G$ is a planar bipartite graph in the disk representing $\S$, then the graph $\overline{\mathcal G}$ obtained inverting the color of all vertices including those at the boundary represents $\Sprime$.

Since dual positroid cells have the same dimension, it is possible to introduce bijections between $\S$ and $\Sprime$ which allow to parametrize $\Sprime$ starting from a parametrization of $\S$. A natural bijection which preserves the total non--negativity property is associated to the transposition of Kasteleyn matrices in Proposition \ref{prop:dual}. A different duality relation between the positroid varieties $\Pi_{\mathcal M}$ and $\Pi_{\overline{\mathcal M}}$ is constructed in Theorem \ref{theo:sol_kas_sys_1} solving Kasteleyn system of relations at the boundary vertices when the vector space is $\mathbb{C}^{n-k}$.

In Figure \ref{fig:strand_Gr26_dual} we show the effect of the duality transformation of positroids for the cell of Figure \ref{fig:strand_Gr26}. The dual cell is four--dimensional in $Gr^{\mbox{\tiny TNN}}(4,6)$ with derangement $\bar{\pi} = \pi^{-1}=\left( 2,3,6,5,4,1 \right) $.

\smallskip

\begin{definition}\textbf{Almost perfect matchings on $\mathcal G$}\label{def:matching}
An almost perfect matching of $\mathcal G=(\mathcal V, \mathcal E)$ is a collection $M$ of edges of $\mathcal G$ that contains exactly once each internal vertex of $\mathcal G$ and each boundary vertex at most once. For an almost perfect matching $M$ its boundary $\partial M$ is defined as follows
\[
\partial M = \{ i \in [n] \, : \, \mbox{ the black boundary vertex } b_i \in M \} \; \cup \; \{ i \in [n] \, : \, \mbox{ the white boundary vertex } b_i \not\in M \}.
\]
\end{definition}

In Figure \ref{fig:graph} [right] we show an almost perfect matching.

If the boundary vertices are colored black, $|\mathcal B|-|\mathcal W|=n-k$ and each almost perfect matching uses $k$ boundary vertices. If the boundary vertices are colored white, $|\mathcal W|-|\mathcal B|=k$ and the set of boundary vertices used in each matching is $n-k$. If we release the requirement that boundary vertices share the same color, the boundary of each matching of $\mathcal G$ has size
\[
k = \# (\mbox{white vertices}) - \#(\mbox{black vertices}) + \# (\mbox{black boundary vertices}).
\]

In \cite{PSW} perfect orientations of $\mathcal G$ are shown to be in bijection with almost perfect matchings in $\mathcal G$. Therefore the following statement holds true.

\begin{proposition}\cite{PSW}
Let $\mathcal G$ be as in Definition \ref{def:graph} and let $\mathcal M_{\mathcal G}$ be the positroid of its perfect orientations. Then $I\in \mathcal M_{\mathcal G}$ if and only if there exists an almost perfect matching $M$ in $\mathcal G$ with $\partial M= I$.
\end{proposition}

\smallskip

In this paper we are interested in networks of graph $\mathcal G$ with real positive weights assigned either to the edges of $\mathcal G$ or to its faces. In \cite{Pos} a natural minimal parametrization of each given positroid cell $\S\subset \GTNN$ is obtained in terms of face weights on reduced graphs representing $\S$. 

\textbf{Notation for edges on undirected and on directed graphs} If the graph is undirected, the edge $e$ connecting the black vertex $b$ and the white vertex $w$  will be denoted $e=\overline{bw}$. If the graph is directed, the edge $e$ starting at the vertex $u$ and ending at the vertex $v$, will be denoted $e=\overrightarrow{uv}$.

\textbf{Terminology for faces}
A face $\Omega$ is internal if its boundary has empty intersection with the boundary of the disk, otherwise it is an external face. There is a unique external face including the boundary segment from $b_n$ to $b_1$ clockwise and we call it the infinite face. All other faces are called finite. 

If $\mathcal G$ possesses $g+1$ faces, we label the finite faces $\Omega_i$, $i\in [g]$, and denote $\Omega_0$ the infinite face.
The same labeling rule applies to face weights.

\textbf{Networks}
A network $\mathcal N = (\mathcal G, f)$ is a graph $\mathcal G$ as in Definition \ref{def:graph} with $g+1$ faces and a choice of non zero face weights on the finite faces $f(\Omega_i)=f_i \not =0$, $i\in [g]$. The weight of the infinite face $\Omega_0$ is then $f(\Omega_0)= (\prod_{i\in [g]} f_i)^{-1}$ (see Figure \ref{fig:graph_2}[left] for an example).

There is a natural way to pass from the face weights to the edge weights on undirected or directed networks.

\begin{figure}
\centering{\includegraphics[width=0.32\textwidth]{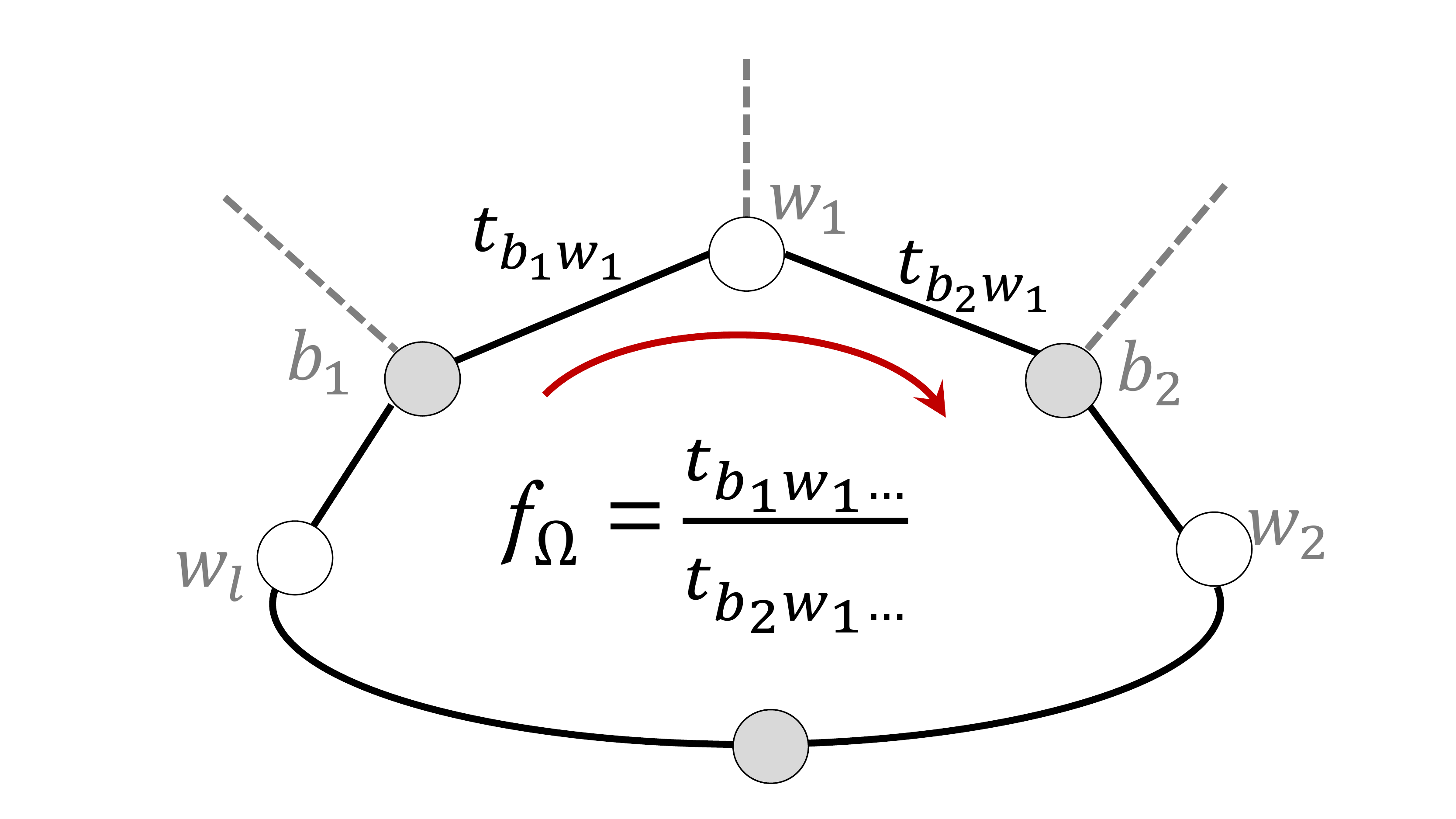}
\hfill
\includegraphics[width=0.32\textwidth]{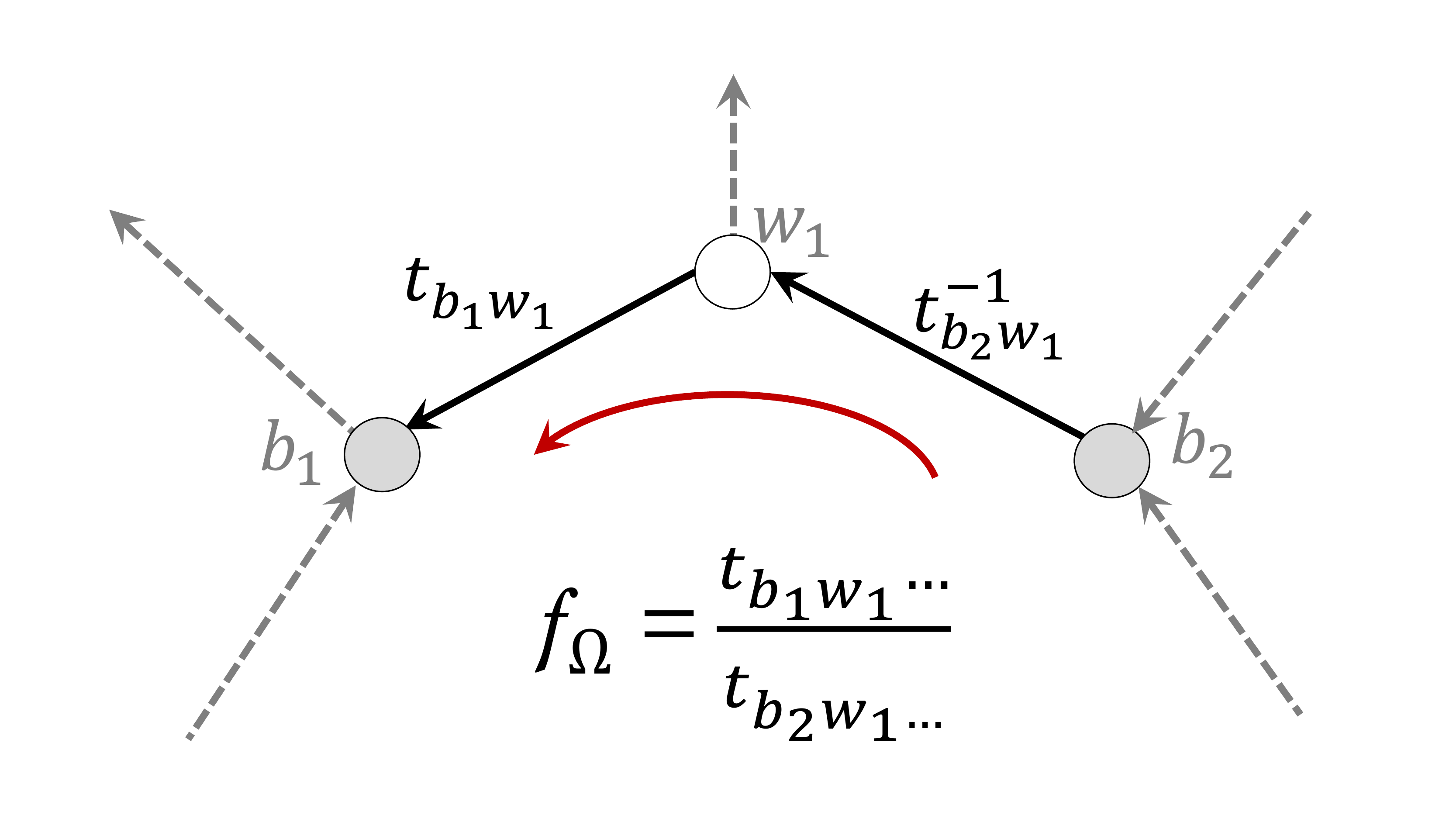}
\hfill
\includegraphics[width=0.32\textwidth]{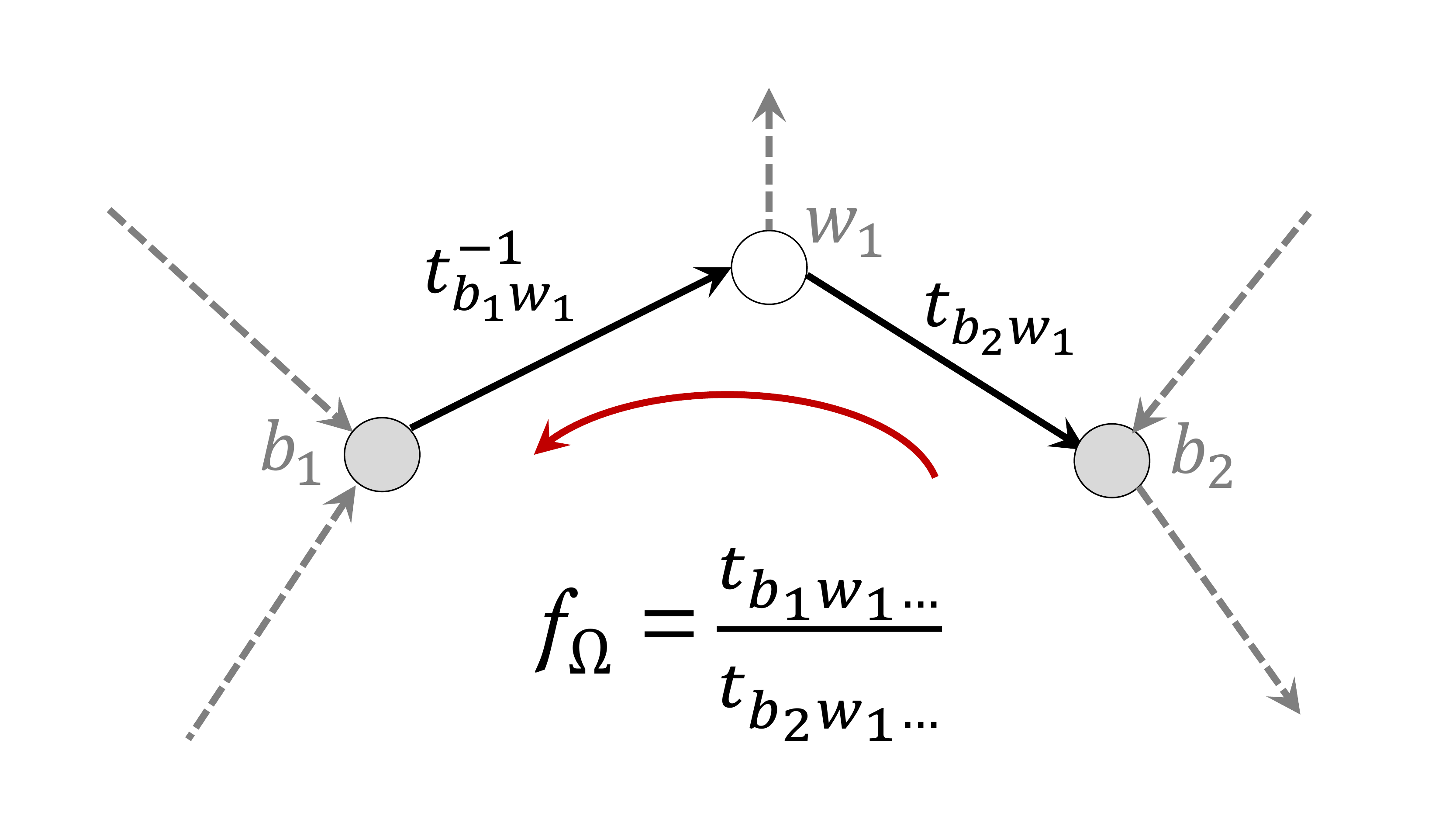}
\vspace{-.3 truecm}
\caption{\small{\sl The rule of transformation between face and edge weights for undirected graphs [left] and for directed ones [center and right].}}\label{fig:weight_rule}}
\end{figure}

\textbf{The rule for assigning weights on undirected and on directed networks} If $\mathcal G$ is undirected, let us label in clockwise order the vertices bounding a given face $\Omega$, $b_1,w_1,b_2,\dots,b_l,w_l$. Then the relation between the face weight $f_{\Omega}$ and the edge weights $t_{b_iw_j}$ is
\begin{equation}\label{eq:face_weights}
f_{\Omega} = \frac{\prod_{i=1}^l t_{b_iw_i} }{\prod_{i=1}^l t_{b_iw_{i-1}}},
\end{equation}
with obvious modifications if $\Omega$ is an external face (see Figure \ref{fig:weight_rule} [center]).
If $\mathcal G$ is directed, the face weight is obtained multiplying the edge weights for the edges bounding $\Omega$ and directed anticlockwise, and dividing the edge weights for the edges bounding $\Omega$ and directed clockwise (see Figure \ref{fig:weight_rule} [right]). 
These two rules may be easily combined in an explicit transformation between edge weights for undirected and directed graphs with equal face weights:
\begin{equation}\label{eq:edge_weights}
t_{\overrightarrow{uv}} = \left\{ \begin{array}{ll}
t_{bw}, & \mbox{ if } u=w, \; v=b;\\
\noalign{\medskip}
t_{bw}^{-1},& \mbox{ if } u=b, \; v=w.
\end{array}
\right.
\end{equation}
Finally if the directed edge $e=\overrightarrow{uv}$ has weight $t_e$, then the directed edge $e^{\prime} =\overrightarrow{vu}$ has weight $t_{e^{\prime}} = t_e^{-1}$. We illustrate these rules in Figure \ref{fig:graph_2}, where we only write edge weights different from 1.

\begin{figure}
\centering{\includegraphics[width=0.32\textwidth]{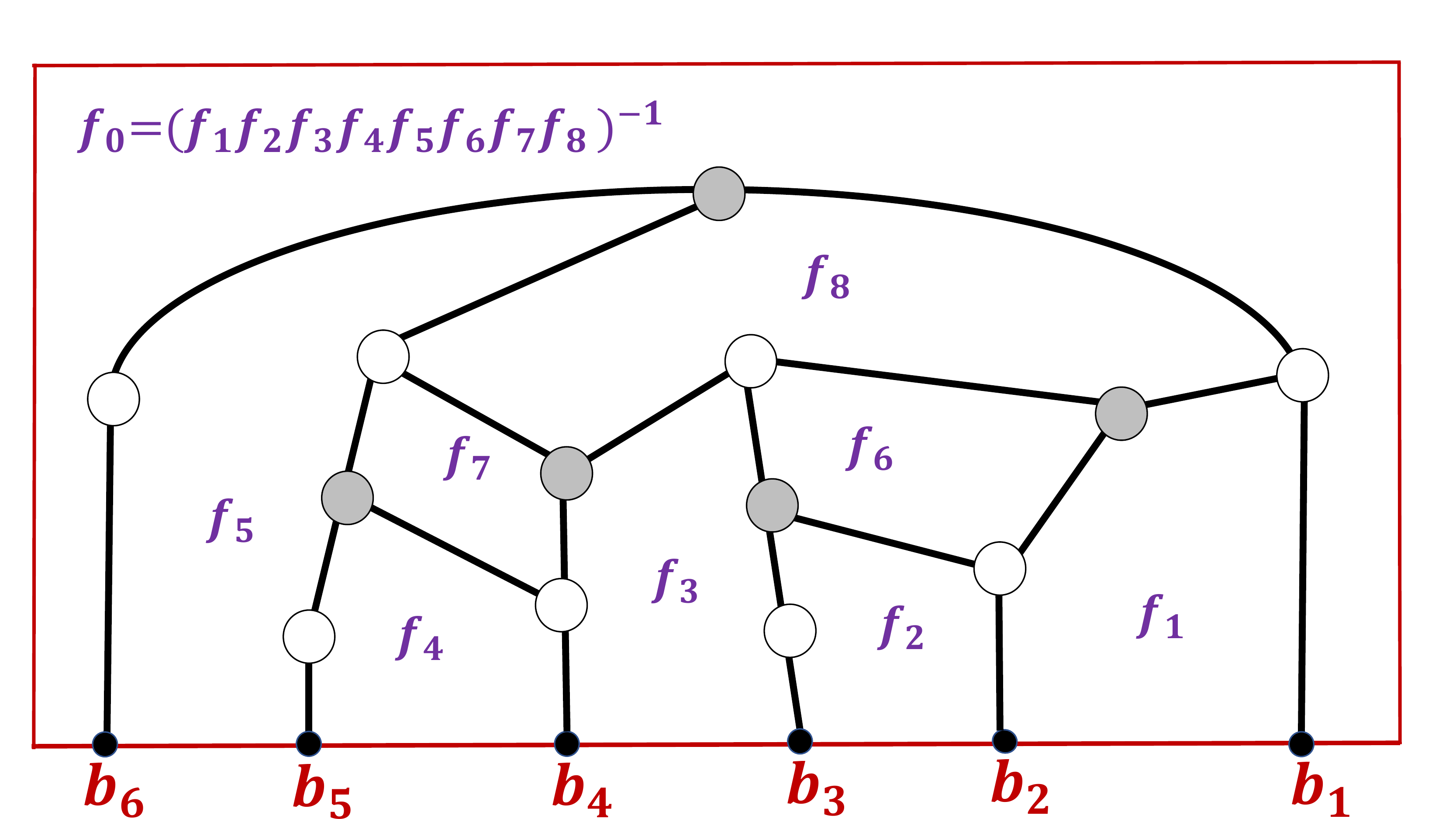}
\includegraphics[width=0.32\textwidth]{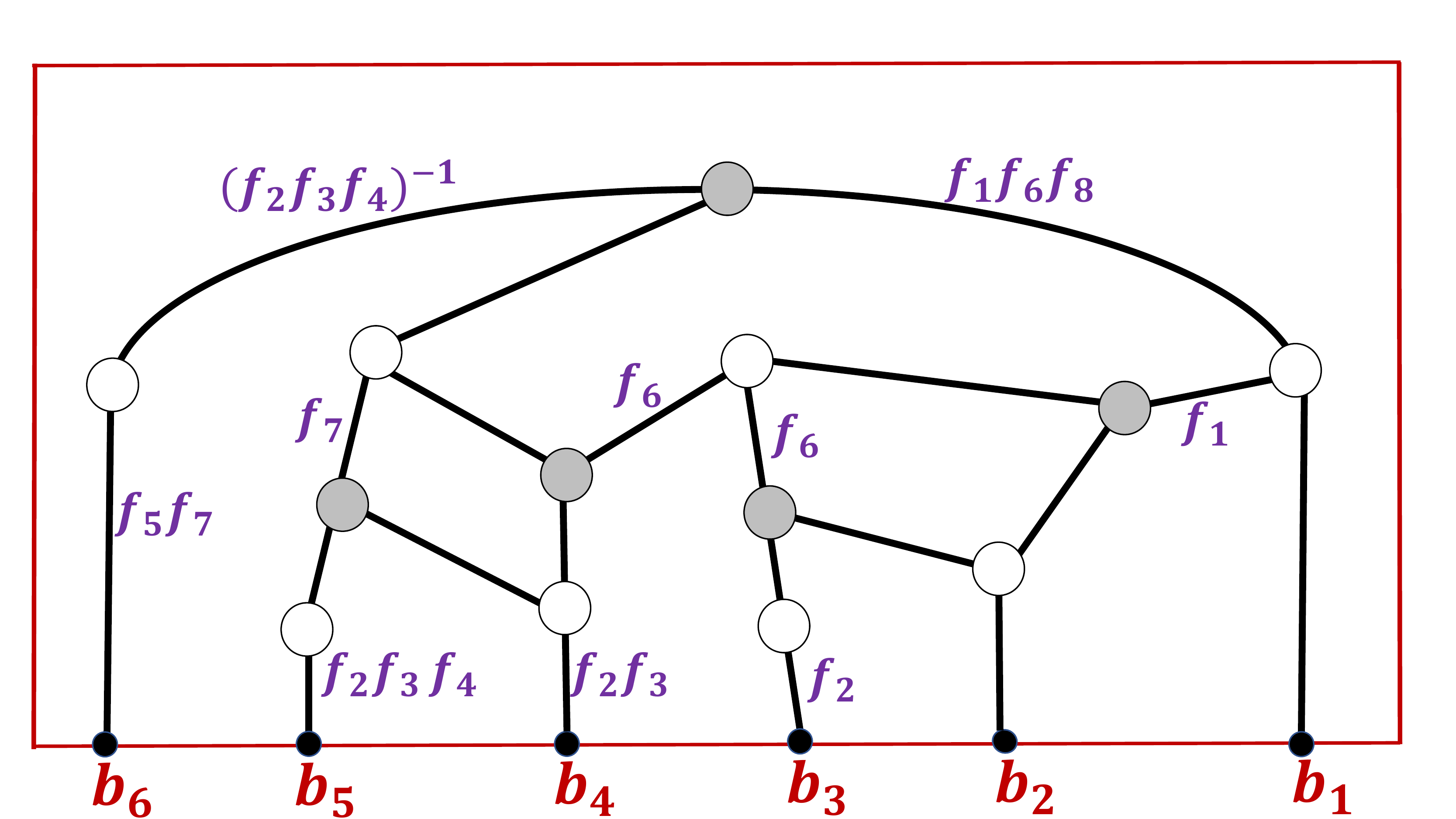}
\includegraphics[width=0.32\textwidth]{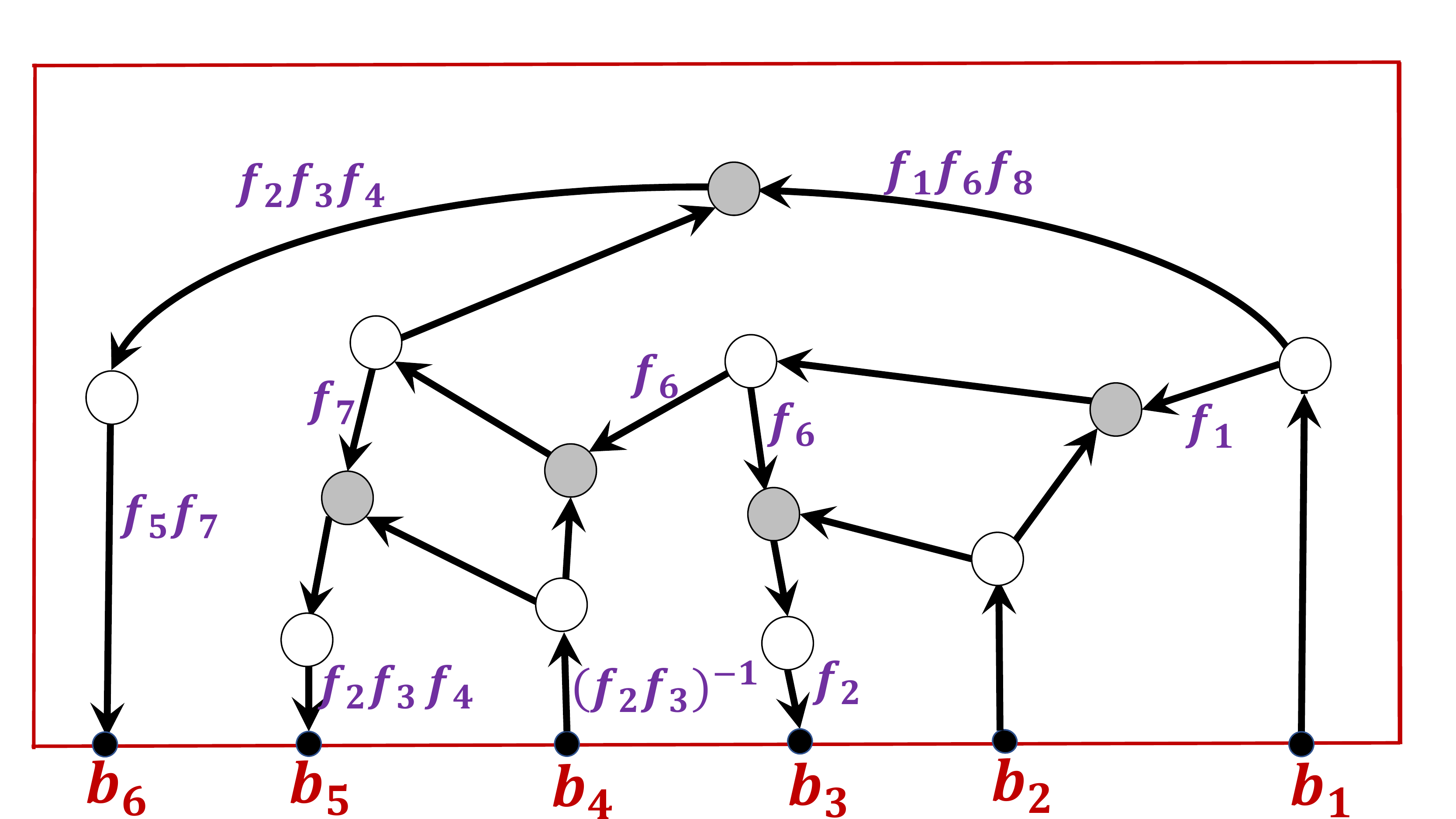}
\vspace{-.3 truecm}
\caption{\small{\sl A planar bipartite network in the disk with face weights [left], an equivalent network sharing the same undirected graph with edge weights satisfying (\ref{eq:face_weights}) [center] and an equivalent directed network satisfying (\ref{eq:edge_weights}) [right]. Unmarked edges carry unit weights.}}\label{fig:graph_2}}
\end{figure}

If the face weights are all real and positive (complex non--zero), there is more than one way to assign real positive (complex non--zero) edge weights to the graph following the above rules. We illustrate the weight gauge transformation only in the case of positive weights; in the case of complex non--zero weights it is sufficient to replace everywhere $c>0$ with $c$ complex non-zero in the formulas (\ref{eq:weight_gauge_und}) and (\ref{eq:weight_gauge_dir}).

\textbf{The weight gauge for undirected and for directed networks} If the graph is undirected and $t_e>0$, $e\in \mathcal E$, is a solution to the system (\ref{eq:face_weights}) at the faces of $\mathcal G$, then the following gauge transformation at an internal vertex $v$ gives another solution for any given $c>0$:
\begin{equation}\label{eq:weight_gauge_und}
t_e^{\prime} = \left\{ \begin{array}{ll}
c t_e, & \mbox{ if } v \mbox{ bounds } e,\\
t_e, &\mbox{otherwise}.
\end{array}
\right.
\end{equation}
Two reduced networks sharing the same graph for two sets of real positive edge weights represent the same point in the totally non--negative Grassmannian if and only if the edge weights can be obtained by composing transformations (\ref{eq:weight_gauge_und}) at the internal vertices of the graph. We remark that on unreduced graphs there is extra gauge freedom \cite{Pos}.

If the graph is directed and we assign a positive number $c_v$ to each internal vertex $v$ and the directed edge $e=\overrightarrow{uv}$ has initial weight $t_e$ then the gauge equivalent network has edge weight $t_e^{\prime}$ with
\begin{equation}\label{eq:weight_gauge_dir}
t_e^{\prime} = c_uc_v^{-1} t_e.
\end{equation}
Again, two networks on a perfectly orientated reduced graph are equivalent if and only if their edge weights are related by (\ref{eq:weight_gauge_dir}).

\smallskip

In \cite{Pos}, for any given oriented planar network in the disk, the formal \textbf{boundary measurement map} is defined as  
\begin{equation}\label{eq:bound_meas_map}
M_{ij} := \sum\limits_{P:b_i\mapsto b_j} (-1)^{\mbox{\tiny Wind}(P)} wt(P),\quad \quad i\in I, \; j\in \bar I,
\end{equation}
where $I$ is the base for the given orientation, the sum is over all directed paths $P$ from the source $b_i$ to the sink $b_j$, $wt(P)$ is the product of the edge weights of $P$ (counting multiplicities if an edge appears more than once in $P$), and $\mbox{Wind}(P)$ is its topological winding index (see \cite{Pos}). These formal power series sum up to subtraction free rational expressions in the weights \cite{Pos} and, for directed networks, their explicit expression in function of flows and conservative flows is provided in \cite{Tal2}. 

Let $I$ be the base inducing the orientation of the network $\mathcal N=(\mathcal G, \mathcal O(I), f)$ used in the computation of the boundary measurement map. Then (see \cite{Pos}), for each choice of positive edge weights associated to $f$, the image of the boundary measurement map is the point $[A^{bmm}]\in \S \subset \GTNN$ represented by the \textbf{boundary measurement matrix} $A\equiv A^{bmm}$ such that:
\begin{itemize}
\item The submatrix $A_I$ in the column set $I$ is the identity matrix;
\item The remaining entries $A^r_j = (-1)^{\sigma(i_r,j)} M_{ij}$, $r\in [k]$, $j\in \bar I$, where $\sigma(i_r,j)$ is the number of elements of $I$ strictly between $i_r$ and $j$.
\end{itemize}
The point $[A^{bmm}]\in \S$ is a function of the face weights $f$, and is independent on both the perfect orientation of $\mathcal G$ and the weight gauge \cite{Pos}. Finally, if $\S$ is an irreducible positroid cell, then the reduced row echelon matrix contains neither zero columns nor rows with just the pivot entry different from zero. 

\begin{figure}
\centering{\includegraphics[width=0.32\textwidth]{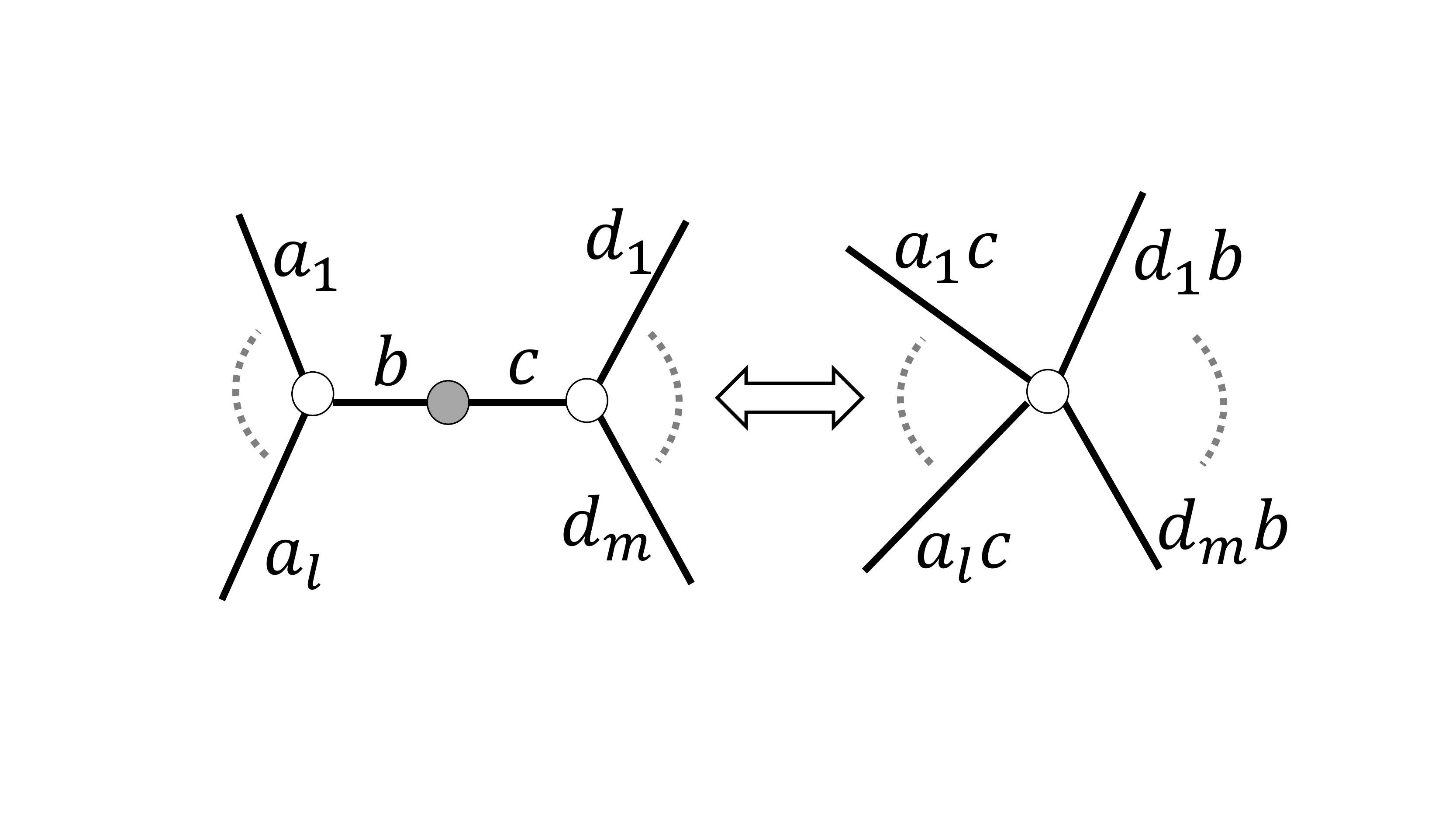}
\includegraphics[width=0.32\textwidth]{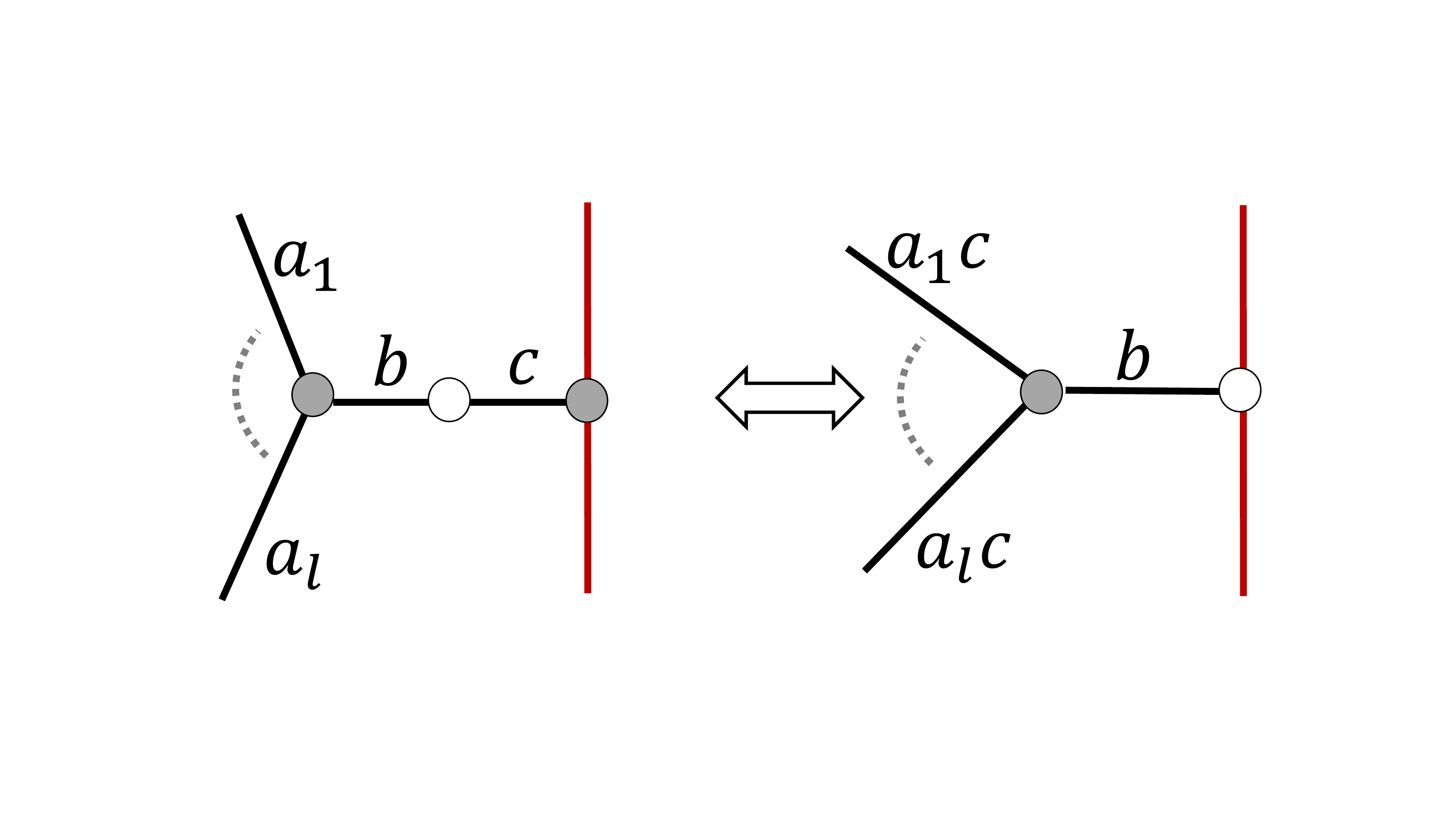}
\hfill
\includegraphics[width=0.33\textwidth]{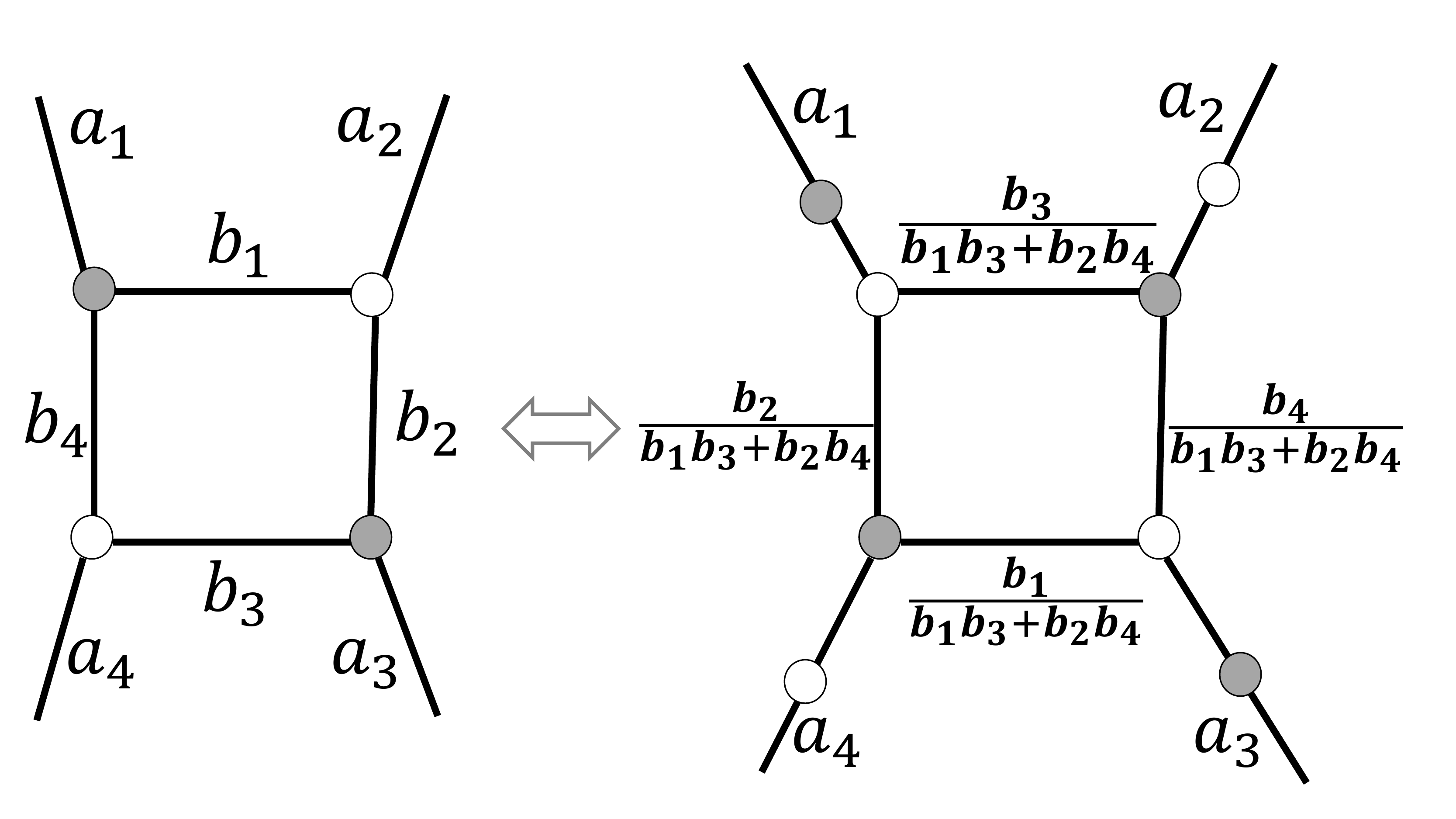}
\vspace{-.3 truecm}
\caption{\small{\sl Left: the contraction/expansion of a degree 2 black vertex; center: the removal/addition of a white boundary--adjacent vertex; right: the square move. Unlabeled edges carry unit weight.}}\label{fig:move}}
\end{figure}

In \cite{Pos} the graphs representing the same positroid $\mathcal M$ are classified. In this equivalence class a special role is played by reduced graphs. Indeed, if the graph is reduced, the boundary measurement map modulo the weight gauge equivalence is a homeomorphism from $\mathbb{R}_+^g$ to $\S$, where $g$ is the dimension of the positroid cell represented by $\mathcal G$ \cite{Pos}. 

\smallskip

Postnikov also classifies the network transformations which preserve the value of the boundary measurement map $[A^{bmm}]$. Since we use only reduced networks, we just need to define the actions of moves; for the reductions see \cite{Pos}. The possible moves for undirected bipartite networks are:
\begin{enumerate}
\item \textbf{Contraction/expansion of a vertex} Any degree 2 internal vertex not adjacent to a boundary vertex can be deleted and the two adjacent vertices merged (see Figure \ref{fig:move} [left]). With the inverse operation one splits an internal vertex into two vertices and inserts a degree 2 vertex of opposite color assigning unit weight to the new edges. 
\item \textbf{Removal/addition of a boundary--adjacent vertex} Any degree 2 internal vertex adjacent to the boundary may be removed, the boundary vertex changes color and the two edges become a single edge as in Figure \ref{fig:move} [center]. With the inverse operation one adds a degree 2 vertex in the middle of a boundary--adjacent edge, changes the color of the boundary vertex and assigns unit weight to the new edge.
\item \textbf{Square move} It is the transformation shown in Figure \ref{fig:move} [right] and is the only untrivial one since face weights are changed.
\end{enumerate}
The graph in Figure \ref{fig:strand_Gr26} [left] is equivalent to the Le--graph in Figure \ref{fig:Le_graph} via the contraction of bivalent vertices.
The corresponding formulas for the same moves on directed bipartite graphs may be easily obtained using (\ref{eq:edge_weights}).
We remark that by repeated expansions one may always arrive to graphs with vertex degrees no more than 3. 

\smallskip

An alternative characterization of the boundary measurement map using almost perfect matchings on bipartite graphs is provided in \cite{Lam2} where its equivalence with the boundary measurement map is proven using the characterization of the latter in terms of flows on directed graphs by Talaska \cite{Tal2} and the bijection between flows and almost perfect matchings proven in \cite{PSW}. 
 
\begin{theorem} \cite{Lam2} \label{theo:dimer_par}
Let $\{ t_e \}_{e\in \mathcal E}$ be the edge weights on the undirected bipartite graph $\mathcal G$.
Then each matching $M\subset \mathcal E$ defines a monomial 
\[
t^M = \prod_{e\in M} t_e.
\]
For any $k$--element subset $I\subset [n]$ define the partition function $D_I$ as the sum of the monomials for the matchings with boundary $I$:
\begin{equation}\label{eq:dimer_part}
D_I = \sum_{M \; : \; \partial M = I} t^M.
\end{equation}
By definition $D_I > 0$ if $I\in \mathcal M(\mathcal G)$, and zero otherwise. In particular, if $t_e=1$, for all edges $e$, then $D_I$ is the number of almost perfect matchings $M$ in $\mathcal G$ such that $\partial M =I$, which we denote $\Delta(\mathcal G, I)$.

Then the collection $\{ D_I\, : \, I \in \mathcal M(\mathcal G)\}$ are the Pl\"ucker coordinates of a point $[A^{dimer}]\in\S$ with $\mathcal M = \mathcal M(\mathcal G)$.  Weight gauge equivalent networks and move--reduction equivalent networks are mapped to the same point in $\S$.

Finally, for any choice of positive face weights on the graph, $[A^{dimer}]$ coincides with $[A^{bmm}]$, the value of Postnikov boundary measurement map.
\end{theorem}

Pl\"ucker coordinates are redundant coordinates because of Pl\"ucker relations; and there exists a minimal number of $D_{I_l}$, $l\in [g]$, where $g$ equals the dimension of $\S$, such that all other Pl\"ucker coordinates may be expressed as subtraction free rational expressions of the $D_{I_l}$. Such a set of Pl\"ucker coordinates forms a totally positive base in the sense of Fomin and Zelevinsky \cite{FZ}. An explicit totally positive base was constructed in \cite{Tal1} using Le--diagrams, whereas almost perfect matchings on reduced planar bipartite networks are used in \cite{MaSc} (see also \cite{MuSp}) for the same purpose.

\section{Kasteleyn matrices on planar bipartite networks in the disk}\label{sec:kas}

In this Section, given $\mathcal G$, a planar bipartite reduced graph  in the disk representing an irreducible positroid cell $\S\subset \GTNN$, we introduce a Kasteleyn signature on $\mathcal G$: such a signature is an assignment of $\pm 1$ to the edges fulfilling (\ref{eq:sign_cond}).
We then prove that this signature realizes the variant of Kasteleyn theorem in \cite{Sp}: 1)
maximal minors of the $|\mathcal B|\times |\mathcal W|$ Kasteleyn matrix share the same sign and count the number of almost perfect matchings of $\mathcal G$ with given boundary conditions; 2) equivalence classes of weighted Kasteleyn matrices provide a parametrization of $\S$. 
We remark that the transpose of a Kasteleyn matrix is a Kasteleyn matrix of a point in the dual cell $\Sprime$.

\begin{remark}
The reduced property of the graph is a sufficient condition to avoid zero elements in the systems of relations studied throughout the paper. Instead the irreducibility of the positroid cell just simplifies the overall construction.
\end{remark}
\begin{figure}
\centering{\includegraphics[width=0.37\textwidth]{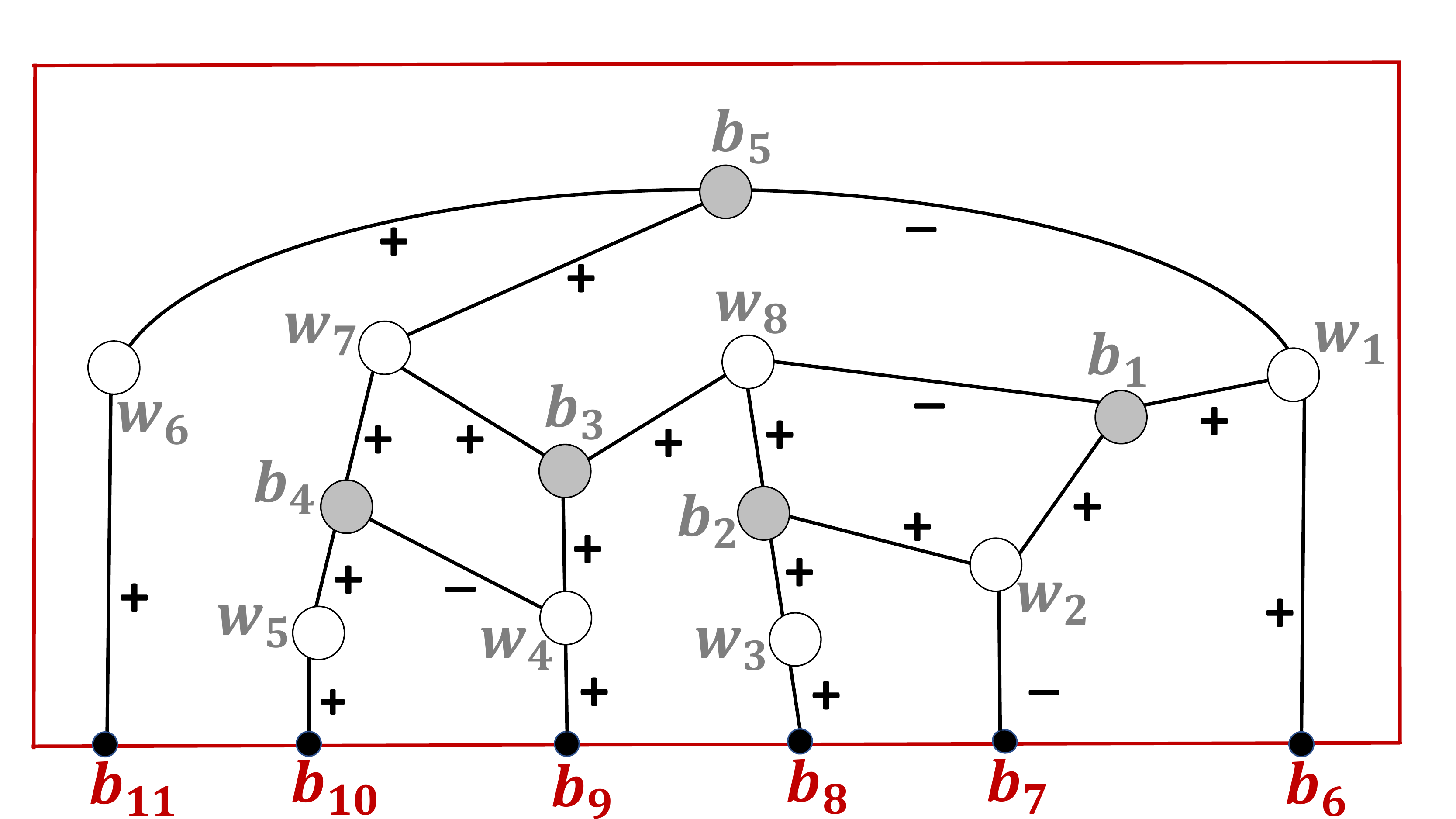}
\vspace{-.3 truecm}
\caption{\small{\sl A Kasteleyn signature for the reduced planar bipartite graph in the disk of Figure \ref{fig:graph}. Labels of the boundary vertices have been changed to be consistent with the notation of Definition \ref{def:kas_mat}.}}\label{fig:kas_sign}}
\end{figure}

\begin{definition}\textbf{Kasteleyn signature on $\mathcal G$}\label{def:kas_sign}
Let $\mathcal G= (\mathcal V=\mathcal B\cup \mathcal W, \mathcal E)$ be a reduced bipartite planar graph in the disk with boundary vertices of equal color as in Definition \ref{def:graph_red}. Assume that $\mathcal M(\mathcal G)$ is irreducible. A function $\sigma : \mathcal E \mapsto \{ \pm 1\}$ is a Kasteleyn signature if, for any finite face $\Omega$, 
\begin{equation}\label{eq:sign_cond}
\sigma(\Omega) = (-1)^{\frac{|\Omega|}{2}+1},
\end{equation}
where $|\Omega|$ denotes the number of edges bounding the face $\Omega$, and $\sigma(\Omega)$ is the total signature  of the face $\Omega$ of $\mathcal G$, that is the product over all edges $e\in \partial \Omega$ of the edge signature $\sigma(e)$:
\begin{equation}\label{eq:sign_face}
\sigma(\Omega) = \prod_{e \in \partial \Omega} \sigma(e).
\end{equation}
\end{definition}

In Figure \ref{fig:kas_sign} we show a Kasteleyn signature for the reduced graph of Figure \ref{fig:graph}.

\begin{remark} The number of edges bounding a finite external face is always even because the graph is bipartite and all boundary vertices share the same color.
\end{remark}

\begin{remark} In \cite{AGPR} it is called Kasteleyn a signature which satisfies (\ref{eq:sign_cond}) at the internal faces and depends also on the number of boundary vertices at the external faces. We compare the properties of the two signatures in Section \ref{sec:AGPR}.
\end{remark}

Next Proposition is the restatement of a classical Lemma by Kasteleyn \cite{Kas2} in the present setting. 

\begin{proposition}\label{prop:kas_sign_exist}
Kasteleyn signatures exist on reduced planar bipartite graphs in the disk with boundary vertices of equal color.
\end{proposition}

\begin{proof}
If $\mathcal G$ is the Le--graph, one assigns $\sigma(e)=1$ to any horizontal edge $e$ included those corresponding to the lexicographically minimal base of the positroid $\mathcal M$ represented by the graph. Then there remain exactly $g$ vertical edges, where $g$ is both the number of finite faces of $\mathcal G$ and the dimension of $\S$. Exactly one such vertical edge $e_i$ is the NW boundary of the finite face $\Omega_i$, $i\in [g]$. Then, starting from the last row of the Le--diagram and proceeding from right to left and bottom up, one chooses $\sigma(e_i)=\pm 1$ so that (\ref{eq:sign_cond}) is fulfilled for any finite face.

If the reduced graph $\mathcal G$ is move equivalent to the Le--graph via a finite sequence of contraction/expansions at internal vertices and square moves, one can obtain a Kasteleyn signature on $\mathcal G$ using the transformation of face signatures under the action of the moves illustrated in Figure \ref{fig:move_1}. 

\begin{figure}
\centering{\includegraphics[width=0.35\textwidth]{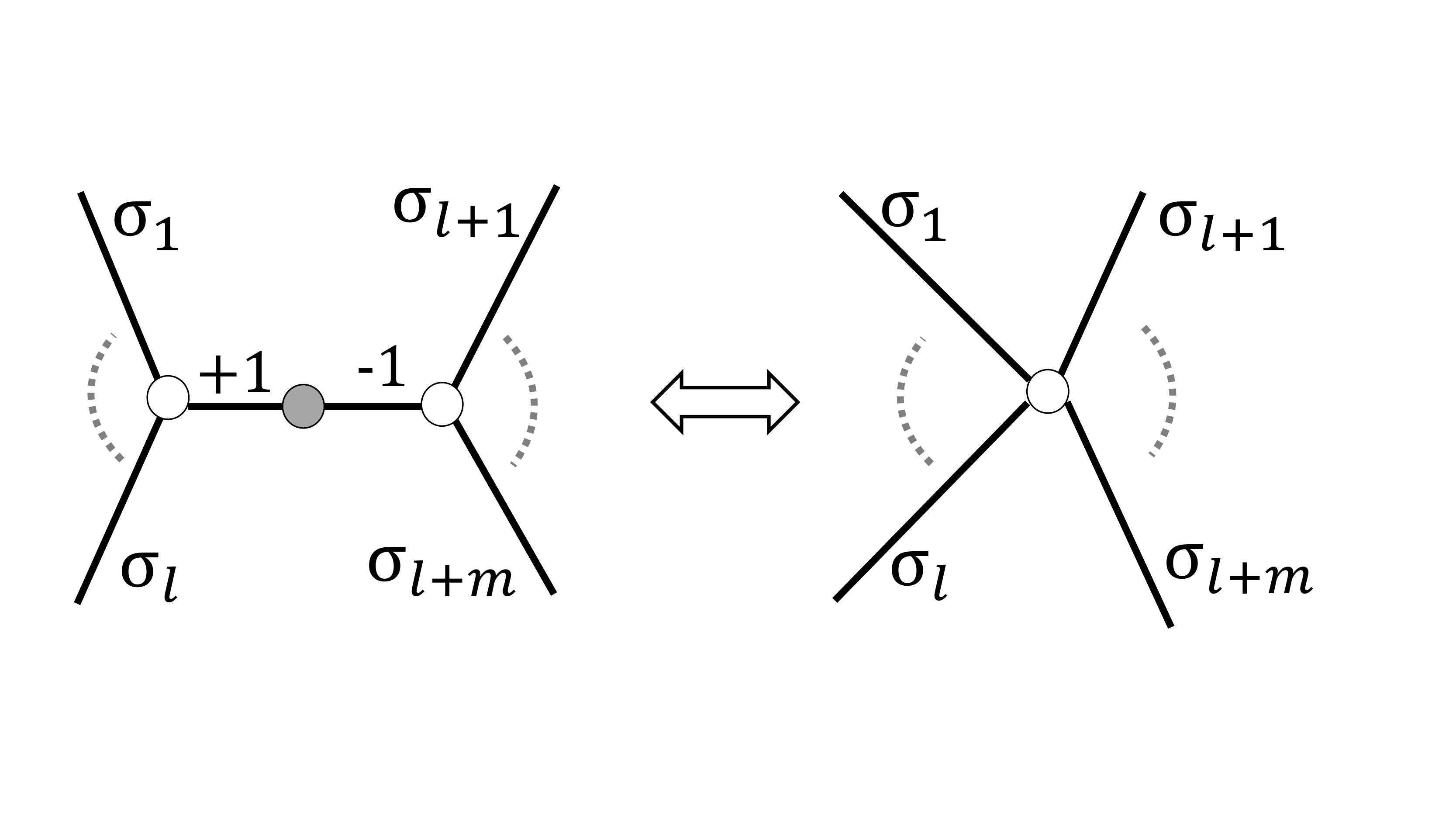}
\hspace{ .5 truecm}
\includegraphics[width=0.35\textwidth]{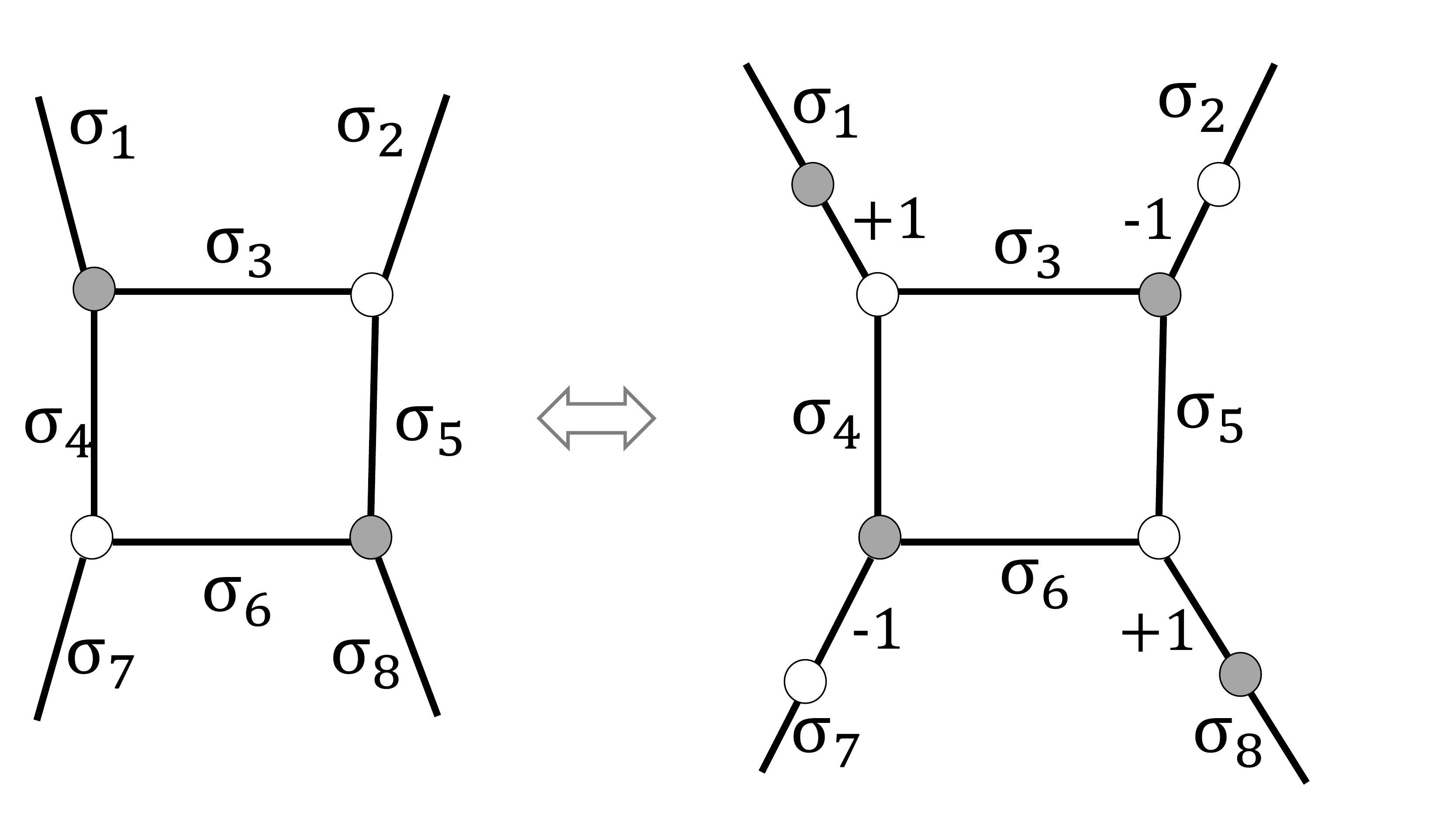}
\vspace{-.3 truecm}
\caption{\small{\sl The transformation of signatures for the contraction/expansion of a degree 2 vertex [left] and for the square move [right].}}\label{fig:move_1}}
\end{figure}

The removal/addition of a boundary--adjacent vertex must involve all boundary vertices in order to keep the even parity property of the external faces. If one adds a boundary-adjacent vertex next to each boundary vertex and call $e_j$ the edge added to the boundary vertex $b_j$, $j\in [n]$. If $\sigma$ is the Kasteleyn signature on the initial graph, then the signature $\sigma^{\prime}$ on the transformed graph such that
\begin{equation}\label{eq:kas_sign_transf}
\sigma^{\prime}(e) =\left\{ \begin{array}{ll} \sigma(e), & \mbox{ if } e \mbox{ is an edge common to both graphs};\\
(-1)^{j-1}, & \mbox{ if } e = e_j, \quad j	\in [n],
\end{array}
\right.
\end{equation}
is Kasteleyn. We illustrate an example of such move in Figures \ref{fig:ex_Gr26} and \ref{fig:example_Gr46_dual} [left].
The explicit transformation in the case of removal of boundary--adjacent vertices follows along similar lines.
\end{proof}

\begin{remark}
In Section \ref{sec:geom_sign} we illustrate an alternative method to construct Kasteleyn signatures using the geometric signatures introduced in \cite{AG4}.
\end{remark}

There is of course not a unique Kasteleyn signature for a given graph $\mathcal G$. As in the case of lattices with periodic boundary conditions \cite{Ken1}, a vertex gauge transformation fully characterizes equivalent signatures on plabic graphs in the disk. 
 
\begin{definition}\textbf{Equivalent Kasteleyn signatures}
Two signatures $\sigma$ and $\tilde \sigma$ on $\mathcal G$ are equivalent if they both satisfy Definition \ref{def:kas_sign}.
\end{definition}

\begin{definition}\textbf{Gauge transformation between Kasteleyn signatures}\label{def:eq_sign}
A function $\alpha : \mathcal V\mapsto \{ \pm 1\}$ is a gauge transformation if it takes the same value at all boundary vertices $b_{j}$:
\begin{equation}\label{eq:equiv_bound}
\alpha (b_j) = \alpha (b_1),  \quad\quad j\in [n].
\end{equation}
\end{definition}

\begin{proposition}\textbf{Equivalent Kasteleyn signatures}\label{prop:equiv_sign} 
Two signatures $\sigma$ and $\tilde \sigma$ on $\mathcal G$ are equivalent if and only if there is a gauge transformation $\alpha : V\mapsto \{ \pm 1\}$ such that at any edge $e=\overline{bw}$,
\begin{equation}\label{eq:equiv_cond}
\tilde 	\sigma (e) = \alpha (b) \sigma (e) \alpha(w).
\end{equation}
\end{proposition}
In one direction the proof is obvious, in the other direction the proof follows as in \cite{Ken1} identifying the boundary vertices.

A Kasteleyn matrix $K^{\sigma}$ is associated to any given Kasteleyn signature $\sigma$. 

\begin{definition}\textbf{Kasteleyn sign matrix}\label{def:kas_mat}
Let $\mathcal G=(\mathcal V=\mathcal B\cup \mathcal W,\mathcal E)$ be a given planar bipartite network in the disk with boundary vertices of equal color. Let $|\mathcal B|$, $|\mathcal W|$ respectively denote the number of black and of white vertices. Let $\sigma$ be a Kasteleyn signature for $\mathcal G$. 
Following \cite{Sp}, we label the white vertices from 1 to $|\mathcal W|$, and the black vertices from 1 to $|\mathcal B|$, so that the boundary vertices share the highest labels of their color and are labeled in increasing order clockwise. Then the $|\mathcal B|\times |\mathcal W|$ Kasteleyn sign matrix $K^{\sigma}$ associated to such data is  
\begin{equation}\label{def:kas_entries}
K^{\sigma}_{ij} = \left\{ \begin{array}{ll}
\sigma (e) & \mbox{ if the edge } e \mbox{ joins } b_i \mbox{ and } w_j,\\
0 & \mbox{ if there is no edge joining } b_i \mbox{ and } w_j.
\end{array}
\right.
\end{equation}
If one assigns positive weights to the edges of $\mathcal G$, $t: \mathcal E\mapsto \mathbb{R}^+$, the weighted Kasteleyn matrix
$K^{\sigma, wt}$, is
\begin{equation}\label{eq:kas_wt_entries}
K^{\sigma,wt}_{ij} = \left\{ \begin{array}{ll}
\sigma (e) t_e & \mbox{ if } e=\overline{b_iw_j},\\
0 & \mbox{ otherwise.}
\end{array}
\right.
\end{equation}
\end{definition}

In the next Theorem we show that $K^{\sigma}$ satisfies the variant of Kasteleyn theorem proven in \cite{Sp}:
the maximal minors of $K^{\sigma}$ indexed by the boundary dimer configurations share the same sign and count the number of almost perfect matchings of $\mathcal G$. To fix ideas we choose the black color for the boundary vertices.

\begin{theorem}\label{theo:kast}\textbf{The number of almost perfect matchings and the minors of $K^{\sigma}$}
Let $\mathcal G$ be a planar bipartite graph in the disk with black boundary vertices representing the positroid cell $\S\subset \GTNN$. Let $N$ be the number of internal black vertices of $\mathcal G$, so that $|\mathcal W|=N+k$. Assume a labeling of vertices such that boundary vertices are labeled clockwise in increasing order $b_{N+1},\dots, b_{N+n}$.
Let $\sigma$ be a Kasteleyn signature on $\mathcal G$ and $K^{\sigma}$ be its Kasteleyn sign matrix. For any $k$ element subset $I$ of $\partial G$, let $K_I$ be the submatrix of $K^{\sigma}$ using all columns, the first $N$ rows and the additional $k$ rows indexed by $I$. 
Then 
\begin{enumerate}
\item For any pair of $k$--element subsets $I,J$ the determinants of the submatrices $K_I, K_J$ share the same sign,
\begin{equation}\label{eq:sign_K_min}
\det (K_I) \, \cdot \, \det (K_J)\, \ge \, 0;
\end{equation}
\item The absolute value of $\det (K_I)$ is $\Delta(\mathcal G, I)$, the number of almost perfect matchings $M$ in $\mathcal G$ such that $\partial M =I$,
\begin{equation}\label{eq:apm_count}
\Delta(\mathcal G, I) = |\det (K_I)|.
\end{equation}
\end{enumerate}
\end{theorem}

\begin{proof}
The proof is a straightforward adaptation of the original Kasteleyn Theorem \cite{Kas2, Ken1} to the present setting. First we check (\ref{eq:apm_count}). Let ${\mathcal B}_i$ and ${\mathcal B}_I$ respectively denote the set of internal black vertices of $\mathcal G$ and that of the boundary vertices indexed by $I$. By definition
\[
\det K_I = \sum_h \mbox{sign}\, (h) \prod_{w\in {\mathcal W}} K_{h(w),w},
\]
where the summation is over all bijections $h:\mathcal W \mapsto {\mathcal B}_i \cup {\mathcal B}_I$. Since $K_{h(w),w} \not =0$ if and only if there is an edge joining
$w$ and $h(w)$, the non zero terms contributing to the determinant correspond exactly to the almost perfect matchings $\pi$ using the boundary vertices ${\mathcal B}_I$. Therefore
\[
\det K_I = \sum_{\pi} \mbox{sign}\, (\pi) \prod_{(b,w)\in \pi} K_{b,w} = \sum_{\pi} \mbox{sign}\, (\pi) \prod_{(b,w)\in \pi} \sigma (\overline{bw}).
\]
Now $\mbox{sign}\, (\pi) \prod_{(b,w)\in \pi} \sigma (\overline{bw})$ is the same for all matchings sharing the same boundary vertices. Indeed, if $\pi,\pi^{\prime}$ are two such matchings, then their union is a set of simple cycles and double edges and the statement easily follows using Definition \ref{def:kas_sign}.

Next let us prove (\ref{eq:sign_K_min}). It is sufficient to check the formula in the case $J= I 	\backslash \{ i \} \cup \{j\}$. Let $\pi, \pi^{\prime}$ be two almost perfect matchings respectively for the boundary sets $I$, $J$. Then $\pi \cup \pi^{\prime}$ is a Temperley--Lieb subgraph, that is the union of simple cycles, double edges and a path from $b_{N+i}$ to $b_{N+j}$. If $j=i+ 1$, then the statement follows identifying the two boundary vertices so to obtain a cycle with flat curvature. In the case $j=i+l$, with $l >1 $, then the boundary vertices $b_{N+i+1},\dots b_{N+j-1}$ are either used by both $\pi$ and $\pi^{\prime}$ or not used. Then again identifying $b_{N+i}$ and $b_{N+j}$ we obtain a cycle with flat curvature and the statement follows.
\end{proof}

\begin{remark}\label{rem:lab_pos} Since the sign of the minors changes by exchanging two consecutive rows, in the rest of the paper we assume a labeling of the internal vertices such that $\det (K_I)\ge 0$ for all $k$--element subsets. 
\end{remark}

\begin{remark}\textbf{The point in the totally non--negative Grassmannian associated to $K^{\sigma,wt}$}
Given a Kasteleyn signature $\sigma$ and a positive edge weighting on the graph $\mathcal G$, the maximal minors of the weighted Kasteleyn matrix $K^{\sigma,wt}_I$ are different from zero if and only if $I\in \mathcal M$, where $\mathcal M$ is the positroid of $\S$ represented by $\mathcal G$. Therefore, following \cite{Sp}, 
the minors $K^{\sigma,wt}_{I}$ are the Pl\"ucker coordinates of a point $[A^{\sigma, wt}]\in \S\subset \GTNN$ uniquely identified by the condition that for any $k$--element subset $I\subset [n]$ 
\begin{equation}\label{eq:plu_part}
\det (A^{\sigma, wt}_I )= \det K^{\sigma, wt}_I.
\end{equation}
\end{remark}

\begin{lemma}\label{lem:kas_point}
The point $[A^{\sigma, wt}]$ is the same for gauge equivalent edge weightings on $\mathcal G$. Moreover, if $\sigma^{\prime}$ is another Kasteleyn signature in the equivalence class for $\mathcal G$ and the face weights are kept fixed, then $[A^{{\sigma^\prime}, wt}] = [A^{\sigma, wt}]$, that is for a given graph it is a function of the face weights $f$. Finally if the networks $({\mathcal G}^{\prime}, f^{\prime})$ and $(\mathcal G, f)$ are move equivalent, then for any Kasteleyn signature $\sigma$ on $\mathcal G$ and $\sigma^{\prime}$ on ${\mathcal G}^{\prime}$, the points in the Grassmannian constructed  using the respective Kasteleyn matrices coincide:  $[A^{{\sigma^\prime}, wt}({\mathcal G}^{\prime}, f^{\prime})] =[A^{\sigma, wt}({\mathcal G}, f)]$.
Therefore given the (move equivalence class of the) network $(\mathcal G, f)$, there exists a unique point in $\Pi_{\mathcal M}$ which we denote 
\[
[A^{kas}]\equiv [A^{\sigma, wt}],
\] 
such that the Pl\"ucker coordinates of its representative matrix $A^{\sigma, wt}$ are defined in (\ref{eq:plu_part}).
\end{lemma}

The proof of the above Lemma is trivial taking into account the action of (\ref{eq:weight_gauge_und}) on the edge weights, the characterization of gauge equivalent signatures in Proposition \ref{prop:equiv_sign} and the action of the moves on signatures. 

Given a reduced planar bipartite graph in the disk $\mathcal G$, so far we have illustrated three natural parametrizations of the positroid cell $\S$, where $\mathcal M= \mathcal M(\mathcal G)$. Indeed, the following points in $\S$ are assigned to the move equivalence class of the network $(\mathcal G, f)$:
\begin{enumerate}
\item The point $[A^{bmm}]$ is obtained choosing a perfect orientation on the graph and computing Postnikov boundary measurement map \cite{Pos} (see formula (\ref{eq:bound_meas_map}));
\item The Pl\"ucker coordinates $D_I$ of the point $[A^{dimer}]$ are the weighted matchings with boundary $\partial M= I$ in (\ref{eq:dimer_part}): $D_I = \sum_{M \; : \; \partial M = I} \prod_{e\in M} t_e$ \cite{Lam2};
\item The Pl\"ucker coordinates of the point $[A^{kas}]$ are the minors of a weighted Kasteleyn sign matrix (see (\ref{eq:plu_part})).
\end{enumerate}

In the previous Section we have recalled that $[A^{bmm}] = [A^{dimer}]$ \cite{Lam2}. In \cite{Sp} it is proven that, if $\sigma$ is a signature such that Theorem \ref{theo:kast} holds and $t_{b,w}$ is a positive edge weighting on the graph, then $[A^{kas}]=[A^{dimer}]$.
Therefore the three parametrizations of $\S$ coincide. 

\begin{theorem} \textbf{Parametrization of positroid cells via Kasteleyn weighted matrices}\label{cor:equiv_par} \cite{Sp}
Let $\S\subset \GTNN$ be given and let $\mathcal G$ be a reduced planar bipartite graph with boundary vertices of equal color representing $\S$. Let $f: {\mathcal G}^* \to \mathbb{R}^+$ be a positive face weighting of $\mathcal G$. Let $\mathcal N =(	\mathcal G, f)$ be the corresponding network and let $[A^{bmm}]=[A^{dimer}]\in \S$ be as above.
Let $\sigma$ be a Kasteleyn signature for $\mathcal G$, and let $K^{\sigma, wt}$ be a weighted Kasteleyn matrix representing $\mathcal N$. Let $A^{kas}$ be such that for any $k$--element subset $I	\subset [n]$
\begin{equation}\label{eq:det_A_I}
\det  (A^{kas})_I  = \det K^{\sigma, wt}_I.
\end{equation}
Then 
\begin{equation}\label{eq:A}
[A^{kas}] = [A^{bmm}]=[A^{dimer}].
\end{equation}
\end{theorem}

In the following, we simplify notations to $[A]= [A^{kas}] = [A^{bmm}] = [A^{dimer}]$. 

\begin{remark}\textbf{Alternative proof of Theorem \ref{cor:equiv_par} using geometric signatures}
In \cite{AG4} (see also Section \ref{sec:geom_sign}), geometric signatures on directed plabic graphs are introduced and it is proven that they induce Postnikov boundary measurement map for any choice of positive face weights. In Theorem \ref{theo:main} we explain the relation between geometric and Kasteleyn signatures. Therefore Theorem \ref{cor:equiv_par} follows also from the relation between these two signatures (see Corollary \ref{cor:pos_sp_equiv}).
\end{remark}

If one chooses $\mathcal G$ with black boundary vertices, it is easy to reconstruct a representative $k\times n$ matrix of $[A]$ starting from $(K^{\sigma, wt})^T$, the transpose of $K^{\sigma, wt}$ \cite{Sp}: applying row operations one can transform $(K^{\sigma, wt})^T$ into a matrix in block form
\begin{equation}\label{eq:mat_A}
\renewcommand{\arraystretch}{1.5}
  \begin{blockarray}{ccc}
	& N & n\\
    \begin{block}{c(c|c)}
N\,\, & \,\mbox{Id}_N \, & *\,\, \\
\cline{2-3}		
k\,\, & 0          & \, A \,\,\\
    \end{block}
  \end{blockarray}
\end{equation}
without changing any maximal minor. If $A_I$ denotes the maximal minor using all rows of $A$ and the columns indexed by $I$, then
\[
\det A_I = \det K^{\sigma, wt}_I.
\]

\begin{theorem}\label{theo:speyer} \textbf{The representative matrix in $\S\subset\GTNN$ associated to the weighted Kasteleyn matrix} \cite{Sp}
Let $\mathcal N= (\mathcal G, f)$ be a bipartite network, where the graph $\mathcal G=(\mathcal V,\mathcal E)$ has black boundary vertices and represents the positroid cell $\S\subset \GTNN$, and $f$ are positive face weights. Let $\sigma$ be a Kasteleyn signature on $\mathcal G$ and $K^{\sigma, wt}$ be its Kasteleyn weighted matrix as in (\ref{eq:kas_wt_entries}). 
Then the $k\times n$ matrix $A$ defined in (\ref{eq:mat_A}) represents $[A^{kas}]$.
\end{theorem}

\begin{remark}
In Theorem \ref{theo:sol_kas_sys_2} we provide an alternative way to construct a representative matrix of $[A^{kas}]$ using Kasteleyn system of relations.
\end{remark}

Let us illustrate the construction of $K^{\sigma, wt}$ and of $[A^{kas}]$ for the network shown in Figure \ref{fig:ex_Gr26}.

\begin{figure}
  \centering{\includegraphics[width=0.37\textwidth]{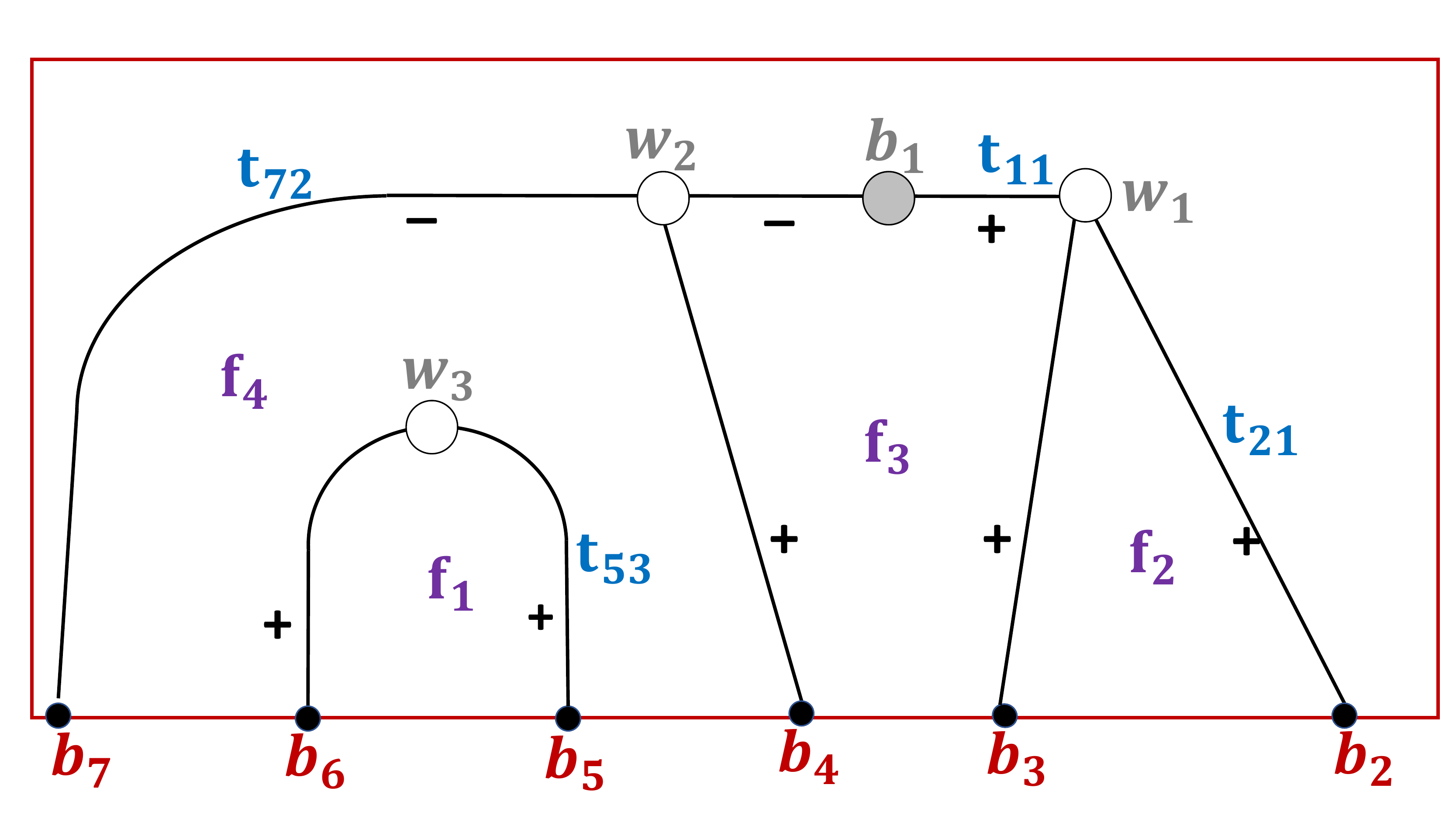}}
	\vspace{-.3 truecm}
  \caption{\small{\sl A choice of Kasteleyn signature and face weights for a network with the reduced graph of Figure \ref{fig:strand_Gr26}. The relation between face and edge weights is as in (\ref{eq:edge_weight_ex}).}\label{fig:ex_Gr26}}
\end{figure}

\begin{example}\label{ex:1}
Let $\mathcal N$ be the reduced bipartite network in Figure \ref{fig:ex_Gr26} representing a point in the positroid cell in $Gr^{\mbox{\tiny TNN}}(2,6)$ of Figure \ref{fig:strand_Gr26} where the face weights $f_1,\dots,f_4$ are assumed positive.
A possible minimal positive edge weight parametrization, $t_{21},t_{11}, t_{53}, t_{72}$, is obtained using (\ref{eq:face_weights})
\begin{equation}\label{eq:edge_weight_ex}
f_1 =t_{53}^{-1}, \quad\quad f_2 = t_{21}^{-1}, \quad\quad f_3 = t_{11}, \quad\quad f_4 = t_{72}t_{53},
\end{equation}
where all other edges carry unit weights.
A Kasteleyn signature satisfying Definition \ref{def:kas_sign} is marked with $\pm$ on the edges of $\mathcal N$ in the Figure. The transpose of the weighted Kasteleyn matrix (\ref{def:kas_mat}) is then
\begin{equation}
(K^{\sigma,wt})^T = \left( 
\begin{array}{ccccccc} 
t_{11} & t_{21} & 1 & 0 & 0 & 0 & 0\\ 
\noalign{\medskip}
-1 & 0 & 0 & 1 & 0 & 0 & -t_{72}\\ 
\noalign{\medskip}
0 & 0 & 0 & 0 & t_{53} & 1 & 0 \end {array} 
\right).
\end{equation}
Upon transforming $(K^{\sigma,wt})^{T}$ into the block form (\ref{eq:mat_A}), 
\begin{equation}\label{eq:ex_Kas_block}
\renewcommand{\arraystretch}{1.5}
  \begin{blockarray}{c|cccccc}
    \begin{block}{(c|cccccc)}
\, 1 & 0 & 0 & -1 & 0 & 0 & t_{72} \\ 
\cline{1-7}		
\, 0 & t_{21} & 1 & t_{11} & 0 & 0& -t_{11}t_{72} \\ 
\, 0 & 0 & 0 & 0 & t_{53} & 1 & 0 \\
    \end{block}
  \end{blockarray}
\end{equation}
the point $[A]\in \S$ is identified by the submatrix $A$ in the SE block; its 
reduced row echelon form is
\begin{equation}\label{eq:A_RREF_ex}
A^{\mbox{\tiny RREF}} = 
 \left( \begin{array}{cccccc} 
1 & t_{21}^{-1} & t_{11} t_{21}^{-1} & 0 & 0 & -t_{11}t_{72}t_{21}^{-1}\\ 
\noalign{\medskip}
0 & 0 & 0 & 1 & t_{53}^{-1} & 0
\end {array}
 \right).
\end{equation}
We remark that, in agreement with Theorem \ref{cor:equiv_par}, $A^{\mbox{\tiny RREF}}$ coincides with the matrix in Postnikov construction \cite{Pos} for the same choice of face weights and the acyclic orientation of $\mathcal G$ with respect to the lexicographically minimal base $\{1, 4\}$.
\end{example}

\smallskip

A natural bijection between dual positroid cells $\S$ and $\Sprime$ (see Definition \ref{def:dual_cell}) is associated to the operation of transposition of Kasteleyn matrices. 

\begin{definition}\textbf{Duality transformation between networks}\label{def:dual_net}
Let $\S\subset \GTNN$ be given and let $\mathcal G$ be a reduced planar bipartite graph with boundary vertices of equal color representing $\S$. Let $f: {\mathcal G}^* \to \mathbb{C}^*$ be a face weighting of $\mathcal G$: for any finite face $\Omega_i$ let $f_i =f(\Omega_i) \in \mathbb{C}^*$ be the face weight of $\Omega_i$. Let $\mathcal N =(	\mathcal G, f)$ be the corresponding network. Then the dual network $\overline{\mathcal N} = (	\overline{\mathcal G}, \bar{f})$ is obtained from $\mathcal N$ by the following transformation:
\begin{enumerate}
\item $\overline{\mathcal G}$ is the dual graph to ${\mathcal G}$ obtained by changing the color of all vertices of ${\mathcal G}$, boundary vertices included;
\item $\bar{f}$ is the reciprocal of the weighting $f$: if $f_i$ is the face weight of $\Omega_i$ in ${\mathcal N}$ then $\bar{f}_i=1/f_i$ is its weight in $\overline{\mathcal N}$.
\end{enumerate}
\end{definition}

In Figure \ref{fig:example_Gr46_dual} we assume that the face weights are real and positive: then the network on the left represents a point in $\S\subset Gr^{\mbox{\tiny TNN}} (2,6)$ where the cell is the same as in Figure \ref{fig:strand_Gr26}, whereas the one on the right represents the dual point in $\Sprime\subset Gr^{\mbox{\tiny TNN}} (4,6)$ for the cell of Figure \ref{fig:strand_Gr26_dual}.
\begin{proposition}\textbf{Weighted Kasteleyn matrices and duality in totally non--negative Grassmannians}\label{prop:dual}
Let $\mathcal G=(\mathcal B\cup \mathcal W, \mathcal E)$ be a reduced planar bipartite graph in the disk with white 
boundary vertices representing the positroid cell $\S\subset \GTNN$. Let $[A]\in \Pi_{\mathcal M} $ be the point represented by $(\mathcal G, f)$.
 
Then  $(\overline{\mathcal G}, 1/f)$ represents a point in $[\bar A]\in\Pi_{\overline{\mathcal M}} \subset Gr (n-k,n)$, where $\Pi_{\overline{\mathcal M}}$ is the dual positroid variety of Definition \ref{def:dual_cell}. In particular, if the face weights $f$ are real and positive, then $[A]\in\S\subset Gr^{\mbox{\tiny TNN}} (k,n)$ and $[\bar A] \in\Sprime\subset Gr^{\mbox{\tiny TNN}} (n-k,n)$, where $\Sprime$ is the dual positroid cell of Definition \ref{def:dual_cell}.

Finally, if $K^{\sigma, wt}$ is a weighted Kasteleyn matrix for the network $(\mathcal G, f)$ and the Kasteleyn signature $\sigma$, then its transpose, $(K^{\sigma, wt})^T$ is a Kasteleyn matrix for the dual network $(\overline{\mathcal G}, 1/f)$. In this case, the transformation of  $K^{\sigma, wt}$ into the block form 
\begin{equation}\label{eq:mat_A_white}
 \renewcommand{\arraystretch}{1.5}
 \begin{blockarray}{ccc}
	& N & n\\
    \begin{block}{c(c|c)}
 N\,\, & \,\mbox{Id}_N \, & * \\
\cline{2-3}		
n-k\,\, & 0          & \, \bar A \,\\
    \end{block}
  \end{blockarray}
\end{equation}
provides a representative matrix $\bar A$ of $[\bar A]$.
\end{proposition}

\begin{figure}
\centering{\includegraphics[width=0.37\textwidth]{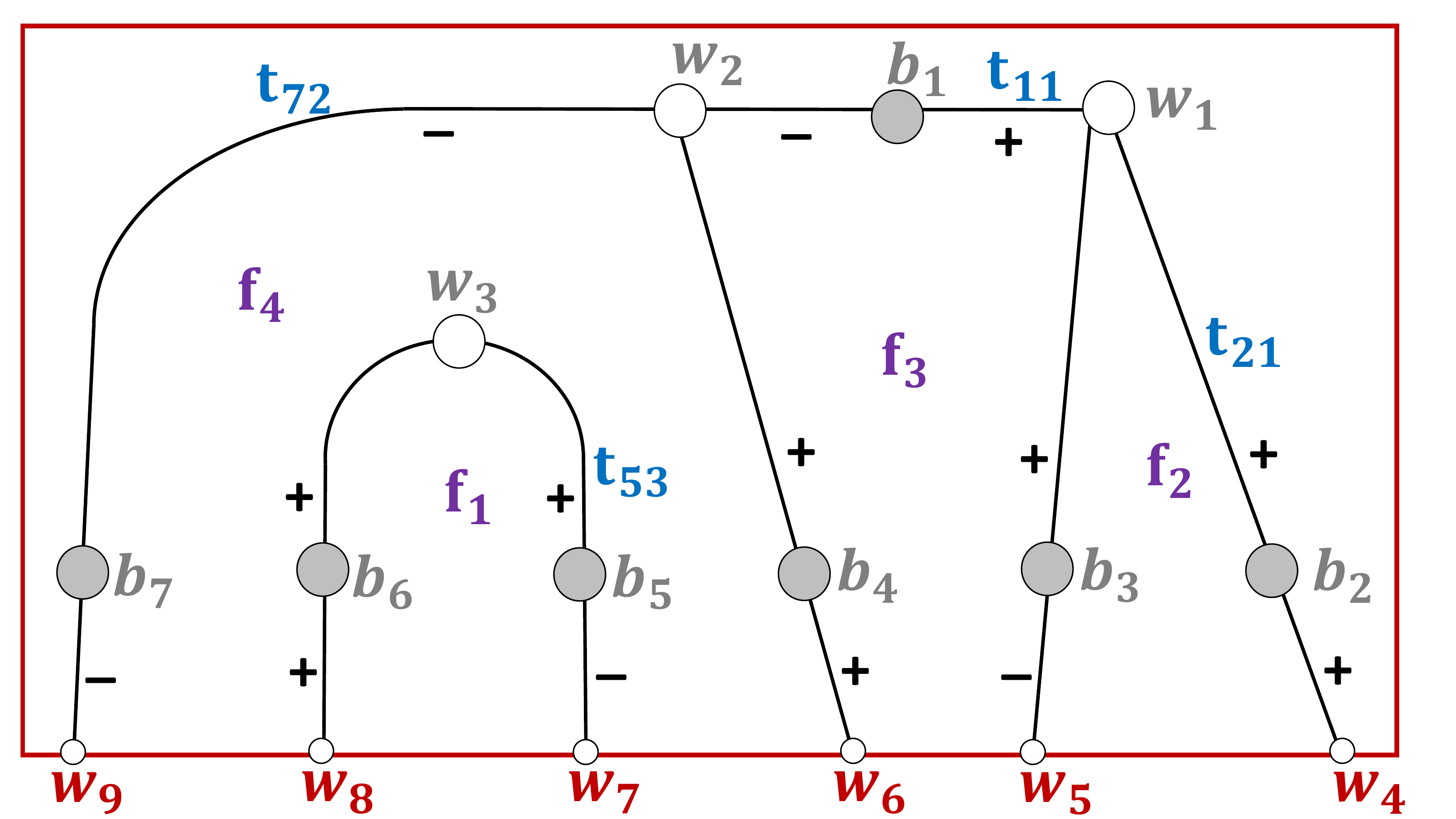}
\hspace{.7 truecm}
\includegraphics[width=0.37\textwidth]{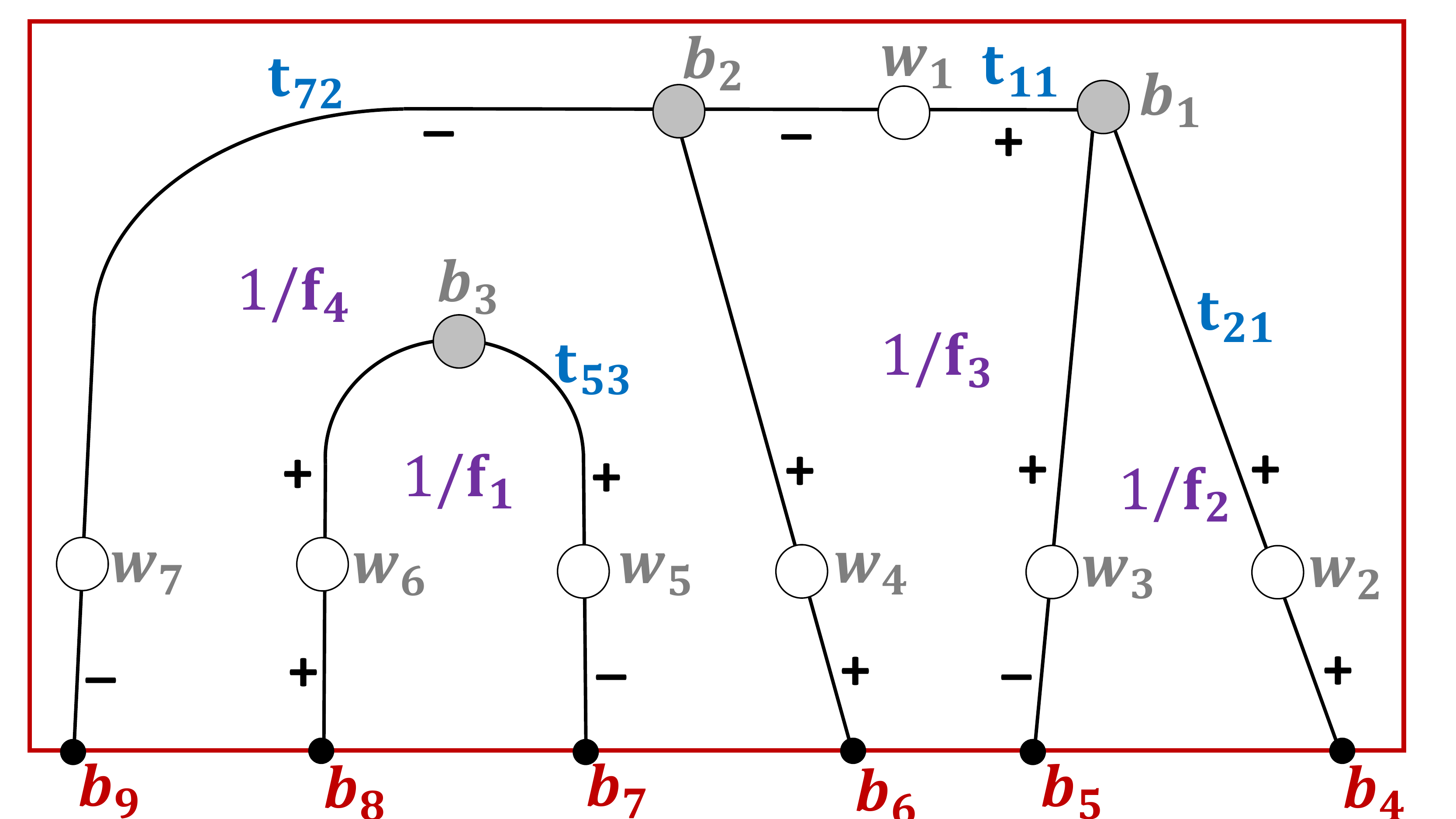}}
\caption{\small{\sl The network on the left represents a point in $Gr^{\mbox{\tiny TNN}}(2,6)$, whereas its dual network on the right represents a point in $Gr^{\mbox{\tiny TNN}}(4,6)$. The Kasteleyn matrix of the dual network is the transpose of the initial one. For the correspondence between Le--diagrams and derangements see Figures \ref{fig:strand_Gr26} and \ref{fig:strand_Gr26_dual}.}\label{fig:example_Gr46_dual}}
\end{figure}

\begin{proof}
The transformed graph represents the dual positroid cell $\Pi_{\overline{\mathcal M}}$, and the face transformation preserves the total non--negativity property. It is also evident that the Kasteleyn weighted matrices of the dual networks are related by transposition (the labeling of the black (respectively white) vertices of $\mathcal G$ becomes the labeling of the white (respectively black) vertices of $\overline{\mathcal G}$). Therefore
\[
\det A_I \; =	\; \det {\bar A}_{\bar I} \; = \; \det ( K^{\sigma, wt}_I),
\]
where $\sigma$ is a Kasteleyn signature for $\mathcal G=(  \mathcal B\cup  \mathcal W, \mathcal E)$, and we assume a labeling of the vertices satisfying Definition \ref{def:kas_mat} and Remark \ref{rem:lab_pos}.
Finally, the transformation of  $K^{\sigma, wt}$ into the block form satisfies (\ref{eq:mat_A_white}).
\end{proof}

This map looks similar to the twist map studied in \cite{MuSp} to relate Postnikov parametrization to the one introduced in \cite{MaSc} using dimer partitions, but this is not so. Indeed in our setting the transformed point $[\bar A]$ belongs to the dual cell $\Sprime\subset Gr^{\mbox{\tiny TNN}} (n-k,n)$, whereas the twist map $\tau$ in \cite{MaSc,MuSp} acts on face weights as in Definition \ref{def:dual_net}, but without changing the graph so that both $[A]$ and $\tau ([A])$ belong to $\S\subset  Gr^{\mbox{\tiny TNN}} (k,n)$.

Let us apply the duality transformation to the Example of the previous Section.

\begin{example}
The network in Figure \ref{fig:example_Gr46_dual}[left] represents the same point in $Gr^{\mbox{\tiny TNN}}(2,6)$ as in Example \ref{ex:1}. As before the correspondence between edge and face weights is
$f_1 =t_{53}^{-1}$, $f_2 = t_{21}^{-1}$, $f_3 = t_{11}$ and $f_4 = t_{72}t_{53}$.
For the labeling of the vertices in the Figure, the transpose of the weighted Kasteleyn matrix (\ref{def:kas_mat}) is 
\begin{equation}
(K^{\sigma,wt})^T = \left( 
\begin{array}{ccccccccc} 
t_{11}  & -1 & 0 & 0 & 0 & 0 & 0 & 0 & 0 \\ 
\noalign{\medskip}
t_{21} & 0 & 0 & 1 & 0 & 0 & 0 & 0 & 0\\ 
\noalign{\medskip}
1 & 0 & 0 & 0 & -1 & 0 & 0 & 0 & 0 \\ 
\noalign{\medskip}
0 & 1 & 0 & 0 & 0 & 1 & 0 & 0 & 0\\ 
\noalign{\medskip}
0 & 0 &t_{53} & 0 & 0 & 0 & -1 & 0 & 0 \\
\noalign{\medskip}
0 & 0 & 1 & 0 & 0 & 0 & 0 & 1 & 0 \\ 
\noalign{\medskip}
0 & -t_{72} & 0 & 0 & 0 & 0 & 0 & 0 & -1
\end{array}\right).
\end{equation}
Using the correspondence $\det A_I = \det K^{\sigma,wt}_I$, a representative matrix of $[A]\in \S$ is $A^{\mbox{\tiny RREF}}$ defined in (\ref{eq:A_RREF_ex}).

The network in Figure \ref{fig:example_Gr46_dual}[right] represents its dual point in $Gr^{\mbox{\tiny TNN}}(4,6)$ and is obtained changing the color of all vertices and inverting the face weights. By construction $(K^{\sigma,wt})^T$ is the Kasteleyn matrix of such dual network. Upon transforming $(K^{\sigma,wt})^{T}$ into the block form, 
\begin{equation}\label{eq:ex_Kas_block_dual}
\renewcommand{\arraystretch}{1.5}
  \begin{blockarray}{ccc|cccccc}
    \begin{block}{(ccc|cccccc)}
1 & 0 & 0 & 0 & -1 & 0 & 0 & 0 & 0 \\ 
0 & 1 & 0 & 0 & 0 & 1 & 0 & 0 & 0\\
0 & 0 & 1 & 0 & 0 & 0 & 0 & 1 & 0\\
\cline{1-9}		
0 & 0 & 0 & 0 & t_{11} & 1 & 0 & 0 & 0\\ 
0 & 0 & 0 & 1 & t_{21} & 0 & 0 & 0 & 0\\
0 & 0 & 0 & 0 & 0 & 0 & -1 & -t_{53} & 0\\ 
0 & 0 & 0 & 0 & 0 & t_{72} & 0 & 0 & -1	\\
    \end{block}
  \end{blockarray}
\end{equation}
the point $[\bar A] \in \Sprime \subset Gr^{\mbox{\tiny TNN}} (4,6)$ is identified by the submatrix $\bar A$ in the SE block; its 
reduced row echelon form is
\begin{equation}
\bar{ A}^{\mbox{\tiny RREF}} = 
 \left( \begin{array}{cccccc} 
1 & 0 & 0 & 0 & 0 & -t_{21}(t_{11}t_{72})^{-1} \\ 
\noalign{\medskip}
0 & 1 & 0 & 0 & 0 & (t_{11}t_{72})^{-1} \\ 
\noalign{\medskip}
0 & 0 & 1 & 0 & 0 & -t_{72}^{-1}  \\ 
\noalign{\medskip}
0 & 0 & 0 & 1 & t_{53} & 0 
\end {array}
 \right).
\end{equation}
\end{example}

\section{Kasteleyn systems of relations}\label{sec:kas_rel}

In this Section we characterize Kasteleyn systems of relations on planar bipartite graphs in the disk with black boundary vertices, and discuss their properties for the natural choices $V=\mathbb{C}^{n-k}, \mathbb{R}^n$. Then in Section \ref{sec:comb} we use Kasteleyn system for $V$ the space of polynomials in a finite number of variables to solve a spectral problem in KP theory. If $V=\mathbb{C}^{n-k}$, the solution at the boundary vertices embeds a duality transformation between positroid varieties different from the one of Proposition \ref{prop:dual} since it does not preserve total positivity, whereas in the case $V=\mathbb{R}^{n}$ the solution at the boundary vertices allows to reconstruct the point in the positroid cell $\S$ represented by the network when edge weights are positive. In Section \ref{sec:AGPR} we compare a different variant of Kasteleyn theorem recently proposed in \cite{AGPR} to the present construction.
We use both the representation of systems of relations on bipartite graphs introduced in \cite{AGPR} and that in \cite{Lam2}.

\begin{definition}\label{def:kas_sys_rel_new}\textbf{Kasteleyn system of relations for black boundary vertices}
Let $\mathcal G=(\mathcal B\cup \mathcal W, \mathcal E)$ be a reduced bipartite graph with black boundary vertices and let $\sigma : \mathcal E\mapsto \{ \pm 1\}$ be a Kasteleyn signature satisfying Definition \ref{def:kas_sign}.
For any given edge weighting $t_{bw} : \mathcal E \mapsto \mathbb{C}^*$ we call Kasteleyn the system $(v^{(k)}=\{v_b^{(k)} : b \in \mathcal{B}\}, R_w)$, where:
\begin{enumerate}
\item $v^{(k)}_b$ is an element in the vector space $V$ assigned to the black vertex $b\in 	\mathcal B$;
\item At any given white vertex $w\in \mathcal W$, the variables $v^{(k)}_b$ satisfy the linear relation
\begin{equation}\label{eq:rel_kas_strong}
R_w (v^{(k)}) \, \equiv \, \sum_{b\in \mathcal B} \, \sigma(\overline{bw}) \, t_{bw} \, v^{(k)}_{b} \, \equiv \, \sum_{b\in \mathcal B} \, K^{\sigma,wt}_{bw} \, v^{(k)}_{b} 	\, =	\, 0, 
\end{equation}
where $K^{\sigma,wt}_{bw}$ is the weighted Kasteleyn matrix defined in (\ref{eq:kas_wt_entries}).
\end{enumerate}
\end{definition}

\smallskip

Since the Kasteleyn matrix has full rank $|\mathcal W|$, the kernel of the linear operator $R_w$ has dimension $n-k$. Therefore
a natural choice for the vector space is $V=\mathbb{C}^{n-k}$.  Let us denote $v^{(k)}_{b_i}$ the variables at the boundary vertices ${b_i}$, $i\in [n]$. Then by construction the following statement holds true.

\begin{proposition}\label{prop:kas_sys}
Let $\mathcal G=(\mathcal B\cup \mathcal W, \mathcal E)$ be a reduced bipartite graph with black boundary vertices such that $\mathcal M(\mathcal G)$ is irreducible. Let $\sigma : \mathcal E\mapsto \{ \pm 1\}$ be a Kasteleyn signature on it. 
Then, for any given edge weighting $t_{bw} : \mathcal E \mapsto \mathbb{C}^*$, there exists a choice of $v^{(k)} = \{ v^{(k)}_{b}\in \mathbb{C}^{n-k}, \, b\in \mathcal B\}$, such that
\begin{enumerate}
\item The vectors at the boundary vertices $\{ v^{(k)}_{b_i}, \, i\in [n] \}$ span $\mathbb{C}^{n-k}$;
\item The system of vectors solves the linear system at the white vertices: $R_w(v^{(k)}) =0$, $w\in \mathcal W$.
\end{enumerate}
\end{proposition}

\begin{remark}
\textbf{The case of white boundary vertices} A system of relations may be also introduced if $\mathcal G=(\mathcal B\cup \mathcal W, \mathcal E)$ is a reduced bipartite graph with white boundary vertices. Given a Kasteleyn signature $\sigma$ and an edge weighting $t_{bw}$, we call Kasteleyn the system $(\bar v^{(k)}= \{v_w^{(k)}: w \in \mathcal{W}\}, R_b)$, where: $\bar v^{(k)}_w$ is an element in the vector space $V$ assigned to the white vertex $w\in \mathcal W$, and, at any given black vertex $b\in \mathcal B$, the variables $\bar v^{(k)}_w$ satisfy the linear relation
$R_b (\bar v^{(k)}) \equiv \sum_{w\in \mathcal W} \sigma(e) t_{bw} \bar v^{(k)}_{w} \equiv \sum_{w\in \mathcal W} K^{\sigma,wt}_{bw} \bar v^{(k)}_{w} =0$, 
with $K^{\sigma,wt}_{bw}$ as in (\ref{eq:kas_wt_entries}).
All results in this Section hold with obvious modifications in this case as well after replacing $n-k$ with $k$.
\end{remark}

Next let us classify the solutions to the linear system.

\begin{theorem}\textbf{Kernel of $R_w$ and Euclidean duality between positroid varieties}\label{theo:sol_kas_sys_1}
Let $\mathcal G=(\mathcal B\cup \mathcal W, \mathcal E)$ be a reduced bipartite graph with black boundary vertices and let $\sigma : \mathcal E\mapsto \{ \pm 1\}$ be a Kasteleyn signature on it. Let $t_{bw} : \mathcal E \mapsto \mathbb{C}^*$ be a choice of edge weights and let $[A] \in \Pi_{\mathcal M} \subset Gr(k,n)$ be the point represented by the network $\mathcal N = (\mathcal G, t_{bw})$. 
Let $v^{(k)}$ be a solution to the system of relations described in Proposition \ref{prop:kas_sys}. Then 
the $(n-k)\times n$ matrix $\bar A^{o}$ whose columns are the ordered vectors at the boundary vertices,
\begin{equation}\label{eq:matr_orth}
\bar A^{o} = (v^{(k)}_{b_1}, v^{(k)}_{b_2},\dots, v^{(k)}_{b_n})
\end{equation} 
is orthogonal to $[A]$ in the usual sense: if $A=(A^i_j)$, $i\in [k]$, $j\in [n]$, is a representative matrix of $[A]$, then
\begin{equation}\label{eq:orthog}
A \, \cdot \, (\bar A^o)^T\, = \, \left( \sum_{j=1}^n A^i_j \, (\bar A^o)^l_j \right)_{i\in[k], l\in [n-k]} \, = \, 0 .
\end{equation}
Moreover, if $u^{(k)}$ is another solution fulfilling Proposition \ref{prop:kas_sys} for the given choice of edge weights, then  
the $(n-k)\times n$ matrix $(u^{(k)}_{b_1},u^{(k)}_{b_2},\dots, u^{(k)}_{b_n})\in [\bar A^o]$.

Finally the point $[\bar A^o]$ is the point in the dual positroid variety $\Pi_{\overline{\mathcal M}}$, represented by the dual network
$\overline{\mathcal N} = (	\overline{\mathcal G}, \bar f^{\prime})$ obtained from $\mathcal N$ by the following transformation:
\begin{enumerate}
\item $\overline{\mathcal G}$ is the dual graph to ${\mathcal G}$ obtained from it changing the color of all vertices, boundary vertices included;
\item ${\bar f}^{\prime}$ is the face weighting of $\overline{\mathcal G}$ such that, if $f_i$ is the face weight of $\Omega_i$ in ${\mathcal N}$ then $\bar f^{\prime}_i$ is its weight in $\overline{\mathcal N}$, where
\begin{equation}\label{eq:face_tras_dual}
\resizebox{\textwidth}{!}{$ 
\bar {f}^{\prime}_i=\left\{ \begin{array}{ll}
f_i^{-1}, & \mbox{ if } \Omega_i \mbox{ is an internal face},\\
\noalign{\medskip}
(-1)^{b_{\Omega_i}}f_i^{-1}& \mbox{ if } \Omega_i \mbox{ is an external finite face and } 2b_{\Omega_i} \mbox{ is the number of boundary vertices bounding } \Omega_i,
\end{array}\right.
$}
\end{equation}
\begin{equation}\label{eq:bound_vert}
b_{\Omega_i} = \frac{1}{2} \; \# \{ b_j \mbox{ boundary vertex } \; : \;\; b_j\in \partial \Omega_i\textsl{} \, \}.
\end{equation}
\end{enumerate}
\end{theorem}

(\ref{eq:orthog}) easily follows using the equivalent representation of the Kasteleyn weighted matrix given in (\ref{eq:mat_A}). We prove that $[\bar A^o]$ is the point in $\Pi_{\overline{\mathcal M}}$ represented by the duality relation between networks described above in
 Section \ref{sec:AGPR} using the weak Kasteleyn signature introduced in \cite{AGPR}. 

We remark that if face weights are positive in the initial network, at least one face weight $\bar {f}^{\prime}_i$ is negative. Therefore the transformation between dual positroid varieties described by (\ref{eq:orthog}) does not preserve the total non--negativity property since $[A]\in \S \subset \GTNN$ is mapped to $[\bar A^o] \in \Pi_{\overline{\mathcal M}} \backslash \Sprime \subset Gr(n-k,n) \backslash Gr^{\mbox{\tiny TNN}} (n-k,n)$. 

\smallskip

There is a second interpretation for the system of relations of Definition \ref{def:kas_sys_rel_new}. Indeed we may freely assign quantities to $n-k$ boundary vertices $b_j$, $j\in \bar I$, for some $I\in \mathcal M (\mathcal G)$, and solve the resulting $|\mathcal W|\times |\mathcal W|$ linear system in the $|\mathcal W|$ unknowns $v^{(k)}_b$, $b\not =b_j$, $j\in \bar I$.

\begin{theorem}\textbf{Reconstruction of [A] using the system of relations} \label{theo:sol_kas_sys_2}
Let $\mathcal G=(\mathcal B\cup \mathcal W, \mathcal E)$ be a reduced bipartite graph with black boundary vertices and let $\sigma : \mathcal E\mapsto \{ \pm 1\}$ be a Kasteleyn signature on it. Let $I =\{ 1\le i_1 < i_2 < \cdots < i_k \le n\} \in \mathcal M \equiv \mathcal M (\mathcal G)$ be a base in the positroid of the graph. 
For any given edge weighting $t_{bw} : \mathcal E \mapsto \mathbb{R}^+$, let
$[A] \in \S \subset \GTNN$ be the point represented by the network $\mathcal N = (\mathcal G, t_{bw})$, and let its representative matrix $A$ be the reduced row echelon one with respect to the base $I$.
Let $v^{(k)}_b$ be $n$--row vectors, and let us assign $E_j$, the $j$-th canonical basis vector, to $v^{(k)}_{b_j}$ at the boundary vertex $b_j$, for any $j\in \bar I$,
\begin{equation}\label{eq:rel_sink_kas}
v^{(k)}_{b_j}=E_{j}, \quad\quad j\in \bar I.
\end{equation}
Then the vectors $v^{(k)}_{b} \not = 0$ for all $b\in \mathcal B$, and
$v^{(k)}_{b_{i_r}}$ at the boundary vertices $i_r\in I$ satisfy
\begin{equation}\label{eq:rel_sou_kas}
v^{(k)}_{b_{i_r}} =  E_{i_r} -A[r] ,
\end{equation}
where $E_{i_r}$ is the $i_r$--th vector of the canonical basis and $A[r]$ is the $r$--th row of A.
\end{theorem}

\begin{proof}
(\ref{eq:rel_sou_kas}) and the fact that $v^{(k)}_{b} \not = 0$ for all $b\in \mathcal B$ easily follow using the representation of the Kasteleyn matrix in (\ref{eq:mat_A}).
\end{proof}

\begin{remark}\textbf{The explicit solution to the system of relations in Theorem \ref{theo:sol_kas_sys_2}}
In Theorem \ref{theo:kas_geo_sys} we give the explicit solution $v^{(k)}_{b}\in\mathbb{R}^n$ at all vertices using the relation between Kasteleyn and geometric systems of relations, and the solution to the geometric system of relations in terms of flows and conservative flows proved in \cite{AG4}.
\end{remark}

\begin{remark}\textbf{Interpretations of the linear system}
Theorems \ref{theo:sol_kas_sys_1} and \ref{theo:sol_kas_sys_2} provide two possible uses of the linear system: it provides both a natural parametrization of the open dual variety and the parametrization of the positroid variety equivalent to the boundary measurement map.
\end{remark}

\begin{figure}
  \centering{\includegraphics[width=0.37\textwidth]{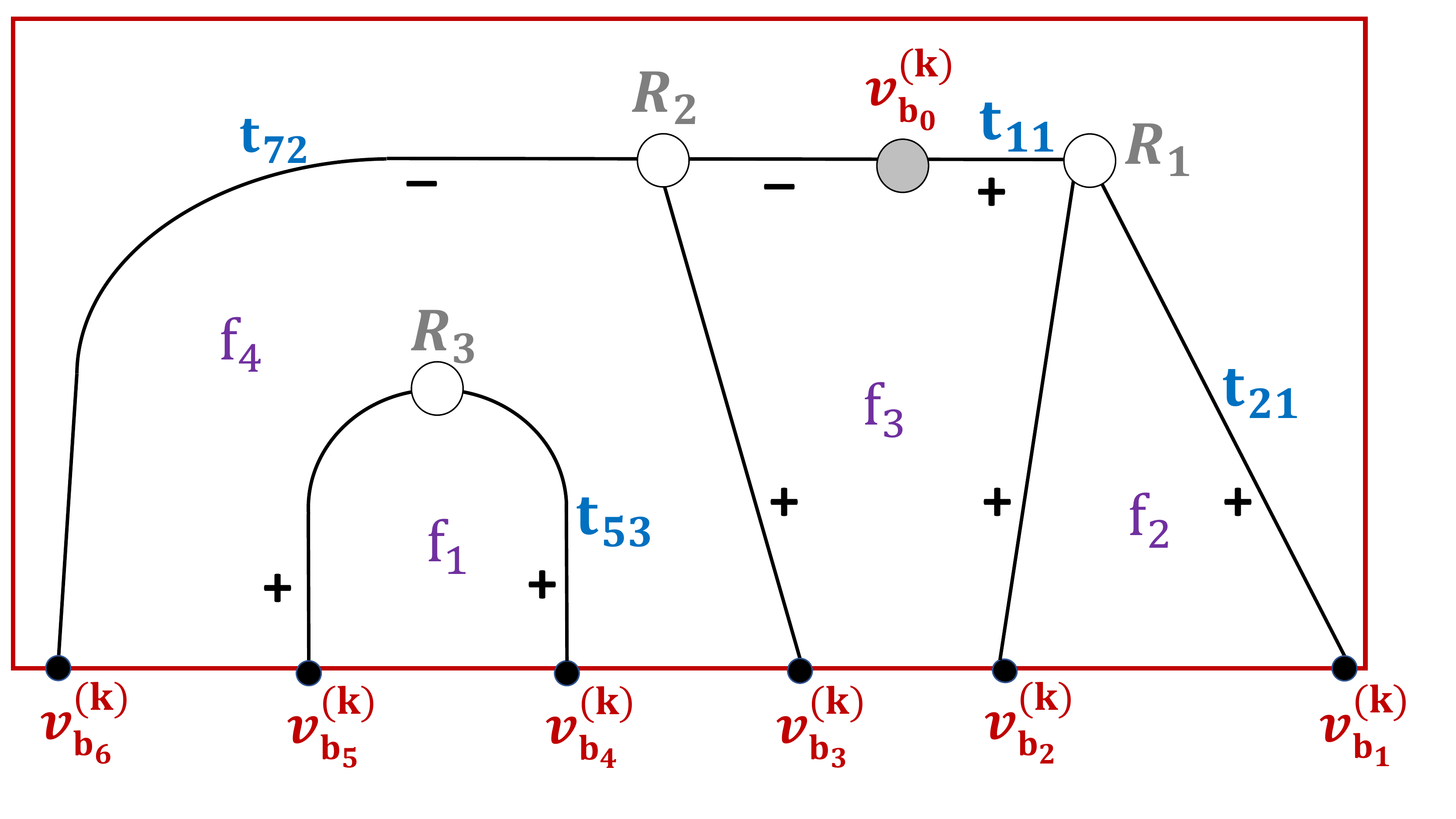}}
	\vspace{-.3 truecm}
  \caption{\small{\sl Kasteleyn system of relation for the network of Example \ref{ex:6}. The relation between face and edge weights is as in (\ref{eq:edge_weight_ex}).}\label{fig:ex_Gr26_sys_rel}}
\end{figure}

\begin{example}\label{ex:6}
Let us solve Kasteleyn system of relation for the network of Example \ref{ex:1} (see Figure \ref{fig:ex_Gr26_sys_rel}).
\begin{enumerate}
\item
First we solve system (\ref{eq:rel_example}) using Theorem \ref{theo:sol_kas_sys_2} so that $v^{(k)}$ are $6$-dimensional row vectors. In such case, we choose a base $I\in \mathcal M$ and assign the canonical basis vectors to $v^{(k)}_{b_j}$, $j\in \bar I$. For instance if $I=(1,4)$ and 
\[
v^{(k)}_{b_2} =(0,1,0,0,0,0), \quad v^{(k)}_{b_3} =(0,0,1,0,0,0), \quad v^{(k)}_{b_5} =(0,0,0,0,1,0), \quad v^{(k)}_{b_6} =(0,0,0,0,0,1),
\]
then we have three linear equations in the unknowns $v^{(k)}_{b_0}, v^{(k)}_{b_1}, v^{(k)}_{b_4}$:
\[
\displaystyle R_j (v^{(k)}) \, =\, 0, \;\,\; j \in [3] \quad \iff \quad 
\left\{ \begin{array}{ll}
t_{11} v^{(k)}_{b_0} + t_{21} v^{(k)}_{b_1} & =	\,  - v^{(k)}_{b_2},\\
\noalign{\medskip}
v^{(k)}_{b_0} & =\, v^{(k)}_{b_3} - v^{(k)}_{b_0} - t_{72} v^{(k)}_{b_6} ,\\
\noalign{\medskip}
t_{53} v^{(k)}_{b_4}& =\, - v^{(k)}_{b_5},
\end{array}\right.
\]
Its solution at the boundary vertices $v^{(k)}_{b_i}$, $i\in I$, satisfies
\[
\begin{array}{c}
v^{(k)}_{b_1} = (0, -t_{21}^{-1},-t_{21}^{-1}t_{11},0,0,t_{21}^{-1}t_{11}t_{72}) =  (1,0,0,0,0,0)-A[1], \\
\noalign{\medskip}
v^{(k)}_{b_4} = (0,0,0,0,0,-t_{53}^{-1},0) = (0,0,0,1,0,0)-A[2],
\end{array}
\]
where $A$ is the reduced row echelon matrix with respect to the base $(1,4)$ represented by the network in the Figure:
\begin{equation}\label{eq:AA}
A = \left( \begin{array}{cccccc} 
1 & t_{21}^{-1} & t_{11} t_{21}^{-1} & 0 & 0 & -t_{21}^{-1}t_{11}t_{72}\\ 
\noalign{\medskip}
0 & 0 & 0 & 1 & t_{53}^{-1} & 0
\end{array}
\right).
\end{equation}
\item
If we use Theorem \ref{theo:sol_kas_sys_1}, the variables $v^{(k)}_b$ are $4$-dimensional column vectors. We have three relations at the white vertices
\begin{equation}\label{eq:rel_example}
\left\{ \begin{array}{l}
\displaystyle R_1 (v^{(k)}) \, = \, t_{11} v^{(k)}_{b_0} + t_{21} v^{(k)}_{b_1} + v^{(k)}_{b_2} \,=\,0,\\
\displaystyle R_2 (v^{(k)}) \, = \, v^{(k)}_{b_3} - v^{(k)}_{b_0} - t_{72} v^{(k)}_{b_6} \,=0\, ,\\
\displaystyle R_3 (v^{(k)}) \, = \, t_{53} v^{(k)}_{b_4} + v^{(k)}_{b_5} \, =\, 0,
\end{array}
\right.
\end{equation}
so that its solution gives the point $[\bar A^o]\in \Pi_{\overline{\mathcal M}}\subset Gr(4,6)$ represented by
\[
\bar A^o = (v^{(k)}_{b_1},v^{(k)}_{b_2},\dots,v^{(k)}_{b_6}) =\left(
\begin {array}{cccccc} 
1 & 0 & 0 & 0 & 0 & t_{21}( t_{11} t_{72})^{-1}\\ 
\noalign{\medskip}
0 & 1 & 0 & 0 & 0 & (t_{11} t_{72})^{-1}\\ 
\noalign{\medskip}
0 & 0 & 1 & 0 & 0 & t_{72}^{-1}\\ 
\noalign{\medskip}
0 & 0 & 0 & 1 & -t_{53} & 0
\end{array}
\right).
\]
and $A \, \cdot \, (\bar A^o)^T=0$ with $A$ is as in (\ref{eq:AA}).
\end{enumerate}
\end{example}

\subsection{Weak Kasteleyn signatures and duality relations in positroid varieties }\label{sec:AGPR} 

In this Section we recall the definition and properties of another variant of Kasteleyn Theorem proposed in \cite{AGPR} naturally connected to Postnikov boundary measurement map. In the following, we call weak the Kasteleyn signature introduced in \cite{AGPR}; this terminology is appropriate since the absolute value of the maximal minors of the Kasteleyn matrix of \cite{AGPR} give the number of almost perfect matchings with prescribed boundary conditions; however they do not share the same sign. Finally we complete the proof of Theorem \ref{theo:sol_kas_sys_1}.

\begin{definition}\textbf{Weak Kasteleyn signature and weak Kasteleyn matrix}\label{def:weak_sign} \cite{AGPR}
Let ${\tilde \sigma}_{bw} = \pm 1$ for each edge of $\mathcal G$. They constitute a weak Kasteleyn signature if the product of the signs along the boundary of each face is 
\[
\prod_{e\in \partial \Omega} {\tilde \sigma}_{e}= \left\{ \begin{array}{ll}
\displaystyle (-1)^{\frac{|\Omega|}{2}+1}, & \mbox{ if } \Omega \mbox{ is an internal face};\\
\noalign{\medskip}
\displaystyle (-1)^{\frac{|\Omega|}{2}+b_{\Omega}+ 1}, & \mbox{ if } \Omega \mbox{ is a finite external face},
\end{array}
\right.
\]
where $\Omega$ is the number of edges bounding $\Omega$ and $b_{\Omega}$ is half the number of boundary vertices bounding $\Omega$.

If one labels the white vertices of $\mathcal G$ from 1 to $|\mathcal W|$, and the black vertices from 1 to $|\mathcal B|$, so that the boundary vertices share the highest labels of their color and are labeled in increasing order clockwise, then the $|\mathcal B|\times |\mathcal W|$ weak Kasteleyn sign matrix ${\tilde K}^{\tilde \sigma}$ associated to such data is  
\begin{equation}\label{def:kas_entries_weak}
{\tilde K}^{{\tilde \sigma}}_{ij} = \left\{ \begin{array}{ll}
{\tilde \sigma} (e) & \mbox{ if  } e=\overline{b_iw_j},\\
0 & \mbox{ otherwise.}
\end{array}
\right.
\end{equation}
If $t: \mathcal E\mapsto \mathbb{C}^*$ is an edge weighting of $\mathcal G$, then the weighted weak Kasteleyn matrix
${\tilde K}^{{\tilde \sigma}, wt}$ is
\begin{equation}\label{eq:kas_wt_entries_weak}
{\tilde K}^{{\tilde \sigma},wt}_{ij} = \left\{ \begin{array}{ll}
{\tilde\sigma} (e) t_e & \mbox{ if } e=\overline{b_iw_j},\\
0 & \mbox{ otherwise.}
\end{array}
\right.
\end{equation}
\end{definition}

Weak Kasteleyn signatures exist \cite{AGPR}. Here we explicitly construct them starting from Kasteleyn signatures satisfying Definition \ref{def:kas_sign}.

\begin{figure}
  \centering{\includegraphics[width=0.37\textwidth]{Example_Gr_26.pdf}
	\hspace{.7 truecm}
	\includegraphics[width=0.37\textwidth]{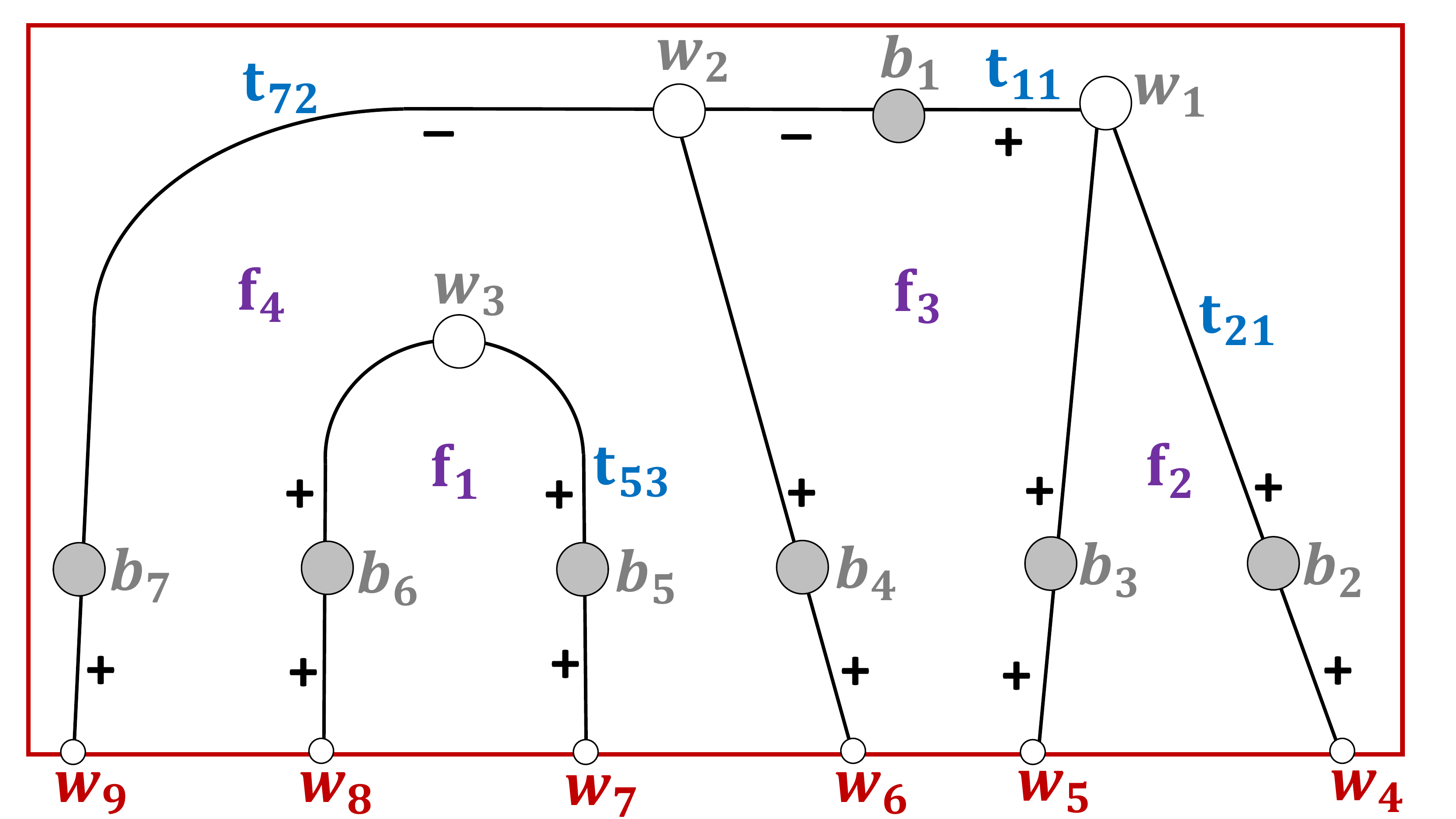}}
	\vspace{-.3 truecm}
  \caption{\small{\sl We illustrate Proposition \ref{prop:weak_K}: the Kasteleyn signature of the graph $\mathcal G$ [left] is transformed into the weak Kasteleyn signature of the graph $\tilde {\mathcal G}$[right].}\label{fig:ex_Gr26_AGPR_1}}
\end{figure}

\begin{proposition}\textbf{Construction of weak Kasteleyn signatures}\label{prop:weak_K}
Let $\mathcal G=(\mathcal B\cup \mathcal W,\mathcal E)$ be a reduced bipartite graph with black boundary vertices representing $\S\subset \GTNN$ and let $\sigma$ be a Kasteleyn signature on it satisfying Definition \ref{def:kas_sign}. Assume that the initial graph $\mathcal G$ has $r=N+k$ white vertices and $N+n$ black vertices. Let $\tilde{\mathcal G}$ be the bipartite graph obtained from $\mathcal G$ 
adding a black boundary--adjacent vertex $b_{N+j}$ next to each boundary vertex, changing the color of the boundary vertices to white, relabeling the boundary vertices $w_{r+j}$, and assigning unit weight to the added edges $e_j = \overline{b_{N+j}w_{r+j}}$, $j\in [n]$. The transformed graph $\tilde {\mathcal G}$ has $N+n$ black vertices and $r+n$ white vertices. Then the signature 
\begin{equation}\label{eq:weak_sign_constr}
{\tilde \sigma}(e) =\left\{ \begin{array}{ll} \sigma(e), & \mbox{ if } e \mbox{ is an edge common to } \mathcal G \mbox{ and } \tilde{\mathcal G};\\
+1 & \mbox{ if } e = e_j, \quad j	\in [n].
\end{array}
\right.
\end{equation}
is a weak Kasteleyn signature on $\tilde{\mathcal G}$.

Finally any weak Kasteleyn signature on $\tilde{\mathcal G}$ is equivalent to that defined in (\ref{eq:weak_sign_constr}) up to the gauge transformation of Definition \ref{def:eq_sign}.
\end{proposition}

\begin{proof} The number of edges and the signatures are the same at each internal face. At the finite external faces the total signatures are the same but the number of edges differ by $2\, b_{\Omega}$.
\end{proof}

The relation of weak Kasteleyn signatures to the boundary measurement map has been proven in \cite{AGPR}. In next Theorem we prove a weaker version of such result assuming that the graphs are related as in Proposition \ref{prop:weak_K}.

\begin{theorem} \textbf{Weak Kasteleyn signatures and the boundary measurement map}\label{prop:AGPR_1}
Let ${\mathcal G}$ and $\tilde{\mathcal G}$ be as in Proposition \ref{prop:weak_K}.
Let $[A]\in \Pi_{\mathcal M}\subset Gr(k,n)$ be the point represented by both networks $(\mathcal G, f)$ and $(\tilde{\mathcal G}, f)$. Let $t\, : \, \mathcal G \mapsto \mathbb{C}^*$ be an edge weighting in the equivalence class represented by $f$, and let $\tilde t$ be the edge weighting on $\tilde{\mathcal G}$ such that ${\tilde t}_e =t_e$ for any edge common to ${\mathcal G}$ and $\tilde{\mathcal G}$, and ${\tilde t}_{e_j}=1$, $j\in [n]$. 

Let $\sigma$ be a Kasteleyn signature on $\mathcal G$ and ${\tilde \sigma}$ be the weak Kasteleyn signature on $\tilde{\mathcal G}$ fulfilling (\ref{eq:weak_sign_constr}). Let ${\tilde K} \equiv {\tilde K}^{{\tilde \sigma},wt}$ be the corresponding weighted weak Kasteleyn matrix. Then, for any $k$--element subset $I=\{1\le i_1< i_2<\dots < i_k\le n\}$ the Pl\"ucker coordinates of $ A$ are
\begin{equation}\label{eq:minor_rel}
\det A_I= (-1)^{\delta(I)} \det {\tilde K}_{\mathcal W\backslash I},
\end{equation}
where $\delta(I) = (i_1-1)+(i_2-2)+\cdots + (i_k-k)$ and ${\tilde K}_{\mathcal W\backslash I}$ is the minor of $\tilde K$ containing all columns except those corresponding to the boundary vertices in $I$, and all rows.
\end{theorem}

\begin{proof}
Under the hypotheses, the weak Kasteleyn matrix ${\tilde K}$ for $(\tilde {\mathcal G}, \tilde \sigma)$ is obtained from the Kasteleyn matrix $K$ for $(\mathcal G,\sigma)$ by adding to its right a block containing the identity matrix
\begin{equation}\label{eq:kas_weak_block}
\renewcommand{\arraystretch}{1.5}
{\tilde K} = 
 \left(\begin{array}{@{}c|c@{}}
  K &
  \begin{matrix}
  0 \\
  \hline
  \mbox{Id}_n
  \end{matrix}
\end{array}\right).
\end{equation}
Therefore, if we use (\ref{eq:mat_A}) to transform ${\tilde K}$ into block form, we get
\begin{equation}\label{eq:tildeK_K}
 \renewcommand{\arraystretch}{1.5}
 \begin{blockarray}{cccc}
	& N & k & n\\
    \begin{block}{c(c|c|c)}
 N\,\, & \,\mbox{Id}_N \, & \, 0 \, & \, 0 \, \\
\cline{2-4}		
n\,\, & *          & \,  A^T \, & \mbox{Id}_n \\
    \end{block}
  \end{blockarray} \, ,
\end{equation}
and it is straightforward to check that
\[
\det A_I = \det (-1)^{\delta(I)} \det {\tilde K}_{\mathcal W\backslash I},
\]
where notations are as in (\ref{eq:minor_rel}). Finally  $[A]\in \Pi_{\mathcal M}$ is the point represented by the network $(\mathcal G, f)$ because of Theorem \ref{cor:equiv_par}.
\end{proof}

\begin{corollary}
Under the hypotheses of Theorem \ref{prop:AGPR_1}, the minors of ${\tilde K}$ are real, but do not share the same sign for any given choice of real positive face weights. Moreover, in case of unit weights 
$|\det {\tilde K}_{\mathcal W\backslash I}|$ equals the number of almost perfect matchings $M$ of $\mathcal G$ such that $\partial M =I$.
\end{corollary}

Next, if one introduces the system of relations for the weak signature, then its kernel provides the point $[A]\in \Pi_{\mathcal M}$ represented by the given network \cite{AGPR}.

\begin{theorem} \textbf{Construction of the representative matrix of the network using the system of relations for the weak Kasteleyn signature} \label{prop:AGPR_2} \cite{AGPR}	
Let $[A] \in \Pi_{\mathcal M} \subset Gr(k,n)$ be the point represented by the network $(\tilde{\mathcal G}, t_{bw})$, where $\tilde{\mathcal G}=(\mathcal B\cup \mathcal W, \mathcal E)$ is a reduced bipartite graph with white boundary vertices, and $t_{bw} : \mathcal E \mapsto \mathbb{C}^*$ an edge weighting. Let
$\tilde \sigma : \mathcal E\mapsto \{ \pm 1\}$ be a weak Kasteleyn signature on $\tilde {\mathcal G}$, and ${\tilde K} \equiv {\tilde K}^{{\tilde \sigma},wt}_{bw}$ be the weighted Kasteleyn matrix for these data. 
Let the weak Kasteleyn system $(\tilde v^{(k)}, {\tilde R}_b)$ be defined as follows:
\begin{enumerate}
\item $\tilde v^{(k)}_w$ is an element in the vector space $\mathbb{C}^k$ assigned to the white vertex $w\in \mathcal W$;
\item At any given black vertex $b\in \mathcal B$, the variables $\tilde v^{(k)}_w$ satisfy the linear relation
\begin{equation}\label{eq:rel_kas_weak}
{\tilde R}_b (\tilde v^{(k)}) \equiv \sum_{w\in \mathcal W} {\tilde\sigma}(e) t_{bw} \tilde v^{(k)}_{w} \equiv \sum_{w\in \mathcal W} {	\tilde K}^{{\tilde \sigma},wt}_{bw} \tilde v^{(k)}_{w} =0.
\end{equation}
\end{enumerate}
Then, there exist solutions to the above system such that the vectors at the boundary vertices $\{ \tilde v^{(k)}_{i}, \, i\in [n] \}$ span $\mathbb{R}^{k}$, and in such case the $k\times n$ matrix $A$ whose columns are the vectors $\tilde v^{(k)}_i$, $i\in [n]$, at the boundary vertices,
\[
A = (\tilde v^{(k)}_1, \tilde v^{(k)}_2,\dots, \tilde v^{(k)}_n),
\]
represents $[A]$. 
\end{theorem}
In particular, if one restricts the map to the real positive octant, the kernels of the corresponding weighted weak Kasteleyn matrices span $\S$. 

Let us now complete the proof of Theorem \ref{theo:sol_kas_sys_1}:

\begin{proof} Let $(\mathcal G, f)$ be the network with black boundary vertices representing $[A]\in \Pi_{\mathcal M}$ in the statement of Theorem \ref{theo:sol_kas_sys_1}, and let $(\tilde G,f)$ be the equivalent network obtained from it adding a black boundary--adjacent vertex next to each boundary vertex, and changing the color of the boundary vertices to white. Let $\sigma$ be a Kasteleyn signature on $\mathcal G$ and $\tilde \sigma$ the weak Kasteleyn signature on $\tilde G$ fulfilling (\ref{eq:weak_sign_constr}).
Then the weak weighted Kasteleyn matrix ${\tilde K}$ on $(\tilde {\mathcal G}, t_{bw}, {\tilde \sigma})$ (see (\ref{def:kas_entries_weak})), where $t_{bw}$ is an edge weighting for the network $(\tilde {\mathcal G}, f)$, has maximal rank by construction and may be put in the block form
\begin{equation}\label{eq:mat_A_AGPR}
 \renewcommand{\arraystretch}{1.5}
 \begin{blockarray}{ccc}
	& N +k & n\\
    \begin{block}{c(c|c)}
 N+k\,\, & \,\mbox{Id}_{N+k} \, & * \\
\cline{2-3}		
n-k\,\, & 0          & \,  \tilde A \,\\
    \end{block}
  \end{blockarray} \, .
\end{equation}
By Theorem \ref{prop:AGPR_2}, $\tilde A$ is orthogonal to $A$:
\begin{equation}\label{eq:ortho_3}
\tilde A . A^T = 0.
\end{equation}
Therefore $[\tilde A]$ and $[\bar A^{o}]$ in Theorem \ref{theo:sol_kas_sys_1} are the same point in $Gr(n-k,n)$ 
\[
[\tilde A] = [\bar A^o].
\]
Next, let ${\mathcal N}^{\prime} = (\tilde {\mathcal G}, f^{\prime})$, be the network such that 
\begin{equation}\label{eq:face_tras}
{f}^{\prime}_i =\left\{ \begin{array}{ll}
f_i, & \mbox{ if } \Omega_i \mbox{ is an internal face},\\
\noalign{\medskip}
(-1)^{b_{\Omega}}f_i& \mbox{ if } \Omega_i \mbox{ is an external finite face with } 2b_{\Omega} \mbox{ boundary vertices},
\end{array}\right.
\end{equation}
and let $\sigma^{\prime}$ be the Kasteleyn signature on $\tilde {\mathcal G}$ satisfying (\ref{eq:kas_sign_transf}), 
\[
\sigma^{\prime}(e) =\left\{ \begin{array}{ll} \sigma(e), & \mbox{ if } e \mbox{ is an edge common to both graphs};\\
(-1)^{j-1}, & \mbox{ if } e = e_j, \quad j	\in [n].
\end{array}\right.
\] 
Then, $t^{\prime}_e = \sigma^{\prime}_e {\tilde \sigma}_e t_e$ is an edge weighting for the network $(\tilde {\mathcal G}, f^{\prime})$, and $\tilde K$ coincides with the Kasteleyn matrix for the data $(\tilde {\mathcal G}, f^{\prime}, \sigma^{\prime})$ (see (\ref{eq:kas_wt_entries})).

A network representing $[\tilde A]=[\bar A^{o}]\in \Pi_{\overline {\mathcal M}}$ is obtained applying Proposition \ref{prop:dual}: start with the network $(	\tilde {\mathcal G}, f^{\prime})$ and apply the duality transformation of Definition \ref{def:dual_net}. Then the network $(\tilde {\mathcal G}^{\prime}, \bar{f^{\prime}})$ represents $[\bar A^{o}]$ where:
\begin{enumerate}
\item $\tilde{\mathcal G}^{\prime}$ is the dual graph to $\tilde {\mathcal G}$ obtained by changing the color of all vertices of $\tilde {\mathcal G}$, boundary vertices included;
\item $\bar{f^{\prime}}_i =(f^{\prime}_i)^{-1} =(-1)^{b_{\Omega_i}} \, (f_i)^{-1}$, where $2b_{\Omega_i}$ is the number of boundary vertices bounding the face $\Omega_i$ and thus it satisfies (\ref{eq:face_tras_dual}).
\end{enumerate}
\end{proof}

\begin{figure}
  \centering{\includegraphics[width=0.37\textwidth]{Example_Gr_26_white_AGPR.pdf}
\hspace{.5 truecm}
	\includegraphics[width=0.37\textwidth]{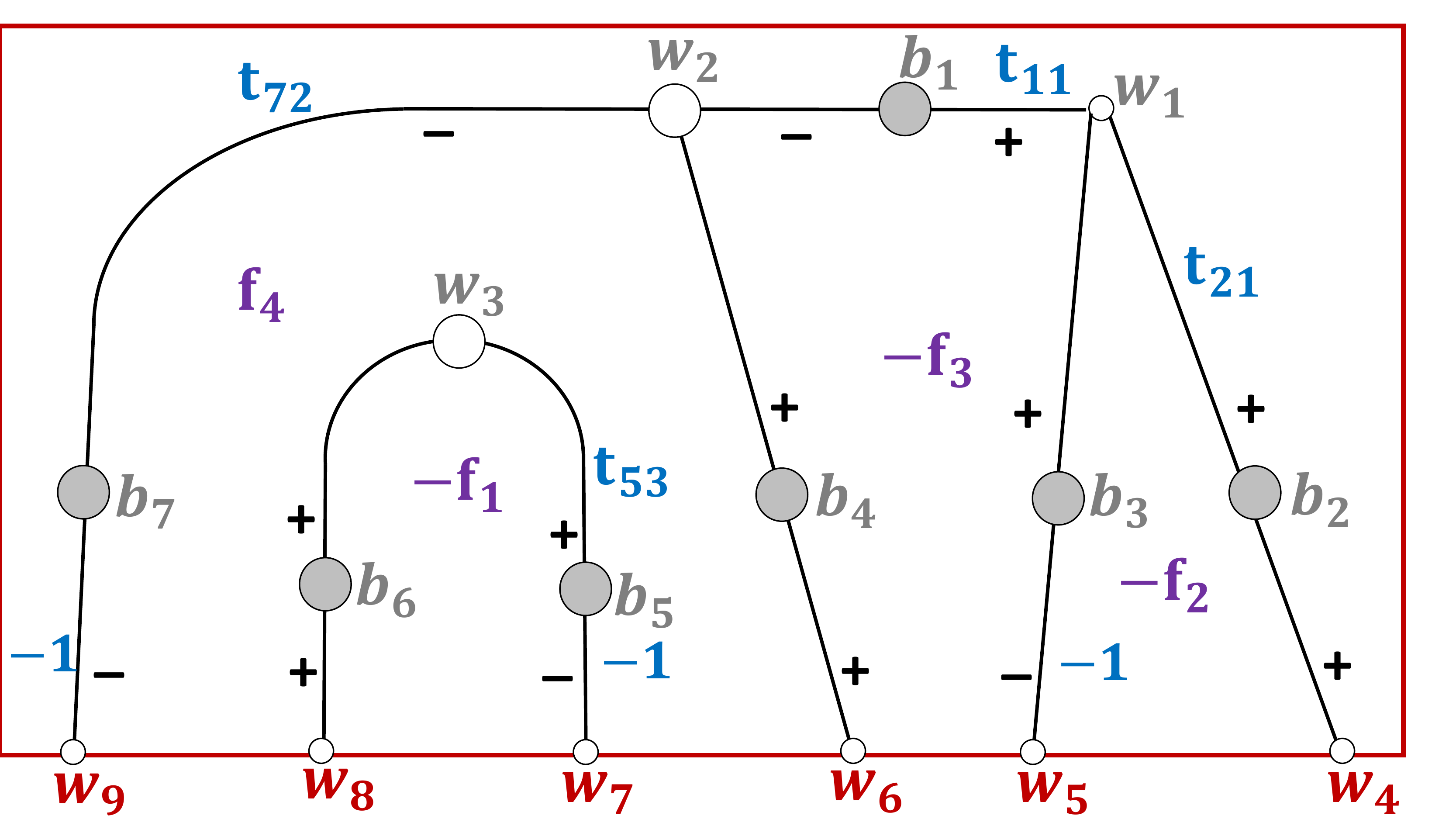}}
	\vspace{-.3 truecm}
  \caption{\small{\sl The signature on the left is weak Kasteleyn, whereas on the right the signature is Kasteleyn. The face weights of the two networks are related by  (\ref{eq:face_tras});  therefore the weighted weak Kasteleyn matrix associated to the network on the left coincides with the weighted Kasteleyn matrix of the network on the right. }\label{fig:ex_Gr26_AGPR_dual}}
\end{figure}

\smallskip

We illustrate Theorem \ref{theo:sol_kas_sys_1} and the duality relations of the two Kasteleyn signatures for the networks in Figure \ref{fig:ex_Gr26_AGPR_dual}. The two networks represent distinct points in the positroid variety $\Pi_{\mathcal M} \subset Gr(2,6)$ indexed by the derangement $\pi = \left( 6,1,2,5,4,3  \right)$: indeed they share the same graph, but have different face weights. Assume $f_i>0$, $i\in [4]$.
The graph on the left has face weights $f= (f_1,f_2,f_3,f_4)$ and has a weak Kasteleyn signature, whereas on the right graph the face weights $f^{\prime} =(f^{\prime}_1,f^{\prime}_2,f^{\prime}_3,f^{\prime}_4)=(-f_1,-f_2,-f_3,f_4)$ satisfy (\ref{eq:face_tras}) and the signature is Kasteleyn. We mark in blue the edge weights different from one on the graphs.

The weighted weak Kasteleyn matrix for the network $(\mathcal G, f)$ and the weighted Kasteleyn matrix for the network $(\mathcal G, f^{\prime})$ coincide and we denote both of them ${\tilde K}$,
\[
{\tilde K} = 
\left( 
\begin {array}{ccccccccc} 
t_{11} & -1 & 0 & 0 & 0 & 0 & 0 & 0 & 0 \\
\noalign{\medskip}
t_{21} & 0 & 0 & 1 & 0 & 0 & 0 & 0 & 0\\ 
\noalign{\medskip}
1 & 0 & 0 & 0 & 1 & 0 & 0 & 0 & 0 \\ 
\noalign{\medskip}
0 & 1 & 0 & 0 & 0 & 1 & 0 & 0 & 0 \\ 
\noalign{\medskip}
0 & 0 & t_{53} & 0 & 0 & 0 & 1 & 0 & 0 \\ 
\noalign{\medskip}
0 & 0 & 1 & 0 & 0 & 0 & 0 & 1 & 0 \\ 
\noalign{\medskip}
0 & -t_{72} & 0 & 0 & 0 & 0 & 0 & 0 & 1 
\end {array} \right)
\]
Applying Theorem \ref{prop:AGPR_2}, a matrix for the point $[A]\in \S \subset Gr^{\mbox{\tiny TNN}}(2,6)$ represented by the network $(G,f)$ is
\[
A = \left(  v_{w_4}, v_{w_5},v_{w_6},v_{w_7},v_{w_8},v_{w_9}  \right) =
 \left( \begin{array}{cccccc} 
1 & t_{21}^{-1} & t_{21}^{-1} t_{11} & 0 & 0 & -t_{21}^{-1} t_{11} t_{72}\\ 
\noalign{\medskip}
0 & 0 & 0 & 1 & t_{53}^{-1} & 0
\end {array}
 \right).
\]
Then Theorem \ref{theo:speyer} provides a matrix for the point $[A^{\prime}]\in \Pi_{\mathcal M} \backslash \S \subset Gr(2,6) \backslash Gr^{\mbox{\tiny TNN}}(2,6)$ represented by the network $(\mathcal G, f^{\prime})$:
\[
A^{\prime} =
 \left( \begin{array}{cccccc} 
1& -t_{21}^{-1} & t_{21}^{-1} t_{11} & 0 & 0 & t_{21}^{-1} t_{11} t_{72}\\ 
\noalign{\medskip}
0 & 0 & 0 & 1 & -t_{53}^{-1} & 0
\end {array}
 \right).
\]
Finally, trasforming ${\tilde K}$ to the block form as in (\ref{eq:mat_A_AGPR}), we obtain the matrix representing $[\tilde A]\in Gr(4,6)$
\[
\tilde A = \left( \begin {array}{cccccc} 
1 & 0 & 0 & 0 & 0 & t_{21}( t_{11} t_{72})^{-1} \\ 
\noalign{\medskip}
0 & 1 & 0 & 0 & 0 & (t_{11} t_{72})^{-1} \\ 
\noalign{\medskip}
0 & 0 & 1 & 0 & 0 & t_{72}^{-1} \\ 
\noalign{\medskip}
0 & 0 & 0 & 1 & -t_{53} & 0
\end {array} \right).
\]
Summarizing, we have the following relations:
\begin{enumerate}
\item $[\tilde A]$ is dual to $[A^{\prime}]$ in the sense of Definition \ref{def:dual_net} and Proposition \ref{prop:dual}, therefore $[\tilde A]\in \Pi_{\overline{\mathcal M}}\backslash \Sprime\subset Gr(4,6) \backslash Gr^{\mbox{\tiny TNN}}(4,6)$ where $\overline{\mathcal M}$ is the dual positroid to $\mathcal M$;
\item $[\tilde A]$ is dual to $[A]$ in the sense of (\ref{eq:ortho_3}), that is $\tilde A A^T=0$;
\item $[\tilde A]\equiv[\bar A^o]$ where $\bar A^o$ is the matrix constructed in Theorem \ref{theo:sol_kas_sys_1}.
\end{enumerate}

\subsection{Lam representation of systems of relations}\label{sec:lam}

In this Section we recall an alternative representation of systems of relations originally introduced in \cite{Lam2} to provide a mathematical framework for the computation of scattering amplitudes on on--shell diagrams for $N=4$ SYM theory \cite{AGP2}. Lam formulation involves variables on directed half--edges, and Kasteleyn system of relations may be equivalently expressed in this form. We shall apply Lam representation in Section \ref{sec:geom_sign} to construct the solutions to Kasteleyn systems of relations.

\begin{definition}\textbf{Lam system of relations} \label{def:lam_sys} \cite{Lam2} Let $(\mathcal G, \mathcal O)$ be a reduced planar bipartite graph in the disk with black boundary vertices and perfect orientation $\mathcal O$. Let $\epsilon: \mathcal E\mapsto \{0,1\}$ be a signature defined on the oriented edges of the graph, and let $t_{uv}$ be the weight of the oriented edge $e=\overrightarrow{uv}$. Then Lam system of relations associated to such signature on the directed network $(\mathcal G, \mathcal O, t_{uv})$ is the following system in the formal half--edge variables $z_{u,e}$:
\begin{enumerate} 
\item For any edge $e=\overrightarrow{uv}$, $z_{u,e} = (-1)^{\epsilon(e)} t_{uv} z_{v,e}$;
\item If $e_i$, $i\in [m]$, are the edges at an $m$-valent white vertex $v$, then $\sum_{i=1}^m z_{v,e_i} =0$;
\item If $e_i$, $i\in [m]$, are the edges at an $m$-valent black vertex $v$, then $z_{v,e_i} =z_{v,e_j}$ for all $i,j\in[m]$.
\end{enumerate}
\end{definition}

Signatures on oriented graphs form equivalence classes with respect to the following gauge equivalence transformation.

\begin{definition}\textbf{Equivalence between edge signatures} \label{def:admiss_sign_gauge}
Let $\epsilon^{(1)}$ and $\epsilon^{(2)}$ be two signatures on the perfectly oriented reduced bipartite graph $(\mathcal G, \mathcal O)$. We say that the two signatures are equivalent if there exists an index $\eta(u) \in \{ 0,1\}$ at each internal vertex $u$ such that 
\begin{equation}\label{eq:equiv_sign}
\epsilon^{(2)} (e) = \left\{\begin{array}{ll}
\epsilon^{(1)} (e) +\eta(u)+\eta(v) \mod 2, & \mbox{ if } e=\overrightarrow{uv} \mbox{ is an internal edge},\\
\epsilon^{(1)} (e) +\eta(u) \mod 2, & \mbox{ if } e=\overrightarrow{uv} \mbox{ is the edge at the boundary sink } v,\\
\epsilon^{(1)} (e) +\eta(v) \mod 2, & \mbox{ if } e=\overrightarrow{uv} \mbox{ is the edge at the boundary source } u.
\end{array}\right.
\end{equation} 
\end{definition}

If a system of relations has full rank for one signature, it has also full rank for any other signature equivalent to it and the solution at the boundary vertices is the same \cite{Lam2}.

Next, let us reformulate Kasteleyn system of relations as a Lam system for half--edge variables.

\begin{proposition}\textbf{Kasteleyn system in Lam form} 	\label{prop:kas_lam_form}
Let $\mathcal G=(\mathcal B\cup \mathcal W, \mathcal E)$ be a reduced planar bipartite graph in the disk with black boundary vertices. Let $\sigma : \mathcal E\mapsto \{ \pm 1\}$ be a Kasteleyn signature on $\mathcal G$, and let $t_{b,w} : \mathcal E \mapsto \mathbb{C}^*$ be an edge weighting on the undirected graph. Let $(v^{(k)}, R_b)$ be a Kasteleyn system of relations for such data as in Definition \ref{def:kas_sys_rel_new} on some vector space $V$. Let $\mathcal O$ be a perfect orientation on $\mathcal G$.
For any edge $e\in \mathcal E$ define 
\begin{equation}\label{eq:sign_lam_kas}
\epsilon^{(k)} (e) = \left\{ \begin{array}{ll}
0 & \quad \mbox{ if } \sigma(e) = 1; \\
1 & \quad \mbox{ if } \sigma(e) = -1,
\end{array}\right.
\end{equation}
and let $z^{(k)}_{u,e}\in V$ be the Lam variables in Definition \ref{def:lam_sys} where we use (\ref{eq:edge_weights}) to pass to the edge weights on the directed graph:
\[
t_{\overrightarrow{uv}} = \left\{ \begin{array}{ll}
t_{bw}, & \mbox{ if } u=w, \; v=b;\\
t_{bw}^{-1},& \mbox{ if } u=b, \; v=w.
\end{array}
\right.
\]
Then Lam system of relations for the signature $\epsilon^{(k)}$ is equivalent to Kasteleyn system of relations for the signature $\sigma$ using the following correspondence:
\begin{enumerate}
\item $z_{b,e}=v^{(k)}_b$ at any given black vertex $b\in \mathcal B$ and for any edge $e$ at $b$;
\item $z_{w,e}=\sigma(e) t_{b,w}v^{(k)}_b$ at any given white vertex $w\in \mathcal W$ and for any edge $e$ at $w$.
\end{enumerate}
Therefore Theorems \ref{theo:sol_kas_sys_1} and \ref{theo:sol_kas_sys_2} may be reformulated for Lam system of relations. In particular, if $V=\mathbb{R}^n$ and we restrict ourselves to equivalence classes of positive edge weights, then the solution to Lam system of relations at the boundary vertices induces Postnikov boundary measurement map.
\end{proposition}

\begin{corollary}\textbf{Invariance of the signature $\epsilon^{(k)}$} Let $\epsilon$ be a signature in the equivalence class of $\epsilon^{(k)}$ on the perfectly oriented graph $(\mathcal G, \mathcal O)$. Then Lam system of relations for the signature $\epsilon$ has full rank for any choice of positive real weights and for any perfect orientation of $\mathcal G$. 

In particular, if  $z_{u,e}= (-1)^{\epsilon (e)} t_{uv} z_{u,e}$ for the oriented edge $e=\overrightarrow{uv}$, then $z_{v,-e}= (-1)^{\epsilon (e)} t_{uv}^{-1} z_{v,-e}$ for the reversed orientation $-e=\overrightarrow{vu}$.
\end{corollary}

\section{The geometric nature of Kasteleyn signatures}\label{sec:geom_sign}

In this Section we investigate the geometric nature of Kasteleyn signatures by reformulating the results in \cite{AG4,AG6} in the setting of perfectly orientable reduced bipartite graphs, and providing the explicit relation between Kasteleyn and geometric signatures (Theorem \ref{theo:main}). For the wider class of planar bicolored graphs used in \cite{AG4}, formula (\ref{eq:kas_vs_geo}) in Theorem \ref{theo:main} defines the natural candidate for a Kasteleyn signature. Therefore we conjecture that geometric signatures  explicitly realize Kasteleyn signatures for the more general variant of Kasteleyn Theorem in \cite{Sp}, and are therefore naturally connected to dimer models also for the more general class of graphs used in \cite{AG4}.

The main consequence of Theorem \ref{theo:main} is Formula (\ref{eq:expl_sol}) in Theorem \ref{theo:kas_geo_sys} which provides the explicit solution of Kasteleyn system of relations in terms of edge flows and conservative flows on the perfectly oriented network.

\subsection{Loop erased walks, edge flows and the geometric construction of vectors on edges}\label{sec:flows}

In \cite{AG4} the components of the edge vectors have been defined through summations over all walks sharing both the initial edge and the final destination to the boundary, and they have been computed explicitly using loop--erased walks \cite{Fom,Law} and flows \cite{Tal2}. Below we recall these results restricting ourselves to reduced planar bipartite networks in the disk with black boundary vertices. 

In \cite{AG4} the many gauge freedoms on the graph are fixed introducing a gauge ray direction. A geometric signature is then assigned to each path in terms of the summation of the local winding number between consecutive edges, and of the number of intersections of its edges with the gauge rays starting at the boundary sources. 

\begin{definition}\label{def:gauge_ray}\textbf{The gauge ray direction $\mathfrak{l}$.}
A gauge ray direction is an oriented direction ${\mathfrak l}$ with the following properties:
\begin{enumerate}
\item The ray with the direction ${\mathfrak l}$ starting at a boundary vertex points inside the disk; 
\item No internal edge is parallel to this direction;
\item All rays starting at boundary vertices do not contain internal vertices.
\end{enumerate}
\end{definition}
The first property may always be satisfied since one may deform the boundary of the disk so that all boundary vertices lie at a common straight interval. 

Gauge ray directions were used in \cite{GSV} to measure the local winding number. 
The local winding number between a pair of consecutive edges $e_k,e_{k+1}$ measures whether or not the triple $(e_k,\mathfrak{l},e_{k+1})$ is ordered, and the sign depends on whether such ordering is clockwise or counterclockwise.

\begin{figure}
  \centering{\includegraphics[width=0.37\textwidth]{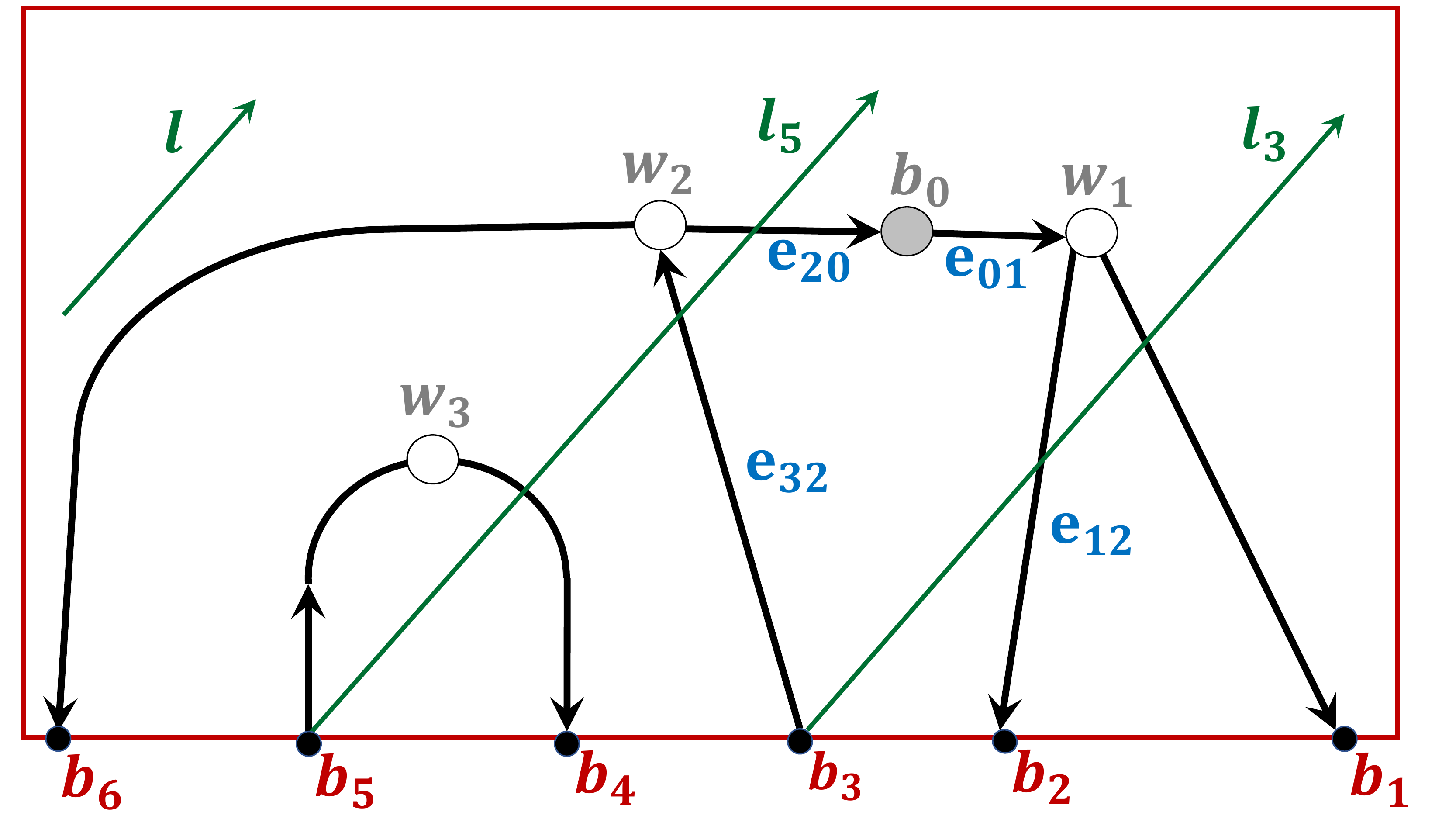}
	\vspace{-.3 truecm}
      \caption{\small{\sl The local winding number at an ordered pair of oriented edges and the intersection number of gauge rays starting at the boundary sources with an oriented edge depend on the choice of gauge ray direction $\mathfrak l$.
					}}\label{fig:wind_int}}
\end{figure}

\begin{definition}\label{def:winding_pair}\textbf{The local winding number at an ordered pair of oriented edges}
For an ordered pair $(e_k,e_{k+1})$ of oriented edges, define
\begin{equation}\label{eq:def_s}
s(e_k,e_{k+1}) = \left\{
\begin{array}{ll}
+1 & \mbox{ if the ordered pair is positively oriented }  \\
0  & \mbox{ if } e_k \mbox{ and } e_{k+1} \mbox{ are parallel }\\
-1 & \mbox{ if the ordered pair is negatively oriented }
\end{array}
\right.
\end{equation}
Then the winding number of the ordered pair $(e_k,e_{k+1})$ with respect to the gauge ray direction $\mathfrak{l}$ is
\begin{equation}\label{eq:def_wind}
wind(e_k,e_{k+1}) = \left\{
\begin{array}{ll}
+1 & \mbox{ if } s(e_k,e_{k+1}) = s(e_k,\mathfrak{l}) = s(\mathfrak{l},e_{k+1}) = 1\\
-1 & \mbox{ if } s(e_k,e_{k+1}) = s(e_k,\mathfrak{l}) = s(\mathfrak{l},e_{k+1}) = -1\\
0  & \mbox{otherwise}.
\end{array}
\right.
\end{equation}
\end{definition}

Next, one counts the intersections of gauge rays with a given path using the rays ${\mathfrak l}_{i_r}$ parallel to $\mathfrak l$ and starting at the boundary source vertices $b_{i_r}$, $r\in[k]$, where $I=\{ i_1< i_2 <\cdots < i_k\}$ is the base of the given perfect orientation (see Figure \ref{fig:wind_int} for an example).

\begin{definition}\label{def:int_number}\textbf{The intersection number at an oriented edge}. 
  Given a perfect orientation ${\mathcal O}(I)$ on the graph and a gauge ray direction, the intersection number
  $\mbox{int}(e)$ for an edge $e$ is the number of intersections of the gauge rays starting at the boundary sources with $e$. For each intersection of $l_s$ with $e$ we assign $+1$ if a pair $({\mathfrak l},e)$ is positively oriented, and $-1$ otherwise. 
\end{definition}

In Figure \ref{fig:wind_int} we illustrate the above definitions. It is straightforward to check that $\mbox{int} (e_{32}) =1$, $\mbox{int} (e_{20}) =\mbox{int} (e_{12}) =-1$, $\mbox{wind}(e_{32}, e_{20})=-1$, $\mbox{wind}(e_{01}, e_{12})=0$.

\smallskip 

Next we adapt the construction of edge vectors in \cite{AG4} to the case of bipartite graphs and introduce a system of edge vectors for all edges $e$ with initial vertex colored black. In \cite{AG4} edge vectors are defined also when the starting vertex is white since the directed graph is not assumed to be bipartite. 

\begin{remark}
In the following, we assign the edge vector at $e=\overrightarrow{bw}$ to its initial black vertex $b$ since there is a unique edge starting at $b$ in a perfectly oriented graph. For this reason we modify the notation of \cite{AG4} and denote edge vectors $E_b$ where $b$ is the initial vertex of the path.
\end{remark}

The $j$--th component of the edge vector $E_b$ is defined through  a (finite or infinite) summation over all walks starting at the given edge $e$ and ending at the same boundary sink $b_j$. 

\begin{definition}\label{def:edge_vector}\textbf{The edge vector $E_b$.} \cite{AG6}
Let $\mathcal G = (\mathcal B\cup \mathcal W, \mathcal E)$ be a perfectly oriented reduced planar bipartite graph in the disk with black boundary vertices. Let $\mathfrak{l}$ be a given gauge ray direction and $t_e$ be a positive edge weighting on $\mathcal G$.
For an oriented edge $e$ with initial black vertex $b$, consider all possible walks (directed paths) $P:b\rightarrow b_{j}$, 
such that the first vertex is $b$ and the end point is the boundary vertex $b_{j}$, $j\in[n]$.
Then the $j$-th component of $E_{b}$ is defined as:
\begin{equation}\label{eq:sum}
\left(E_{b}\right)_{j} = \sum\limits_{P\, :\, e\rightarrow b_{j}} (-1)^{\mbox{wind}(P)+ \mbox{int}(P)} 
wt(P),
\end{equation}
where, for $P= \{ e_1=e, e_2,\dots,e_m\}$,
\begin{enumerate}
\item The \textbf{weight $wt(P)$} is the product of the weights $t_{e_l}$ of all edges $e_l$ in $P$, $w(P)=\prod_{l=1}^m t_{e_l}$;
\item The \textbf{generalized winding number} $\mbox{wind}(P)$ is the sum of the local winding numbers at each ordered pair of its edges
$\mbox{wind}(\mathcal P) = \sum_{k=1}^{m-1} \mbox{wind}(e_k,e_{k+1}),$ 
with $\mbox{wind}(e_k,e_{k+1})$ as in Definition \ref{def:winding_pair}; 
\item $\mbox{int}(\mathcal P)$ is the \textbf{number  of intersections} between the path and the rays ${\mathfrak l}_{i_r}$, $r\in[k]$: $\mbox{int}(\mathcal P) = \sum\limits_{s=1}^m \mbox{int}(e_s)$, where $\mbox{int}(e_s)$ is the number of intersections of gauge rays ${\mathfrak l}_{i_r}$ with $e_s$.
\end{enumerate}
If there is no path from $b$ to $b_{j}$, the $j$--th component of $E_b$ is assigned to be zero.
\end{definition}

In particular, by definition, if $e=\overrightarrow{bb_j}$, with $b_j$ boundary sink, then the vector $E_{b}$ is
\begin{equation}\label{eq:vec_bou_sink}
\left(E_{b}\right)_{k} =  (-1)^{\mbox{int}(e)} w(e) \delta_{jk}.
\end{equation}

Next, (\ref{eq:sum}) is expressed as a summation over equivalence classes of walks using the notion of loop--erased walk. 
Loop--erased walks are extensively used in the study of random walks \cite{Law} and have been reformulated for directed graphs in \cite{Fom} to prove the total non--negativity property of the boundary measurement matrix in terms of infinite summations over edge weights. 
 
\begin{definition}
\label{def:loop-erased-walk}
\textbf{Edge loop-erased walks.} Let $\mathcal G = (\mathcal B\cup \mathcal W, \mathcal E)$ be a perfectly oriented planar bipartite graph in the disk with black boundary vertices.
Let $P$ be a walk given by
$$
V_0=b \stackrel{e=e_1}{\rightarrow} V_1 \stackrel{e_2}{\rightarrow} V_2 \rightarrow \ldots \rightarrow b_j,
$$
where $V_0=b\in \mathcal V$ is the initial black vertex of the edge $e$. The loop-erased part of $P$, denoted $LE(P)$, is defined recursively as 
follows. If $P$ does not pass any edge twice, then $LE(P)=P$.
Otherwise, set $LE(P)=LE(P_0)$, where $P_0$ is obtained from $P$ removing the first
loop it makes; more precisely, given all pairs $l,s$ with $s>l$ and $e_l = e_s$, one chooses the one with the smallest values of $l$ and $s$ and removes the cycle
$$
V_l \stackrel{e_l}{\rightarrow} V_{l+1} \stackrel{e_{l+1}}\rightarrow V_{l+2} \rightarrow \ldots \stackrel{e_{s-1}}\rightarrow V_{s} ,
$$
from $P$.
\end{definition}

\begin{remark}\label{rem:loop}
For initial black vertices, the definition of edge--loop erased walk coincides with the definition of loop--erased walk in \cite{Fom}.
If the initial vertex of the walk is white, the above definition does not coincide with that in \cite{Fom} (see \cite{AG6}).
\end{remark}

With this procedure, to each walk starting at $e=(b_e,w_e)$ and ending at the boundary sink $b_j$, one associates a unique edge loop-erased walk $LE(P)$, where the latter walk is acyclic.
Then one formally reshuffles the summation over 
infinitely many paths starting at $b_e$ and ending at $b_j$ to a summation over the finite number $S$ of equivalence classes  $[LE(P_s)]$, each one consisting of all walks sharing the 
same edge loop-erased walk, $LE(P_s)$, $s\in[ S]$. Let us remark that 
$ \mbox{int}(P)-\mbox{int}(LE(P_s))=0 \,\,(\!\!\!\!\mod 2) $ for any $P\in [LE(P_s)]$, and, moreover, 
$ \mbox{wind}(P)-\mbox{wind}(LE(P_s))$ has the same parity as the number of simple cycles of $P$.
Then, (\ref{eq:sum}) is equivalent to
\begin{equation}\label{eq:sum2}
\left(E_{b}\right)_{j} = \sum\limits_{s=1}^S (-1)^{\mbox{wind}(LE(P_s))+ \mbox{int}(LE(P_s))}\left[ 
\mathop{\sum\limits_{P:e\rightarrow b_{j}}}_{P\in [LE(P_s)] } (-1)^{\mbox{wind}(P)-\mbox{wind}(LE(P_s))  } 
w(P) \right].
\end{equation}

The definitions of flows and conservative flows in \cite{Tal2} have been conveniently adapted in \cite{AG6} to provide the explicit expression of the above summations. The conservative flows are collections of non-intersecting simple loops in the directed graph $\mathcal G$. 
In our setting an edge flow  $F_{e,b_j}$ in ${\mathcal F}_{e,b_j}(\mathcal G)$ is either an edge loop-erased walk $P_{e,b_j}$ starting at the edge $e$ and ending at the boundary sink $b_j$ or the union of $P_{e,b_j}$ with a conservative flow with no common edges with $P_{e,b_j}$. 

\begin{definition}\label{def:cons_flow}\textbf{Conservative flow \cite{Tal2}}. A collection $C$ of distinct edges on $\mathcal G$ is called a conservative flow if 
\begin{enumerate}
\item For each interior vertex $V_d$ in $\mathcal G$ the number of edges of $C$ that arrive at $V_d$ is equal to the number of edges of $C$ that leave from $V_d$;
\item $C$ does not contain edges incident to the boundary.
\end{enumerate}
The set of all conservative flows $C$ in $\mathcal G$ is denoted ${\mathcal C}(\mathcal G)$. 

The \textbf{weight $w(C)$} of the conservative flow $C$ is the product of the weights of all edges in $C$. Unit weight is assigned to
the trivial flow with no edges.
\end{definition}

The following definition of edge flow coincides with the definition of flow in \cite{Tal2} if $e$ starts at a boundary source except for winding and intersection numbers. In \cite{AG6} edge flows are defined also for white vertices.

\begin{definition}\label{def:edge_flow}\textbf{Edge flow at $e$}. \cite{Tal2, AG6} A collection $F_e$ of distinct edges in $\mathcal G$ is called edge flow starting at the edge $e=e_1=\overrightarrow{b_{e_1}w_{e_1}}$ if 
\begin{enumerate}
\item $e\in F_e$;
\item For each interior vertex $v\ne b_{e_1}$, the number of edges of $F_e$ that arrive at $v$ is equal to the number of edges of $F_e$ that leave from $v$;
\item At $b_{e_1}$ the number of edges of $F_e$ that arrive at $b_{e_1}$ is 0;
\item It contains no edge at a boundary source, except possibly $e$ itself.
\end{enumerate}
${\mathcal F}_{e,b_j}(\mathcal G)$ denotes the set of all edge flows $F$ starting at the edge $e$ and ending at the boundary sink $b_j$ in $\mathcal G$. An element $F_{e,b_j}\in{\mathcal F}_{e,b_j}(\mathcal G)$ is the union of an edge loop-erased walk $P_{e,b_j}$ with a conservative flow with no common edges with $P_{e,b_j}$ (this conservative flow may be the trivial one). 
The following triple $(w(F_{e,b_j}),\mbox{wind}(F_{e,b_j}),\mbox{int}(F_{e,b_j}))$ is assigned to $F_{e,b_j}$:
\begin{enumerate} 
\item The \textbf{weight $w(F_{e,b_j})$} is the product of the weights of all edges in $F_{e,b_j}$.
\item The \textbf{generalized winding number} $\mbox{wind}(F_{e,b_j})$ is the winding number of its loop--erased part:
\begin{equation}
\label{eq:wind_flow}
\mbox{wind}(F_{e,b_j}) = \mbox{wind}(P_{e,b_j});
\end{equation}
\item The \textbf{intersection number} $\mbox{int}(F_{e,b_j})$ is the intersection number of its loop--erased part:
\begin{equation}
\label{eq:int_flow}
\mbox{int}(F_{e,b_j}) = \mbox{int}(P_{e,b_j}).
\end{equation}
\end{enumerate}
\end{definition}

\begin{remark} In \cite{AG6} it is proven the following relation between Postnikov topological winding number of a path $\mbox{Wind}(P)$ from the boundary source $b_{i_r}$ and the boundary sink $b_j$, and its generalized winding number $\mbox{wind}(P)$:
\[
\mbox{Wind}(P) + \sigma(i_r,j) = \mbox{wind}(P)+\mbox{int}(P), \mod 2,
\]
where $\sigma(i_r,j)$ is the number of boundary sources strictly between $b_{i_r}$ and $b_j$.
\end{remark}

In \cite{AG6}, Theorem~3.2 in \cite{Tal2} is adapted to prove that the components of $E_b$ defined in (\ref{eq:sum}) are rational expressions in the weights with subtraction-free denominator and an explicit expression for them is provided in terms of edge flows and conservative flows.

\begin{theorem}\label{theo:null}\textbf{Rational representation for the components of vectors $E_b$} \cite{AG6}
Let $({\mathcal G},\mathcal O, \mathfrak l)$ be a reduced planar bipartite graph in the disk representing the irreducible prositroid cell $\S \subset \GTNN$, with perfect orientation $\mathcal O =\mathcal O(I)$, where $I =\{ 1\le i_1< i_2 < \cdots < i_k\le n\}\in \mathcal M$, and gauge ray direction $\mathfrak{l}$. Let $t_e$ be a positive edge weighting on $\mathcal G$ and let $[A]\in \S$ be the point represented by $(\mathcal G, t_e)$.

Then $\left(E_{b}\right)_{j}$, the $j$--th component of the vector $E_b$ in (\ref{eq:sum2}) with $e=\overrightarrow{bw}$, is a rational expression in the edge weights with subtraction-free denominator: 
\begin{equation}
\label{eq:tal_formula}
\left(E_{b}\right)_{j}= \frac{\displaystyle\sum\limits_{F\in {\mathcal F}_{e,b_j}(\mathcal G)} \big(-1\big)^{\mbox{wind}(F)+\mbox{int}(F)}\ wt(F)}{\sum\limits_{C\in {\mathcal C}(\mathcal G)} \ w(C)},
\end{equation}
where notations are as in Definitions~\ref{def:cons_flow} and~\ref{def:edge_flow}. Moreover, if ${\mathcal F}_{e,b_j}(\mathcal G)\not=\emptyset$, then $\left(E_{b}\right)_{j}\not =0$.

In particular, if the graph is acyclically oriented, then the denominator in (\ref{eq:tal_formula}) equals $1$, and the sum $\mbox{wind}(F)+\mbox{int}(F)$ is the same for all $F\in {\mathcal F}_{e,b_j}(\mathcal G)$. Therefore in such case the $j$--th component of $E_b$ is an untrivial polynomial in the edge weights with coefficients sharing equal signs if ${\mathcal F}_{e,b_j}(\mathcal G)\not=\emptyset$. 

Finally, if $b$ is the boundary source $b_{i_r}$, then (\ref{eq:tal_formula}) becomes
\begin{equation}
\label{eq:tal_formula_source}
\left(E_{b_{i_r}}\right)_{j}= \big(-1\big)^{\sigma(i_r,j)}\ \frac{\sum_{F\in {\mathcal F}_{b_{i_r},b_j}(\mathcal G)} \ w(F)}{\sum_{C\in {\mathcal C}(\mathcal G)} \ w(C)} \; =\; (-1)^{\sigma(i_r,j)} M_{ij},
\end{equation}
where $M_{ij}$ is the entry of Postnikov boundary measurement map  with respect to the base $I$ defined in (\ref{eq:bound_meas_map}), and 
$\sigma(i_r,j)$ is the number of elements of $I$ strictly between $i_r$ and $j$.
Therefore, if we assign the $j$-th canonical vectors $E_j$ to the boundary sinks $b_j$, $j\in \bar I$, the edge vectors at the boundary source are
\begin{equation}\label{eq:source_vec}
E_{b_{i_r}} = A[r] - E_{i_r}, \quad \quad r\in [k],
\end{equation}
where $A^r_j = (-1)^{\sigma(i_r,j)} M_{ij}$, $r\in [k]$, $j\in \bar I$, are the entries of the boundary measurement matrix $A$ representing $[A]$, and $E_{i_r}$ is the $i_r$--th canonical basis vector.
\end{theorem}

\begin{remark}
If the graph is reducible, it may happen that $\left(E_{b}\right)_{j}=0$ even if ${\mathcal F}_{e,b_j}(\mathcal G)\not=\emptyset$ \cite{AG4}.
\end{remark}

\begin{figure}
\centering
{\includegraphics[width=0.4\textwidth]{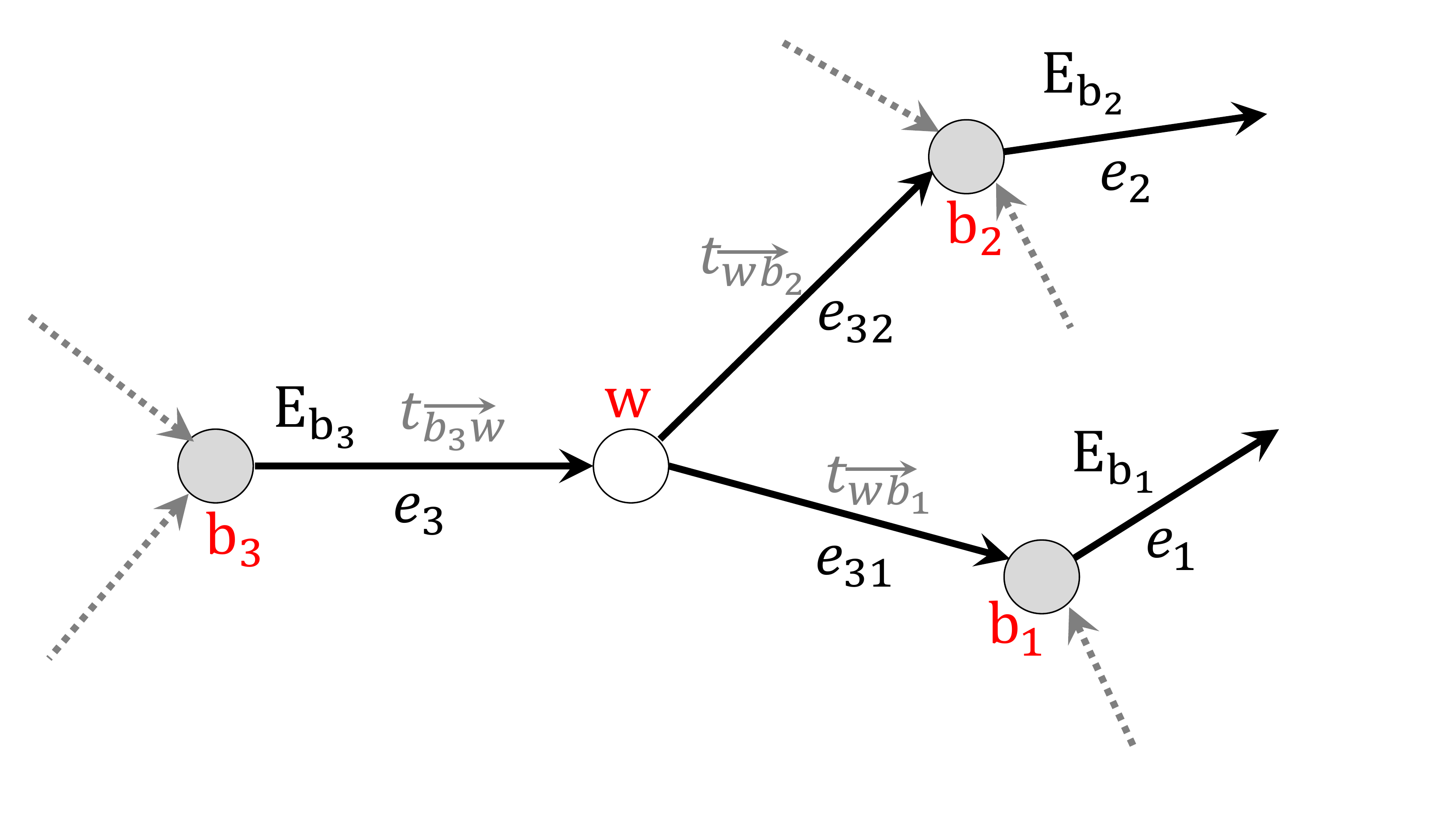}}
\vspace{-.5 truecm}
\caption{\small{\sl The geometric relation at the white vertex $w$ follows from the definition of $E_{b_i}$, $i\in [3]$.}\label{fig:geo_lin}}
\end{figure}

\subsection{The geometric signature and the geometric system of relations}\label{sec:geom}

In this Section we recall the geometric formulation of a signature and of its full rank system of relations following \cite{AG4}, and reformulate it in a form suitable for comparison with Kasteleyn system of relations. Using the notations of Figure \ref{fig:geo_lin}, we remark that a path with initial edge $\overrightarrow{b_3w}$ necessarily passes either through the vertex $b_1$ or $b_2$ (for simplicity we assume $w$ trivalent). Then, by definition, the vectors $E_{b_i}$, $i\in [3]$,  satisfy the following geometric relation
 at the white vertex $w$: 
\begin{equation}\label{eq:triv_white}
\resizebox{\textwidth}{!}{$
\frac{(-1)^{\rm{int}(e_3)}}{t_{\overrightarrow{b_3w}}} E_{b_3} + (-1)^{1+\rm{int}(e_{31})+ \rm{wind}(e_3,e_{31})+\rm{wind}(e_{31},e_{1})} t_{\overrightarrow{wb_1}} E_{b_1} +(-1)^{1+\rm{int}(e_{32})+ \rm{wind}(e_3,e_{32})+\rm{wind}(e_{32},e_{2})} t_{\overrightarrow{wb_2}} E_{b_2} =0.
$}
\end{equation}
If $b_3$ is a boundary source, we write the above formula as
\begin{equation}\label{eq:triv_white_b}
\resizebox{\textwidth}{!}{$
\frac{(-1)^{\rm{int}(e_3)}+1}{t_{\overrightarrow{b_3w}}} \left(-E_{b_3}\right) + (-1)^{1+\rm{int}(e_{31})+ \rm{wind}(e_3,e_{31})+\rm{wind}(e_{31},e_{1})} t_{\overrightarrow{wb_1}} E_{b_1} +(-1)^{1+\rm{int}(e_{32})+ \rm{wind}(e_3,e_{32})+\rm{wind}(e_{32},e_{2})} t_{\overrightarrow{wb_2}} E_{b_2} =0,
$}
\end{equation}
to stress the analogy between (\ref{eq:source_vec}) and the solution to Kasteleyn system of relations in Theorem \ref{theo:sol_kas_sys_2}. 

Following \cite{AG4}, we introduce the following signature on the edges, where notations are consistent with Figure \ref{fig:lin_lam1}. 

\begin{definition}\label{def:geo_sign_expl}\textbf{The geometric signature on $(\mathcal G, \mathcal O, \mathfrak l)$} \cite{AG4}
Let $(\mathcal G , \mathcal O, \mathfrak l)$ be a reduced bipartite graph with black boundary vertices representing the positroid cell $\S \subset \GTNN$, where $\mathcal O$ is a perfect orientation for some base $I\in \mathcal M$ and $\mathfrak l$ is a gauge ray direction. 
We call a signature on $(\mathcal G, O, \mathfrak l)$ geometric if it is equivalent in the sense of Definition \ref{def:admiss_sign_gauge} to the following signature $\epsilon^{(g)}:\mathcal E \mapsto \{0,1\} $ on $(\mathcal G , \mathcal O, \mathfrak l)$: for any edge $e=\overrightarrow{uv}\in \mathcal E$, 
\begin{equation}\label{eq:geo_sign}
\resizebox{\textwidth}{!}{$ 
\epsilon^{(g)}(e) = \left\{ \begin{array}{ll} int(e) \mod 2, & \mbox{ if } v \mbox{ is white};\\
int(e)+1\mod 2, & \mbox{ if } u \mbox{ is a boundary source};\\
1+int(e) +wind(e_1, e) + wind(e,e_5) \mod 2, & \mbox{ if } v \mbox{ is black internal, } e_1 \mbox{ ends at } u \mbox{ and } e_5 \mbox{ starts at } v;\\
1+int(e) +wind(e_1, e) \mod 2, & \mbox{ if } v \mbox{ is a boundary sink and } e_1 \mbox{ ends at } u .																		
\end{array}
\right.
$}
\end{equation}
\end{definition}

In Figure \ref{fig:lam_sign_ex} we compute the geometric signature on the edges of the directed graph of Figure \ref{fig:wind_int}. In \cite{AG4} geometric signatures are defined in the more general case of plabic graphs in the disk with no reference to the color of the boundary vertices. 
\begin{figure}
\centering
{\includegraphics[width=0.4\textwidth]{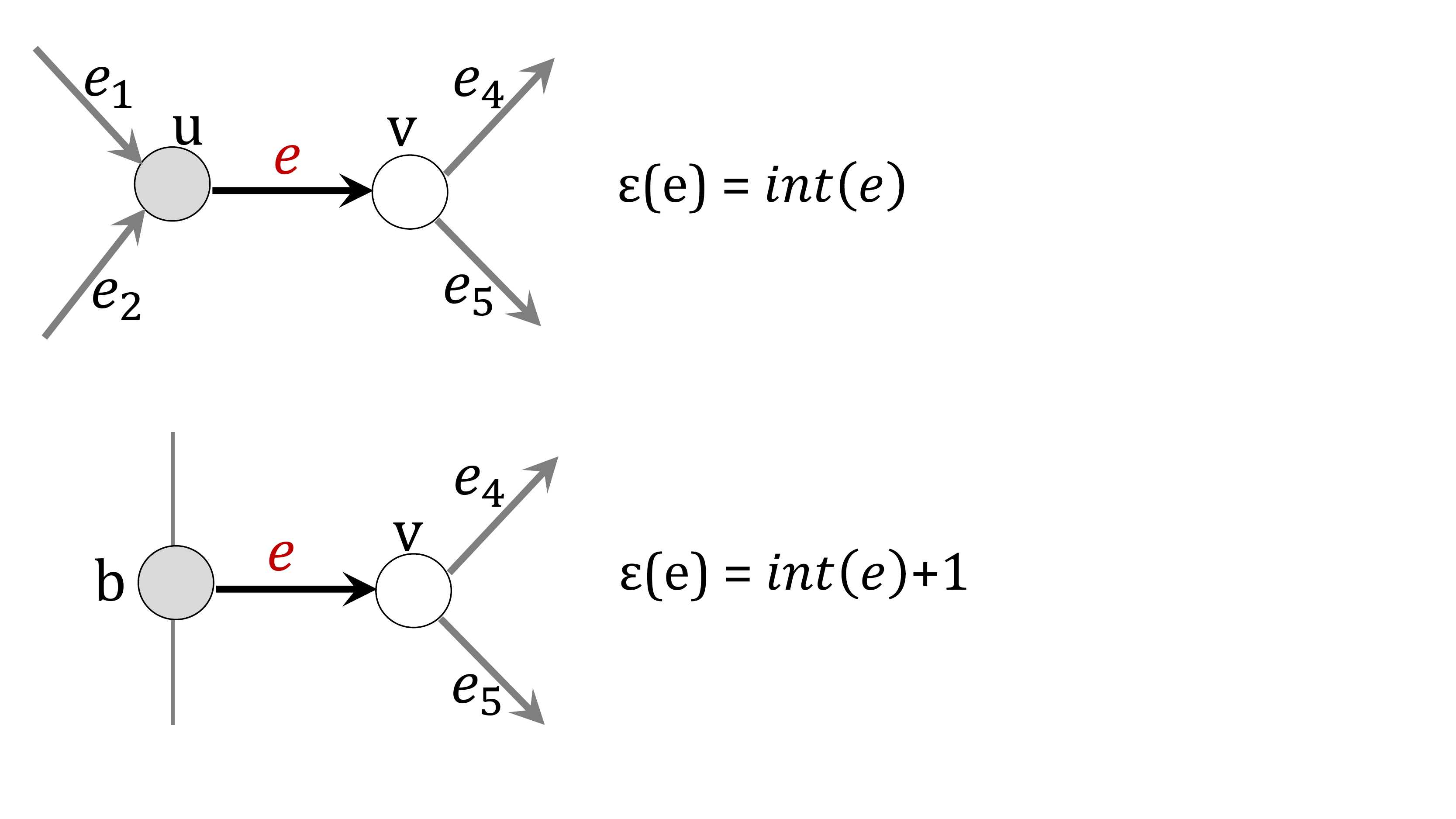}}
\hspace{.5 truecm}
{\includegraphics[width=0.4\textwidth]{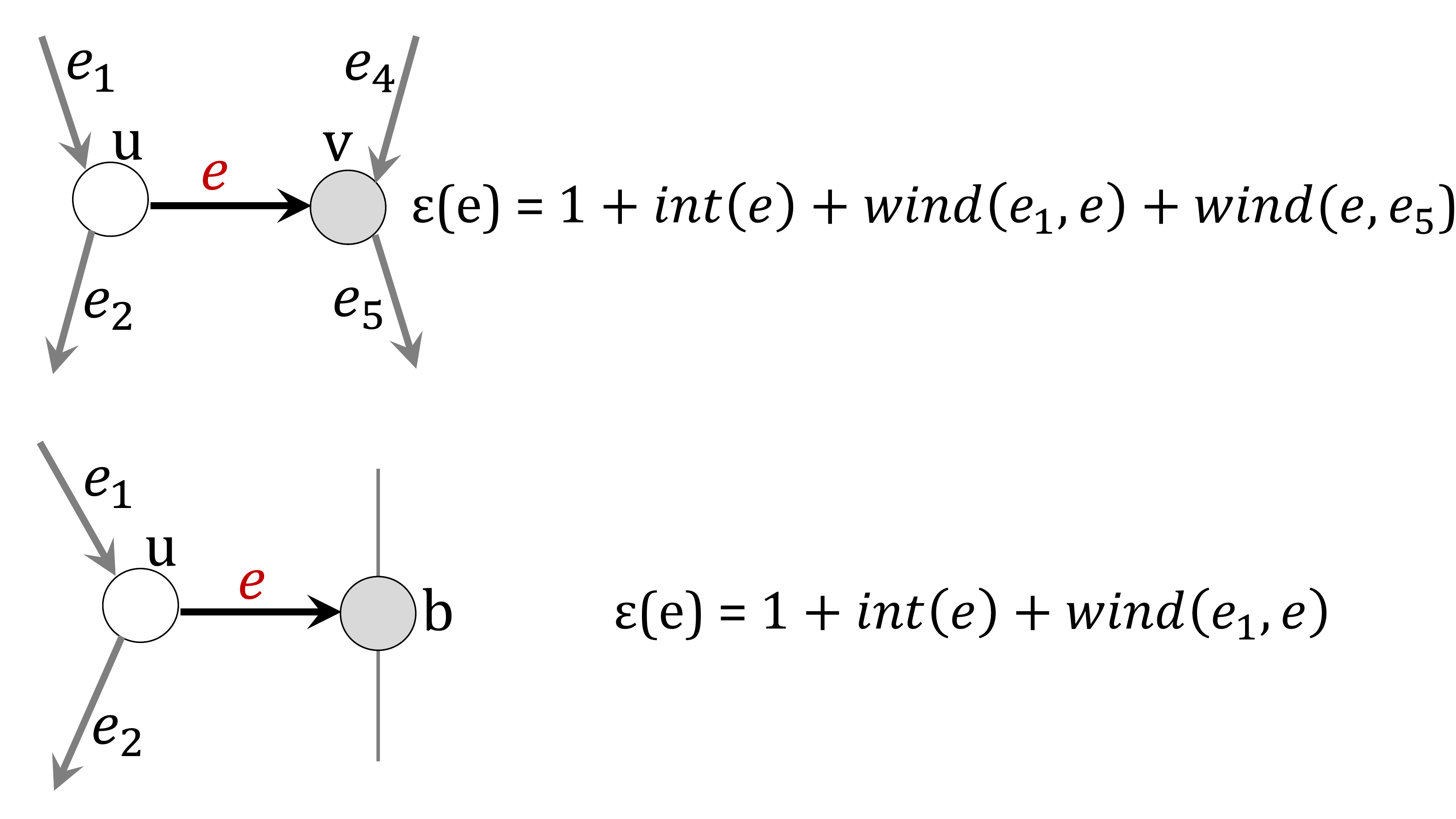}}
\vspace{-.5 truecm}
\caption{\small{\sl The geometric signature at the edges of an oriented bipartite graph with black boundary vertices.}\label{fig:lin_lam1}}
\end{figure}

In \cite{AG4} it is investigated the dependence of the geometric signature on the many gauge freedoms of the graph. 

\begin{theorem}\textbf{The effect of the graph transformations on $\epsilon^{(g)}$}\label{theo:z_orient_gauge} \cite{AG4} 
Let $\mathcal G$ be a reduced planar bipartite graph in the disk. Then the following elementary transformations: change of orientation along a cycle, change of orientation along a simple directed path $\mathcal P$ from the boundary source $b_i$ to the boundary sink $b_j$, change of gauge ray direction or internal vertex position change (which locally modifies winding and intersection numbers) act on the geometric signature as gauge equivalence transformations in the sense of Definition \ref{def:admiss_sign_gauge}.
\end{theorem}

Therefore the equivalence class of $\epsilon^{(g)}$, the geometric signature of Definition \ref{def:geo_sign_expl}, depends only on the graph. Let us denote $\epsilon^{(g)}(\Omega)$ the total contribution of $\epsilon^{(g)}$ at the edges $e$ bounding the face $\Omega$:
\begin{equation}\label{eq:eps_tot_1}
\epsilon^{(g)}(\Omega) = \sum_{e\in\partial\Omega} \epsilon^{(g)}(e).
\end{equation}

\begin{corollary}\textbf{Equivalence class of $\epsilon^{(g)}$ depends only on the graph}\label{cor:equiv} \cite{AG4}
Let $\mathcal G$ be a planar reduced bipartite graph in the disk representing the positroid cell $\S$. 
Let $\epsilon^{(g)}_1$, $\epsilon^{(g)}_2$ be two geometric signatures respectively on $(\mathcal G, \mathcal O_1,\mathfrak l_1)$ and $(\mathcal G, \mathcal O_2,\mathfrak l_2)$, where $\mathcal O_i$ and  $\mathfrak l_i$, $i\in [2]$, respectivaly are perfect orientations of $\mathcal G$ and gauge ray directions. 
Then at each face $\Omega$
\[
\epsilon^{(g)}_2 (\Omega) = \epsilon^{(g)}_1 (\Omega) \mod 2.
\]

Therefore there exists a unique geometric signature $\epsilon^{(g)}$ on $\mathcal G$ modulo the gauge equivalence described in Definition \ref{def:admiss_sign_gauge}. 
\end{corollary}

The solution to Lam system of relations for the geometric signature provides a representation of the edge vectors $E_b$ equivalent to that in Theorem \ref{theo:null} and induces Postnikov boundary measurement map for the class of graphs studied in \cite{AG4}. Below we restate such Theorem in the present setting.

\begin{theorem}\textbf{Lam system for the signature $\epsilon^{(g)}$}\label{theo:lam_exist} \cite{AG4}
Let $\mathcal G$ be a reduced bipartite graph with black boundary vertices representing the positroid cell $\S \subset \GTNN$. Let $ \mathcal O =\mathcal O(I)$ be a perfect orientation for the base $I=\{1\le i_1 < i_2 < \cdots < i_k \le n\}$ and let $\mathfrak{l}$ be a gauge ray direction. Let $\epsilon^{(g)}: \mathcal E\to \{0,1\}$ be the geometric signature defined in (\ref{eq:geo_sign}) for the triple $(\mathcal G, \mathcal O, \mathfrak l)$. Let $z$ be variables in $\mathbb{R}^n$.
Then 
\begin{enumerate}
\item Lam system of relations of Definition \ref{def:lam_sys} for such signature has full rank on $(\mathcal G, \mathcal O, \mathfrak{l},t_{uv})$ for any choice of real positive edge weights $t_{uv}$;
\item For any black vertex $b$, the $j$--th component of the half--edge vector $z^{(g)}_{b,e}$ coincides with the $j$--th component of the edge vector $E_b$ computed in Theorem \ref{theo:null}:
\begin{equation}\label{eq:sol_geo}
(z^{(g)}_{b,e})_j \; = \; \left(E_{b}\right)_{j}\; = \;\frac{\displaystyle\sum\limits_{F\in {\mathcal F}_{e,b_j}(\mathcal G)} \big(-1\big)^{\mbox{wind}(F)+\mbox{int}(F)}\ w(F)}{\sum\limits_{C\in {\mathcal C}(\mathcal G)} \ w(C)},
\end{equation}
\item If we assign the $j$-th basis vector $E_j$ at the half edge vector $z^{(g)}_{b_j}$ at the boundary sink $b_j$,
\begin{equation}\label{eq:rel_sink}
z^{(g)}_{b_j,e}=E_{j}, \quad\quad j\in \bar I,
\end{equation}
then the half--edge vector $z^{(g)}_{b_{i_r},e}$ at the boundary source $i_r\in I$ is
\begin{equation}\label{eq:rel_sou}
z^{(g)}_{b_{i_r},e} = E_{i_r}-A[r],
\end{equation}
where $A[r]$ is the $r$--th row of the boundary measurement matrix represented by the network $(\mathcal G, I,t_{uv})$, and $E_{i_r}$ is the $i_r$--th vector of the canonical basis;
\item If ${\tilde t}_{uv}^{\prime}$ is an edge weighting equivalent to $t_{uv}$ on $(\mathcal G, \mathcal O, \mathfrak{l})$ and ${\tilde z}^{(g)}_{b,e}$ denotes the solution of Lam system of relation for the same signature $\epsilon^{(g)}$ and identical boundary conditions at the boundary sinks, then the solutions of the two systems coincide at the boundary sources $i_r\in I$:
\begin{equation}\label{eq:rel_sou_2}
{\tilde z}^{(g)}_{b_{i_r},e} = z^{(g)}_{b_{i_r},e};
\end{equation}
\item If ${\tilde \epsilon}$ is gauge equivalent to $\epsilon^{(g)}$ and ${\tilde z}_{b,e}$ denotes the solution of Lam system of relation for the new signature ${\tilde \epsilon}$, an equivalent edge weighting to $t_{uv}$ and identical boundary conditions at the boundary sinks on $(\mathcal G, \mathcal O, \mathfrak{l})$, then the solutions of the two systems coincide at the boundary sources $i_r\in I$:
\begin{equation}\label{eq:rel_sou_3}
{\tilde z}_{b_{i_r},e} = z^{(g)}_{b_{i_r},e}.
\end{equation}
\end{enumerate}
\end{theorem}

\begin{remark}
Theorems \ref{theo:lam_exist} and \ref{theo:sol_kas_sys_2} look evidently related (see also Proposition \ref{prop:kas_lam_form}) except for the fact that Kasteleyn signature is defined on an undirected graph whereas the geometric signature is constructed on the same graph but perfectly oriented.
In Theorem \ref{theo:main} we verify that the geometric signature is equivalent to the Kasteleyn signature,
and in Theorem \ref{theo:kas_geo_sys} we provide the relation between $z^{(k)}_{b,u}$ and $z^{(g)}_{b,u}$ at the internal vertices.  
\end{remark}

\begin{figure}
\centering
{\includegraphics[width=0.37\textwidth]{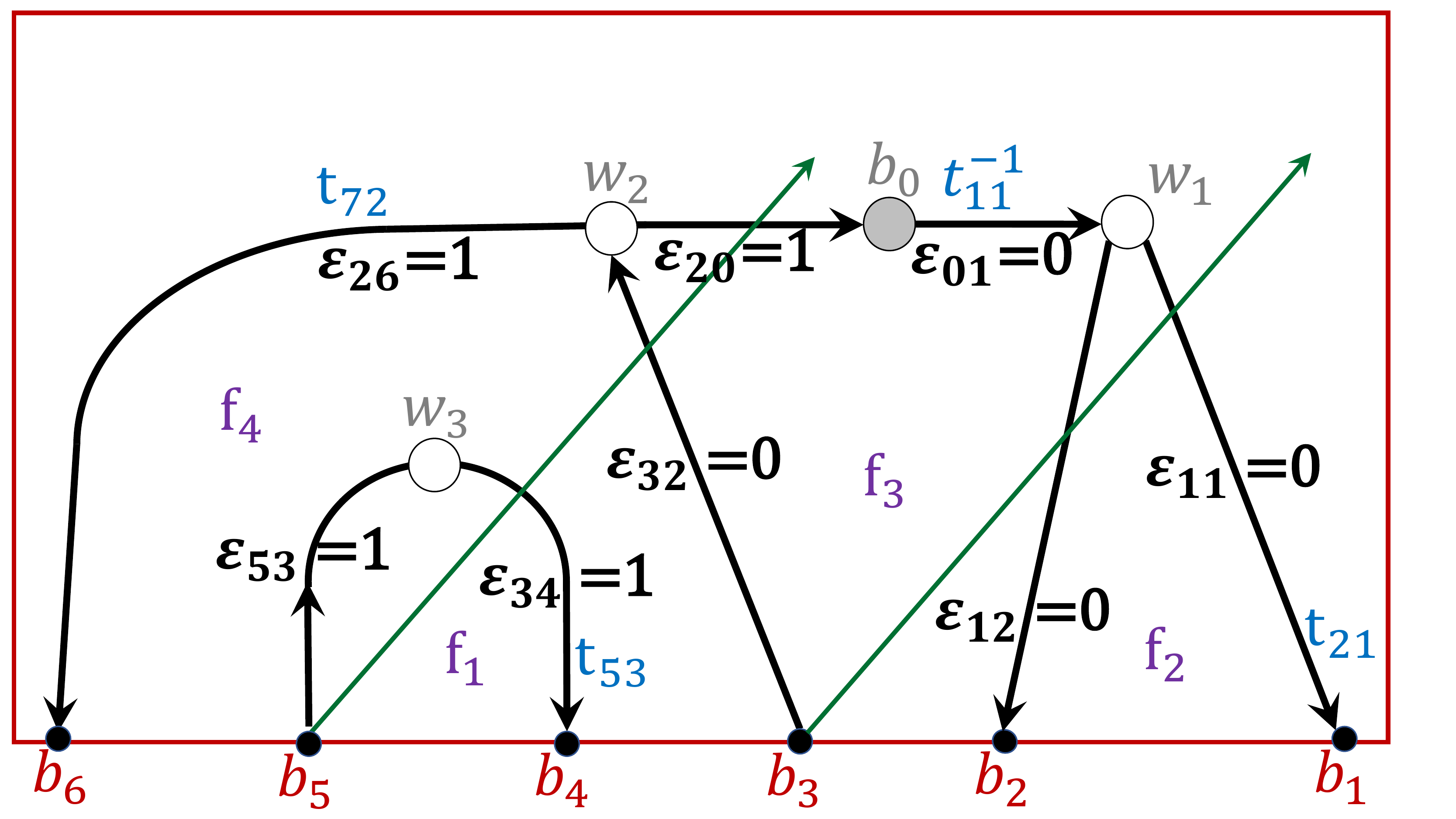}}
\vspace{-.3 truecm}
\caption{\small{\sl The geometric signature of Definition(\ref{def:geo_sign_expl}) for the directed graph of Figure \ref{fig:wind_int}. The directed network represents the same point in $Gr(2,6)$ of Figure \ref{fig:ex_Gr26}.}\label{fig:lam_sign_ex}}
\end{figure}

We end this Section illustrating Theorem \ref{theo:lam_exist} for the example in Figure \ref{fig:lam_sign_ex}.

\begin{example}\label{ex:4}
The network in Figure \ref{fig:lam_sign_ex} is oriented with respect to the base $I=\{3,5\}$ and equivalent to the undirected network of Figure \ref{fig:ex_Gr26} and Example \ref{ex:1}. Lam system of relations for the geometric signature of Definition \ref{def:geo_sign_expl} takes the following form: 
\begin{enumerate}
\item At the oriented edges, the following relations hold:
\[
\begin{array}{ll}
z^{(g)}_{w_1,e_{11}} = (-1)^{\epsilon_{11}}t_{21} z^{(g)}_{b_1,e_{11}} = t_{21} z^{(g)}_{b_1,e_{11}}, 
& \quad z^{(g)}_{w_1,e_{12}} = (-1)^{\epsilon_{12}} z^{(g)}_{b_2,e_{12}} = z^{(g)}_{b_2,e_{12}}, \\	
\noalign{\medskip}
\displaystyle  z^{(g)}_{b_0,e_{01}} = \frac{(-1)^{\epsilon_{01}}}{t_{11}} z^{(g)}_{w_1,e_{01}} = \frac{z^{(g)}_{w_1,e_{01}}}{t_{11}},
& \quad z^{(g)}_{w_2,e_{20}} = (-1)^{\epsilon_{20}} z^{(g)}_{b_0,e_{20}} = -z^{(g)}_{b_1,e_{20}}, \\
\noalign{\medskip}
z^{(g)}_{b_3,e_{32}} = (-1)^{\epsilon_{32}} z^{(g)}_{w_2,e_{32}} = z^{(g)}_{w_2,e_{32}},  
& \quad z^{(g)}_{w_2,e_{26}} = (-1)^{\epsilon_{26}}t_{72} z^{(g)}_{b_6,e_{26}} = -t_{72} \, z^{(g)}_{b_6,e_{26}},\\
\noalign{\medskip}
\displaystyle  z^{(g)}_{w_3,e_{34}} = \frac{(-1)^{\epsilon_{34}}}{t_{53}} z^{(g)}_{b_4,e_{34}} = -\frac{z^{(g)}_{b_4,e_{34}}}{t_{53}}, 
& \quad z^{(g)}_{b_5,e_{53}} = (-1)^{\epsilon_{53}} z^{(g)}_{w_3,e_{53}} = -z^{(g)}_{w_3,e_{53}};
\end{array}
\]
\item At the internal black vertex $b_0$, we have the relation
\[
z^{(g)}_{b_0,e_{20}}= z^{(g)}_{b_0,e_{01}};
\]
\item At the internal white vertices $w_j$, $j\in [3]$, we have the relations
\[
\begin{array}{lll}
z^{(g)}_{w_1,e_{01}}+z^{(g)}_{w_1,e_{11}}+z^{(g)}_{w_1,e_{12}}=0, & \quad z^{(g)}_{w_2,e_{32}}+z^{(g)}_{w_2,e_{20}}+z^{(g)}_{w_2,e_{26}}=0, & \quad z^{(g)}_{w_3,e_{53}}+z^{(g)}_{w_3,e_{34}}=0.
\end{array}
\]
\end{enumerate}
If we assign the canonical basis vectors to the half-edges at the boundary sinks
$z^{(g)}_{{b_1},e_{11}} =(1,0,0,0,0,0)$, $z^{(g)}_{{b_2},e_{12}} =(0,1,0,0,0,0)$, $z^{(g)}_{{b_4},e_{34}} =(0,0,0,1,0,0)$, $z^{(g)}_{{b_6},e_{26}} =(0,0,0,0,0,1)$, then the half--edge vectors at the boundary sources are as expected
\[
\begin{array}{c}
z^{(g)}_{{b_3},e_{32}} = (-t_{21}t_{11}^{-1}, -t_{11}^{-1},0,0,0,t_{72}) = (0,0,1,0,0,0)- A[1], \\
\\
z^{(g)}_{{b_5},e_{53}} = (0,0,0,-t_{53},0,0) =(0,0,0,0,1,0)- A[2],
\end{array}
\]
where
\[
A = \left(
\begin{array}{cccccc}
t_{21}t_{11}^{-1} & t_{11}^{-1} & 1 & 0  &0 &-t_{72}\\
0 & 0 & 0 & t_{53} & 1 & 0
\end{array}
\right)
\]
is Postnikov boundary measurement matrix for the directed network in Figure \ref{fig:lam_sign_ex} with respect to the base $I= \{3, 5\}$.
If we compare this solution to that of Kasteleyn system of relations for the equivalent undirected network (see Example \ref{ex:6} (1)), for the same boundary conditions at the boundary vertices $j\in \bar I$,
$v^{(k)}_{b_1} =z^{(g)}_{{b_1},e_{11}}$, $v^{(k)}_{b_2} =z^{(g)}_{{b_2},e_{12}} $, $v^{(k)}_{b_4} =z^{(g)}_{{b_4},e_{34}}$, $v^{(k)}_{b_6} =z^{(g)}_{{b_6},e_{26}}$, it is straightforward to check that at the internal black vertex 
$b_0$,
\[
v^{(k)}_{b_0} = (-t_{21}t_{11}^{-1}, -t_{11}^{-1},0,0,0,0)=  z^{(g)}_{{b_0},e_{01}}=  z^{(g)}_{{b_0},e_{20}},
\]
and at the boundary vertices $b_i$, $i\in I$
\[
v^{(k)}_{b_3} = z^{(g)}_{{b_3},e_{32}}, \quad\quad
v^{(k)}_{b_5} = z^{(g)}_{{b_5},e_{53}}.
\]
In Theorem \ref{theo:kas_geo_sys} we indeed show that the above relations hold in general.
\end{example}

\subsection{Geometric signatures are Kasteleyn}\label{sec:inv_geo} 

From the characterization of Kasteleyn signatures on reduced bipartite graphs presented in this paper and Theorem \ref{theo:sign_face} in \cite{AG4} there follows the equivalence between geometric and Kasteleyn signatures(Theorem \ref{theo:main}). Therefore explicit solutions of Kasteleyn systems of relations are expressed in terms of flows (Theorem \ref{theo:kas_geo_sys}). In particular,
the equivalence between Postnikov parametrization of positroid cells via the boundary measurement map and Speyer parametrization 
via maximal minors of Kasteleyn weighted matrices (see \cite{Sp} and Theorem \ref{cor:equiv_par}) follows also from the geometric characterization of Kasteleyn signatures (Corollary \ref{cor:pos_sp_equiv}).

\begin{theorem}\textbf{The total geometric signature at faces}\label{theo:sign_face} \cite{AG4}
Let $\mathcal G$ be a planar bicolored graph representing the irreducible positroid cell $\S\subset \GTNN$, and such that, upon fixing a perfect orientation, for any edge of $\mathcal G$ there is a directed path from boundary to boundary containing it. Let $\epsilon^{(g)}$ be a geometric signature of $\mathcal G$. Let $\epsilon^{(g)}(\Omega)= \sum_{e\in\partial\Omega} \epsilon^{(g)}(e)$ be the geometric signature of the face $\Omega$, and let $n_w(\Omega)$ denote the number of internal white vertices bounding $\Omega$. Then
\begin{equation}\label{eq:sign_face_geo}
\epsilon^{(g)}(\Omega) = 
\left\{ \begin{array}{ll}
\displaystyle n_w(\Omega) +1 \quad \mod 2, & \quad \mbox{if } \Omega \mbox{ is a finite face}; \\
\noalign{\medskip}
\displaystyle n_w(\Omega) +k \quad \mod 2, & \quad \mbox{if } \Omega \mbox{ is the infinite face}.
\end{array}
\right.
\end{equation}
\end{theorem}

If $\mathcal G$ is reduced bipartite with black boundary vertices,
\[
n_w(\Omega) = \frac{|\Omega|}{2},
\]
where $|\Omega|$ is the number of edges bounding the face $\Omega$. Therefore for any given geometric signature $\epsilon^{(g)}(e)$ on the reduced bipartite graph $\mathcal G$, $(-1)^{\epsilon^{(g)}(e)}$ is a Kasteleyn signature. Vice versa for any given Kasteleyn signature
$\sigma(e)$ on $\mathcal G$, then
\[
\epsilon (e) = \left\{ \begin{array}{ll}
0, & \mbox{ if } \sigma(e) =1,\\
1, & \mbox{ if } \sigma(e) =-1,
\end{array}\right.
\]
is an element in the equivalence class of $\epsilon^{(g)}$ since it satisfies (\ref{eq:sign_face_geo}).

\begin{theorem}\label{theo:main}\textbf{Equivalence between Kasteleyn and geometric signatures}
Let $\mathcal G = (\mathcal B\cup \mathcal W, \mathcal E)$ be a reduced planar bipartite graph in the disk with black boundary vertices representing the positroid cell $\S\subset \GTNN$. Let $\sigma  :  \mathcal E \mapsto \{ \pm 1\}$ and $\epsilon :  \mathcal E\mapsto \{ 0,1 \}$ be such that 
\begin{equation}\label{eq:sign_geo_kas}
\sigma (e) =  (-1)^{\epsilon (e)}.
\end{equation}
Then $\sigma$ is a Kasteleyn signature on $\mathcal G$ if and only if $\epsilon$ is a geometric signature on $\mathcal G$.

Moreover in such case, for any finite face $\Omega$, its Kasteleyn signature, $\sigma(\Omega)= \prod_{e \in \partial \Omega} \sigma(e)$, and its geometric signature $\epsilon (\Omega)= \sum_{e\in\partial\Omega} \epsilon (e),$ are related as follows:
\begin{equation}\label{eq:kas_vs_geo}
\sigma(\Omega) = (-1)^{\epsilon (\Omega)}.
\end{equation}
\end{theorem}

\begin{conjecture}\label{conj:conj}
In the more general setting of Theorem \ref{theo:sign_face}, the number of internal white vertices bounding $\Omega$, $n_w(\Omega)$, represents the number of relations involving the edges bounding $\Omega$. In such case, formula (\ref{eq:sign_geo_kas}) defines the candidates for Kasteleyn signatures for the graphs considered in \cite{AG4}: A function $\sigma \, : \, \mathcal E \mapsto \{ \pm 1 \}$ is a Kasteleyn signature in the class of graphs defined in \cite{AG4} if $\sigma (\Omega)$, the total signature of the face, fulfils
\begin{equation}\label{eq:kas_new}
\sigma(\Omega) \equiv  \prod_{e \in \partial \Omega} \sigma(e)= (-1)^{n_w(\Omega)+1}, \quad \quad \mbox{ for any finite face } \Omega.
\end{equation}
We conjecture that a signature satisfying (\ref{eq:kas_new}) realizes the variant of Kasteleyn theorem for planar non--bipartite graphs in the disk in \cite{Sp}, {\sl i.e.} pfaffians of the minors of the sign matrix count the number of dimer configurations in the graph involving all internal vertices exactly once with prescribed boundary conditions. We plan to discuss this issue in a different paper.
\end{conjecture}

Finally the following relation holds between geometric and Kasteleyn systems of relation on a given network.

\begin{theorem}\label{theo:kas_geo_sys}\textbf{The solution to Kasteleyn system of relations}
Let $\mathcal N = (\mathcal G, f)$ be a network representing $[A]\in \S\subset \GTNN$, where $\mathcal G=(\mathcal B\cup \mathcal W, \mathcal E)$ is a reduced bipartite graph with black boundary vertices representing the irreducible positroid cell $\S$. Let $I\in \mathcal M$ be fixed.

Let $\epsilon^{(g)} : \mathcal E \mapsto \{ 0,1\}$ be a geometric signature on the directed graph $(\mathcal G, \mathcal O (I))$ and
let $\sigma : \mathcal E\mapsto \{ \pm 1\}$ be the Kasteleyn signature defined by (\ref{eq:sign_geo_kas})
\[
\sigma (e) =  (-1)^{\epsilon^{(g)} (e)}, \quad\quad \forall e\in \mathcal E.
\]
Let $t_{bw} : \mathcal E \mapsto \mathbb{R}^+$ be an edge weighting on the undirected graph $\mathcal G$ representing $\mathcal N$. Let $t_{\overrightarrow{uv}} : \mathcal E \mapsto \mathbb{R}^+$ be the corresponding edge weighting on the directed graph $(\mathcal G, \mathcal O(I))$ fulfilling (\ref{eq:edge_weights})
\[
t_{\overrightarrow{uv}} = \left\{ \begin{array}{ll}
t_{bw}, & \mbox{ if } u=w, \; v=b;\\
t_{bw}^{-1},& \mbox{ if } u=b, \; v=w.
\end{array}
\right.
\]
Let $(z_{b,e},z_{w,e})$ and $(z^{(g)}_{b,e},z^{(g)}_{w,e})$ respectively be Lam system of relations for the Kasteleyn signature
on the undirected network $(\mathcal G, t_{bw})$, and that for the geometric signature on the perfectly oriented network $(\mathcal G, \mathcal O(I), t_{\overrightarrow{uv}})$. Moreover let $(v^{(k)}_b, R_w)$ be the Kasteleyn system of relations associated to $(z_{b,e},z_{w,e})$ in Proposition \ref{prop:kas_lam_form}.

Then the three system of relations have full rank and are equivalent: if we assign the same quantities  $v_j$ at the boundary sinks $b_j$, $j\in \bar I$, to the three systems,
\[
v^{(k)}_{b_j} \, = \, z_{b_j,e}\, =\, z^{(g)}_{b_j,e}\, = \, v_j \, \in \mathbb{R}^n,
\]
then the solution of the system at the black vertices is expressed in terms of edge and conservative flows on the perfectly oriented network $(\mathcal G, \mathcal O(I), t_{\overrightarrow{uv}})$ as in (\ref{eq:sol_geo}):
\begin{equation}\label{eq:expl_sol}
v^{(k)}_{b} \, =\,  z_{b,e} \, =\,  z^{(g)}_{b,e} = \;\frac{\displaystyle\sum\limits_{j\in \bar I} v_j \sum\limits_{F\in {\mathcal F}_{e,b_j}(\mathcal G)} \big(-1\big)^{\mbox{wind}(F)+\mbox{int}(F)}\ w(F)}{\sum\limits_{C\in {\mathcal C}(\mathcal G)} \ w(C)}.
\end{equation}

In particular, if $v_j= E_{j}\in \mathbb R^n$, $j\in \bar I$, are the canonical basis vectors then, at the boundary sources 
\[
v^{(k)}_{b_{i_r}}\, = \, z_{b_{i_r},e}\,= \, z^{(g)}_{b_{i_r},e}\,  = \, E_{i_r} -A[r] , \quad\quad i_r \in I,
\]
where $A[r]$ is the $r$-th row of the matrix in reduced row echelon form with respect to the base $I$ representing $[A]$, and $E_{i_r}$ is the $i_r$--th vector of the canonical basis in $\mathbb R^n$.
\end{theorem}

A consequence of Theorems \ref{theo:lam_exist}, \ref{theo:main} and \ref{theo:kas_geo_sys} is an independent proof that Speyer parametrization of positroid cells using Kasteleyn sign matrices is equivalent to Postnikov boundary measurement map \cite{Sp}:

\begin{corollary} \textbf{Parametrization of positroid cells via Kasteleyn weighted matrices}\label{cor:pos_sp_equiv} 
Let $\S\subset \GTNN$ be given and let $\mathcal G$ be a reduced planar bipartite graph with black boundary vertices representing $\S$. Let $f: {\mathcal G}^* \to \mathbb{R}^+$ be a positive face weighting of $\mathcal G$. Let $\mathcal N =(	\mathcal G, f)$ be the corresponding network and let $[A^{bmm}]\in \S$ be the value of Postnikov boundary measurement map for the network $(G,f)$.

Let $\sigma$ be a Kasteleyn signature for $\mathcal G$, and let $K^{\sigma, wt}$ be a weighted Kasteleyn matrix representing $\mathcal N$. Let $A^{kas}$ be such that for any $k$--element subset $I	\subset [n]$
\[
\det  (A^{kas})_I  = \det K^{\sigma, wt}_I.
\]
Then 
\[
[A^{kas}] = [A^{bmm}], 
\]
and $(K^{\sigma, wt})^T$, the transpose of the Kasteleyn matrix, may be put in block form
\[
\renewcommand{\arraystretch}{1.5}
  \begin{blockarray}{ccc}
	& N & n\\
    \begin{block}{c(c|c)}
N\,\, & \,\mbox{Id}_N \, & *\,\, \\
\cline{2-3}		
k\,\, & 0          & \, A^{kas} \,\,\\
    \end{block}
  \end{blockarray}.
\]
\end{corollary}

\section{Construction of real regular KP divisors using Kasteleyn system of relations}\label{sec:comb} 

The KP hierarchy is the most relevant integrable hierarchy \cite{D, DKN, H, MJD, S}, and contains as special reductions other relevant integrable hierarchies such as the Korteweg de Vries (KdV) and Boussinesq ones. 
In this Section we use Kasteleyn system of relations to construct real regular KP divisors on rational degenerations of $\mathtt M$--curves for KP-II real regular multi--line soliton solutions. The present construction on undirected bipartite graphs has two main advantages: first it unveils the reason why the graph is dual to the spectral curve, second the invariance of the KP divisor is for free. Naturally, the relation between Kasteleyn and geometric signatures implies that, upon fixing the network representing the soliton data, the present construction provides the same algebraic geometric data as in \cite{AG5} where directed graphs and geometric signatures were used. We remark that the class of graphs used in \cite{AG5} is more ample than the present one: therefore, if Conjecture \ref{conj:conj} holds true, a direct invariant construction fo the KP divisor would hold also in the more general case.

\smallskip

The real regular multi--line KP solitons studied in \cite{Mat, Mal, BK, BPPP,CK,KW1,KW2014, A1, AG1,AG2, AG3,AG5} are a family of solutions to the KP-II equation \cite{KP}
\[
(-4u_t+6uu_x+u_{xxx})_x+3u_{yy}=0,
\]
which is the first non trivial member of the KP hierarchy \cite{ZS}.  Before continuing, we recall that there exists another representation of this equation, called KP-I equation. In the following, we always refer to the KP--II representation since the behavior of the real solutions of the two equations is different \cite{DN}. 
Moreover, we use the notation $\vec x$ to denote a finite sequence of KP times $x_j$ where the first three are the independent variables appearing in the KP--II equation:
\[
\vec x = (x_1=x,x_2=y,x_3=t,x_4,x_5,\dots).
\]

Real regular multi--line KP solitons correspond to a well defined reduction of the Sato Grassmannian \cite{S}, and they are parametrized by pairs $(\mathcal K, [A])$, where $\mathcal K$ is a set of $n$ ordered phases 
\[
\mathcal K =\{ \kappa_1<\kappa_2 <\cdots <\kappa_n\},
\]
and $[A]$ is a point in an irreducible positroid cell $[A]\in \S \subset Gr^{\mbox{\tiny TNN}}(k,n)$. In particular a KP-II multi--line soliton solution is real and regular for all real times $(x,y,t)$ if and only if its data are in the totally non--negative part of a real Grassmannian \cite{KW1}. Moreover the combinatorics of totally non--negative Grassmannians has been successfully used in \cite{CK, KW2014} to classify the asymptotic behavior (tropical limit) of this class of KP--II solutions.

The real regular multi--line KP--II soliton solutions are also degenerate finite--gap KP solutions. Krichever
\cite{Kr1,Kr2} showed that finite-gap KP solutions correspond to non special divisors on arbitrary algebraic curves. Dubrovin and Natanzon \cite{DN} then proved that real regular KP--II finite gap solutions correspond to non--special divisors on smooth $\mathtt M$--curves satisfying natural reality and regularity conditions: the degree of the divisor equals the genus of the curve, the essential singularity of the KP wave function belongs to one of the ovals (called infinite) and all other ovals (called finite) contain exactly one divisor point. In \cite{Kr4} Krichever developed, in particular, the direct scattering transform for the real regular parabolic operators associated with the KP spectral problem and proved that the corresponding spectral curves are always $M$-curves, and the divisor points are located in the ovals as in \cite{DN}. In \cite{AG1,AG3,AG5} it was proven that the real regular multi--line KP--II soliton solutions are degenerate finite--gap real regular KP--II solutions by providing an explicit construction of their algebraic geometric data on rational degenerations of $\mathtt M$--curves, and by showing that they satisfy the reality and regularity conditions settled in \cite{DN}.

\smallskip

In the direct spectral problem real regular KP--II multi--line soliton solutions are parametrized by divisors on rational curves. To the solution represented by $(\mathcal K, [A])$ there is associated \cite{Mal}:
\begin{enumerate}
\item A rational spectral curve $\Gamma_0$, with a marked point $P_0$ (essential singularity of the wave function), and a coordinate $\zeta$ such that $\zeta^{-1} (P_0)=0$;
\item $k$ simple poles $\DS\equiv \DS (\vec x_0)=\{ P^{(S)}_r,\  r\in [k] \}$, whose $\zeta$--coordinates are real and bounded, 
\[
\gamma^{(S)}_r \equiv \zeta(P^{(S)}_r)\in [\kappa_1,\kappa_n], \quad\quad r\in [k].
\] 
\end{enumerate}
The following two normalizations for the KP wave function are commonly used in literature: in the Sato normalization the wave function $\tilde \psi(P, \vec x)$ has degree $k$ pole divisor at $P_0$ and zero divisor $\DS (\vec x)$, whereas in Krichever normalization the wave function $\hat \psi (P , \vec x)$ has degree $k$ pole divisor ${\mathcal D}$ such that there is a time $\vec x_0$ and $\mathcal D = \DS (\vec x_0)$. Therefore
one may pass from one normalization to the other through the following relation
\[
\hat \psi (P , \vec x) = \frac{ \tilde \psi (P , \vec x)}{ \tilde \psi (P, \vec x_0)}.
\]

On $\Gamma_0$ the divisor is defined through a Sato dressing transformation
of the vacuum wave function \cite{S}. For this reason $\DS$ was called the Sato divisor in \cite{AG1}. 
For the special class of multi--line solitons, such transformation is represented by a linear differential operator 
\[
\mathfrak D = \partial_x^k -{\mathfrak w}_1 (\vec x) \partial_x^{k-1} -\cdots {\mathfrak w}_k (\vec x),
\]
where $\partial_x \equiv \partial_{x_1} = \frac{\partial}{\partial_{x_1}}$. ${\mathfrak w}_j (\vec x)$, $j\in [k]$, are analytic in the KP-times $\vec t$, fulfil Sato equations, and are such that the kernel of $\mathfrak D$ are $k$ linearly independent solutions to the heat hierarchy of the following form
\[
f_i (\vec x) = \sum_{j=1}^n A^i_j \exp (\theta_j (\vec x)), \quad\quad i \in [k],
\]
where $A= (A^i_j)$ is a representative matrix of $[A]\in \S$ and 
\[
\exp (\theta_j (\vec x)) = \kappa_j x + \kappa_j^2 y + \kappa_j^3 t + \kappa_j^4 x_4 +\cdots.
\]
The multi-line KP soliton solution then takes the form
\[
u(\vec t) = 2 \partial_x {\mathfrak w}_1 (\vec t).
\]
Then the Sato divisor at time $\vec x$ on $\Gamma_0$ is the $k$--tuple
\begin{equation}\label{eq:Sato}
\DS (\vec x) =\{ P^{(S)}_1 (\vec x),\dots , P^{(S)}_k (\vec x)\},
\end{equation} 
whose local coordinates $\zeta (P^{(S)}_l (\vec x) ) = \gamma^{(S)}_l (\vec x)\in [\kappa_1, \kappa_n]$, $l\in [k]$, are the solutions to the characteristic equation $P(\zeta, \vec x)=0$ at time $\vec x$:
\begin{equation}\label{eq:char_eq}
P(\zeta, \vec x) = \zeta^k -{\mathfrak w}_1 (\vec x) \zeta^{k-1} -\cdots {\mathfrak w}_k (\vec x) = \prod_{l=1}^k (\zeta-\gamma^{(S)}_l (\vec x)).
\end{equation}

In the same local coordinates the Sato wave function takes the form
\[
\tilde \psi (\zeta(P), \vec x) = \left( 1 - \frac{{\mathfrak w}_1 (\vec x)}{\zeta} -\cdots -\frac{{\mathfrak w}_k (\vec x)}{\zeta^k}\right) \exp( \zeta x + \zeta^2 y + \zeta^3 t + \zeta^4 x_4 + \cdots),\quad\quad P\in \Gamma_0\backslash \{P_0\}.
\]
In the following we use Krichever normalization and an auxiliary third normalization: $\psi (\zeta , \vec x) = \zeta^k \tilde \psi(\zeta, \vec x)$, so that
\begin{equation}\label{eq:KP_norm_rel}
\hat \psi (\zeta, \vec x) = \frac{\psi (\zeta , \vec x)}{\psi (\zeta , \vec x_0)} = \frac{\zeta^k -{\mathfrak w}_1 (\vec x)\zeta^{k-1} -\cdots -{\mathfrak w}_k (\vec x)}{\zeta^k -{\mathfrak w}_1 (\vec x_0)\zeta^{k-1} -\cdots -{\mathfrak w}_k (\vec x_0) }\; e^{ \displaystyle\zeta (x-x_0) + \zeta^2 (y-y_0) + \zeta^3 (t-t_0) + \cdots},
\end{equation}
and
\begin{equation}\label{eq:tilde_psi_kappa_j}
\psi (\kappa_j, \vec x) = \prod_{l=1}^k (\kappa_j-\gamma^{(S)}_l (\vec x)) \, \exp (\theta_j (\vec x)),\quad\quad j\in [n].
\end{equation}
Therefore
\begin{equation}\label{eq:orthog_2}
0\equiv \mathfrak D f_i (\vec x) \, \equiv \, \sum_{j=1}^n A^i_j \, \mathfrak D e^{\theta_j (\vec x)} \, \equiv \, \sum_{j=1}^n A^i_j  \prod_{l=1}^k (\kappa_j-\gamma^{(S)}_l (\vec x)) \, \exp (\theta_j (\vec x))\, \equiv \, \sum_{j=1}^n A^i_j \, { \psi } (\kappa_j, \vec x), \quad \quad i\in [k], \;\; \forall \vec x,
\end{equation}
that is the vector $(\psi (\kappa_1, \vec x), \dots,\psi(\kappa_n, \vec x) )$ defines a flow in the plane orthogonal to $[A]$ as times $\vec x$ evolve. Then, in view of (\ref{eq:orthog}), such vector is a natural boundary condition for Kasteleyn systems of relations on any network $(\mathcal G, t_{b,w})$ representing $[A]$.

\begin{lemma}\textbf{The wave function at the marked points and Grassmann duality}\label{lem:psi_kappa}
Let $(\mathcal K, [A])$ be real regular soliton data with $[A] \in \S \subset \GTNN$, where $\S$ is an irreducible positroid cell, and let $ \psi(\kappa_j, \vec x)$ be as in (\ref{eq:tilde_psi_kappa_j}), {\sl i.e.} the value of the KP--II wave function for the given soliton data at the phase $\kappa_j$, $j \in [n]$, and at the KP time $\vec x$.
Let $[\bar A^o ] \in \Pi_{\overline{\mathcal M}} \subset Gr(n-k,n)$ be the dual point to $[A]$ satisfying (\ref{eq:orthog}).
Then there exist untrivial analytic functions $c_r (\vec x)$, $r\in [n-k]$, such that
\begin{equation}\label{eq:wave_orth}
\psi (\kappa_j , \vec x) = \sum_{r=1}^{n-k} c_r (\vec x) (\bar A^o)^r_j,\quad\quad  j\in [n], \;\; \forall \vec x.
\end{equation}
\end{lemma}

The relevance of Lemma \ref{lem:psi_kappa} becomes manifest in connection with the solution to the inverse spectral problem, which consists in the reconstruction of the KP-II soliton solution from its divisor on the spectral curve at a fixed time $\vec x_0$. Indeed, the mismatch between the dimension of $Gr^{\mbox{\tiny TNN}}(k,n)$ and that of the variety of Sato divisors implies that generically the Sato divisor is not sufficient to determine the corresponding KP-II solution. 

In \cite{AG1, AG3,AG5} a completion of the Sato algebraic--geometric data has been proposed based on the degenerate finite gap theory on reducible curves introduced in \cite{Kr3}. More precisely:
\begin{enumerate}
\item In \cite{AG1,AG3} the dual graph of the reducible spectral curve $\Gamma$ is the Le--graph representing the soliton data, $\Gamma_0$ is identified with the boundary of the disk, and the divisor on $\Gamma$ is constructed through a recursion. In \cite{AG5}, a larger class of networks representing the given soliton solution is used and the real regular divisor is constructed using Lam system of relations \cite{Lam2} for the geometric signatures introduced in \cite{AG4};
\item In \cite{AG3,AG5}, the reality and regularity properties of the KP--II divisor settled in \cite{DN} follow from the combinatorics of $\GTNN$, whereas in \cite{AG1} classical total positivity was used for soliton data in $\GTP$;
\item The independence of the divisor from the gauge freedoms of the chosen network (perfect orientation, gauge ray direction, weight gauge, gauge freedom of the position of internal vertices) follows from the transformation properties of geometric signatures and was proven in \cite{AG5}.
\end{enumerate}
In the present setting, the comparison between Theorem \ref{theo:sol_kas_sys_1} and Lemma \ref{lem:psi_kappa} makes evident that it is just natural to choose a reducible spectral curve whose dual graph represents the soliton data $[A]$ and use Kasteleyn system of relations to extend $ \psi$ to such augmented curve. 
Moreover, the invariance of the KP divisor is automatically guaranteed by the properties of Kasteleyn system of relations on undirected graphs.

\begin{figure}
  \centering
  \includegraphics[width=0.37\textwidth]{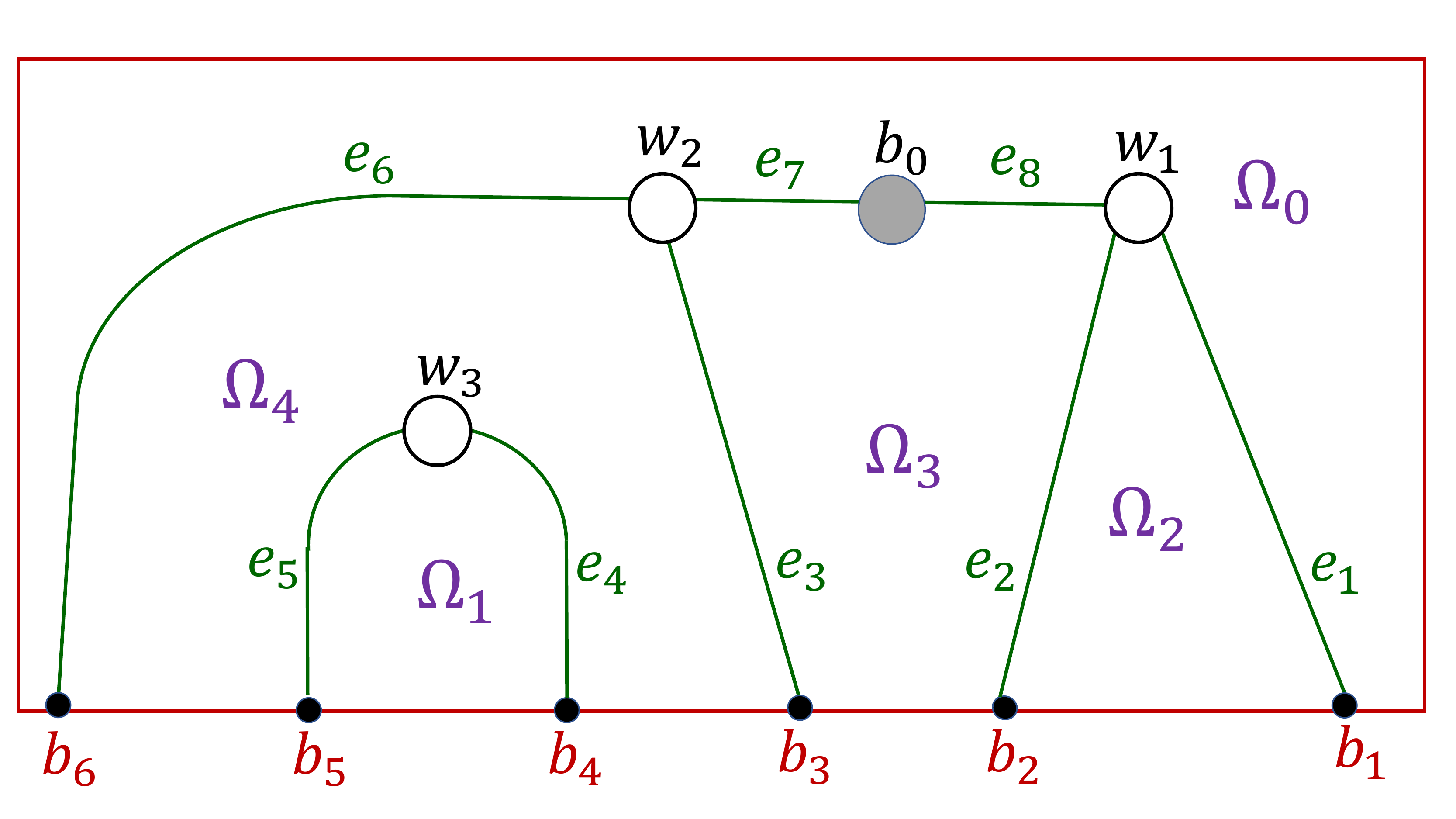}
	\hspace{.5 truecm}
\includegraphics[width=0.37\textwidth]{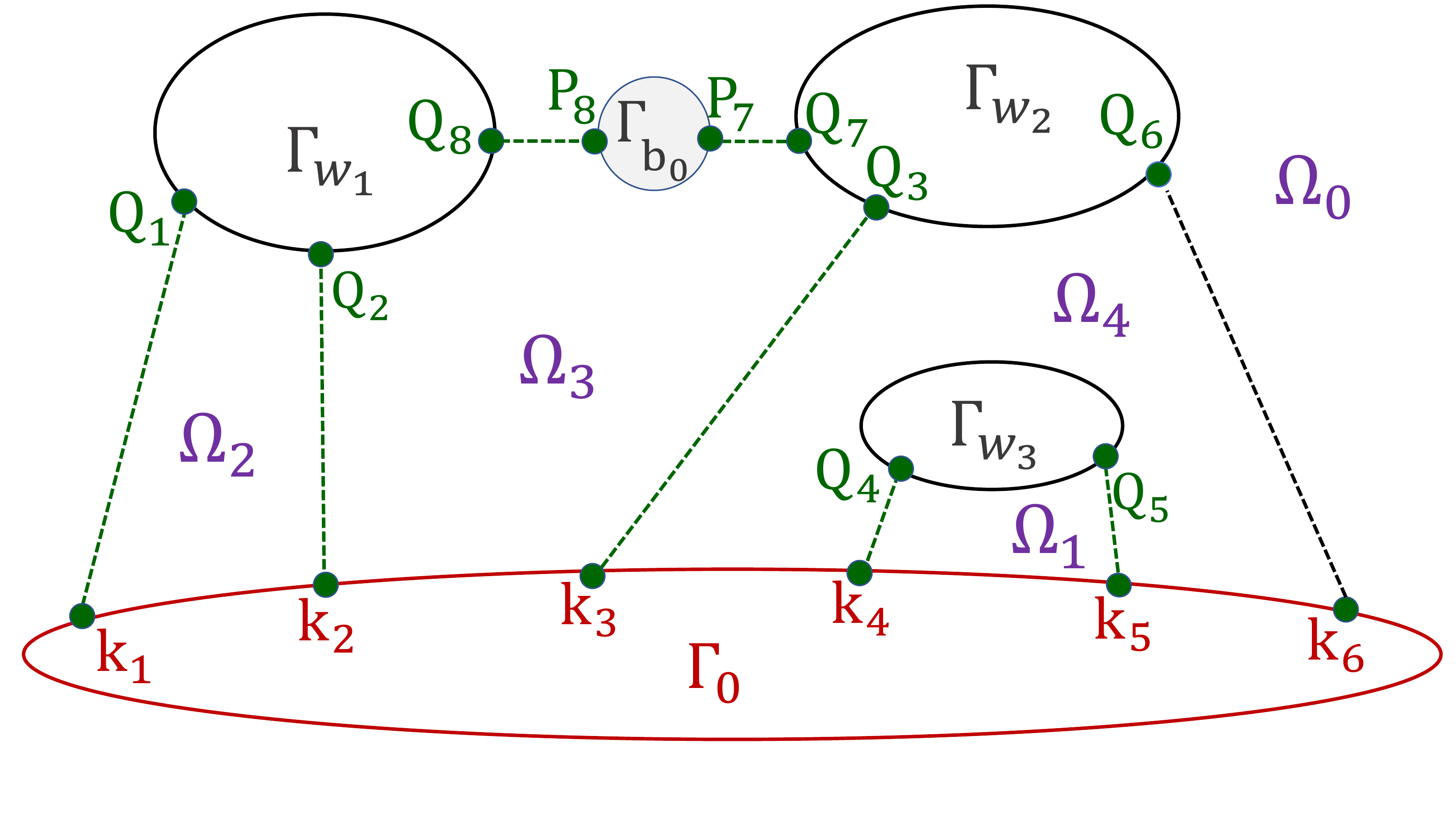}
\vspace{-.3 truecm}
  \caption{\small{\sl The correspondence between graphs [left] and the real part of $\mathtt M$--curves [right] under the assumption that the curve is constructed reflecting the graph w.r.t. a vertical ray. Objects paired by the duality relation between the graph and the curve share the same color: internal vertices $w_i, b_l$ correspond to rational components $\Gamma_{w_i}, \Gamma_{b_l}$; the boundary of the disk is the rational component $\Gamma_0$; faces $\Omega_r$ correspond to ovals denoted with the same symbol; edges $e_s$ joining internal vertices to double points $P_s,Q_s$ (dotted lines mark the gluing in the Figure on the right); and edges $e_s$ joining internal vertices to the boundary vertices $b_j$ are double points $Q_j, \kappa_j$.}}
	\label{fig:curve_graph}
\end{figure}

\textbf{The reducible spectral curve $\Gamma$}: Given the soliton data $(\mathcal K, [A])$, $[A]\in \S \subset \GTNN$, one fixes $\mathcal G$, a reduced planar bipartite graph in the disk with black boundary vertices representing $\S$. To obtain a universal curve, \textbf{we assume that internal vertices are either bivalent or trivalent}. As in \cite{AG3,AG5}, $\Gamma$ is then the reducible curve with dual graph $\mathcal G$: $\Gamma_0$ is the rational component represented by the boundary of the disk, the boundary vertices $b_j$ correspond to the phases $\kappa_j$, $j\in [n]$, the internal vertices are rational components, the edges are the double points at which the rational components of $\Gamma$ are connected, whereas the faces are the ovals of the $\mathtt M$--curve (see Table \ref{table:LeG} and Figure \ref{fig:curve_graph}). 
\smallskip
\begin{table}%[H]
\caption{\small The graph ${\mathcal G}$  vs the reducible rational curve $\Gamma$} 
\centering
\begin{tabular}{|c|c|}
\hline\hline
$\mathcal G$ & $\Gamma$ \\[0.5ex]
\hline
Boundary of disk & Copy of $\mathbb{CP}^1$ denoted $\Gamma_0$ \\
Boundary vertex $b_l$             & Marked point $\kappa_l$ on  $\Gamma_0$\\
Black vertex   $b$                & Copy of $\mathbb{CP}^1$ denoted $\Gamma_{b}$\\
White vertex   $w$                & Copy of $\mathbb{CP}^1$ denoted $\Gamma_{w}$\\
Internal Edge                     & Double point\\
Face                              & Oval\\
Infinite face                     & Infinite oval $\Omega_0$ \\ [1ex]
\hline
\end{tabular}
\label{table:LeG}
\end{table}
\begin{proposition} \cite{AG3} $\Gamma$ is a rational degeneration of a smooth $\mathtt M$--curve of topological genus equal to the dimension of the positroid cell $\S$.  
\end{proposition}

We then extend the wave function $\hat \psi$ from $\Gamma_0$ meromorphically to $\Gamma\backslash \{P_0\}$ with the constraint that it takes equal values at each pair of double points for all times $\vec x$.
Lemma \ref{lem:psi_kappa} suggests a natural way to extend $\hat \psi$ to the double points of $\Gamma$, by fixing a Kasteleyn signature on $\mathcal G$, defining
$v^{(k)}_{b_j} (\vec x)\equiv \psi (\kappa_j, \vec x)$ at the boundary vertices, and assigning the solution to Kasteleyn system of relations of Definition (\ref{def:kas_sys_rel_new}) to the corresponding double points in $\Gamma$. 
Then the relations at the internal white vertices set the degree of the meromorphic extension of the normalized wave function on the corresponding rational components of $\Gamma$ and fully characterize its divisor structure. 

\begin{proposition}\label{prop:KP_double}
Let $(\mathcal K, [A])$ be given soliton data where $\mathcal K =\{ \kappa_1 < \cdots < \kappa_n\}$ and $[A]\in \S$. Let $\psi (\kappa_j, \vec x)$, $j\in [n]$, be as in (\ref{eq:tilde_psi_kappa_j}).
Let $\mathcal G=(\mathcal B\cup \mathcal W, \mathcal E)$ be a given planar bipartite reduced graph in the disk with black boundary vertices representing the irreducible positroid cell $\S\subset \GTNN$. Moreover assume that the internal vertices of $\mathcal G$ are either bivalent or trivalent. Let $t_{bw}$ be an edge weighting such that the network $(\mathcal G, t_{bw})$ represents $[A]$. 
Let $\sigma$ be a Kasteleyn signature on $\mathcal G$,  and let $K^{\sigma,wt}$ be the corresponding Kasteleyn matrix.

Then there exists a unique solution to Kasteleyn system of relations $(v^{(k)}_{b} (\vec x), R_w)$ such that for all $\vec x$
\begin{equation}\label{eq:bound_cond}
v^{(k)}_{b_j} (\vec x)\equiv \psi (\kappa_j, \vec x), \quad j\in [n].
\end{equation}
Such solution has the following properties:
\begin{enumerate}
\item For any $b\in \mathcal B$, $v^{(k)}_{b} (\vec x)$ is an untrivial analytic function in $\vec x$, where we assume that only a finite number of times vary;
\item There exists $\vec x_0$ such that $v^{(k)}_{b} (\vec x_0)\not =0$ for all $b\in \mathcal B$.
\end{enumerate}
Moreover, if ${\tilde t}_{b,w}$ is an edge weighting equivalent to $t_{bw}$ on $\mathcal G$, then the corresponding solution ${\tilde v}^{(k)}_{b} (\vec x)$ to the Kasteleyn system for the new weights differs from  $v^{(k)}_{b} (\vec x)$ by a non--zero multiplicative constant $c_b$ independent of $\vec x$  at any given internal black vertex $b$: 
\begin{equation}\label{eq:wave_change_weight}
{\tilde v}^{(k)}_{b} (\vec x) = c_b v^{(k)}_{b} (\vec x).
\end{equation}
Finally if $\sigma^{\prime}$ is another Kasteleyn signature on $(\mathcal G, t_{bw})$, then its Kasteleyn system of relations $(u^{(k)}_{b} (\vec x), R^{\prime}_w)$ for the boundary conditions as in (\ref{eq:bound_cond}) 
\[
u^{(k)}_{b_j} (\vec x)\equiv \psi(\kappa_j, \vec x), \quad j\in [n], \; \forall \vec x,
\]
fulfils
\begin{equation}\label{eq:wave_change_sign}
|u^{(k)}_{b} (\vec x) |= | v^{(k)}_{b} (\vec x)|, \quad \forall b\in\mathcal B, \forall \vec x.
\end{equation}
\end{proposition}

\begin{proof}
The existence, uniqueness and analyticity of $(v^{(k)}_{b} (\vec x), R_w)$ follow from Lemma \ref{lem:psi_kappa} and Theorem \ref{theo:sol_kas_sys_1}. Let $E_b= (E_{b,1}, \dots E_{b,n})\in \mathbb{R}^n$, $b\in \mathcal B$, be the solution of Kasteleyn system of relations described in Theorem \ref{theo:sol_kas_sys_2}, and let $c(\vec x)=(c_1(\vec x),\dots, c_{n-k} (\vec x))$ be as in Lemma \ref{lem:psi_kappa}. Then $v^{(k)}_{b} (\vec x)$ satisfies
\begin{equation}\label{eq:vbk}
v^{(k)}_{b} (\vec x) \, = \, \sum_{j=1}^n \, E_{b,j} \, \psi (\kappa_j,\vec x)\, =\,  \prec \,c(\vec x) , \, \bar A^o \,E_{b} \succ,
\end{equation}
where  we have used (\ref{eq:wave_orth}), $\prec \cdot, \cdot \succ$ denotes the usual inner product in $\mathbb R^{n-k}$, and $\bar A^o$ is the matrix orthogonal to $A$ in (\ref{eq:orthog}). By Theorem \ref{theo:kas_geo_sys}, there exist a base $I\in \mathcal M$ and a non-zero vector $\alpha_b =( \alpha_{b,1},\dots, \alpha_{b,n-k})$ such that
$E_b \, = \, \sum_{j\in \bar I}\alpha_{b,j} \, E_j$,
where $E_j$ is the $j$--th canonical basis vector in $\mathbb R^n$. Then the right hand side of (\ref{eq:vbk}) is not identically zero for any given internal vertex $b\in \mathcal B$. Therefore, for any choice of the soliton data there exists $\vec x_0$ such that $v^{(k)}_{b} (\vec x_0)\not =0$, for all $b\in \mathcal B$.

Finally, the relations between solutions of systems of relations for equivalent weightings and for equivalent Kasteleyn signatures follows from the properties of Kasteleyn systems of relations.
\end{proof}

\begin{corollary}\label{cor:norm_v}
Let the soliton data $(\mathcal K, [A])$ and the graph $\mathcal G$ be given. Let $(\mathcal G, f)$ be the reduced bipartite network in the disk representing $[A]$ in Proposition \ref{prop:KP_double}. Let
$\vec x_0$ be such that $v^{(k)}_{b} (\vec x_0) \not =0$ for all $b\in \mathcal B$, where $v^{(k)}_{b} (\vec x)$ is the solution to Kasteleyn system of relations of Proposition \ref{prop:KP_double}. Then
\[
\hat v^{(k)}_{b} (\vec x) = \frac{v^{(k)}_{b} (\vec x)}{v^{(k)}_{b} (\vec x_0)}, \quad b\in{\mathcal B}
\]
is an untrivial analytic function in $\vec x$ and its value is independent of the choice of Kasteleyn signature on $\Gamma$ and of the edge weighting in the equivalence class of the network.
\end{corollary}

Corollary \ref{cor:norm_v} implies that $\hat v^{(k)}_{b} (\vec x)$ depends only on the given soliton data $(\mathcal K, [A])$ and the initial time $\vec x_0$. Therefore we assign the value $\hat v^{(k)}_{b} (\vec x)$ to the normalized KP wave function $\hat \psi$  at the corresponding double points of the spectral curve $\Gamma$ whose dual graph is $\mathcal G$:
\begin{enumerate}
\item If the edge $e_j$ joins the boundary vertex $b_j$ to the vertex $w$, we denote $Q_j$ the point in $\Gamma_w$ glued to $\kappa_j$ and  we assign the value $\psi (\kappa_j, \vec x) $ to the KP--II wave function at $Q_j$ at time $\vec x$:  
\begin{equation}\label{eq:bQ}
\psi (Q_j, \vec x) \equiv  \psi (\kappa_j, \vec x),
\end{equation}
and the value $\hat \psi (\kappa_j, \vec x) $ to the normalized KP--II wave function at $Q_j$ at time $\vec x$:  
\begin{equation}\label{eq:bQ_1}
\hat \psi (Q_j, \vec x) \equiv \hat \psi (\kappa_j, \vec x).
\end{equation}
\item If $\Gamma_w, \Gamma_b$ are the rational components corresponding to the vertices $w$, $b$ joined by the edge $e_b$, we denote $Q_b\in \Gamma_w$, $P_b\in \Gamma_b$ the points where we glue these components; and we assign the value $v^{(k)}_{b} (\vec x)$ to the KP--II wave function at both $P_b$ and $Q_b$ at the time $\vec x$:  
\begin{equation}\label{eq:psi_PQ}
\psi (P_b, \vec x) \equiv \psi (Q_b, \vec x) = v^{(k)}_{b} (\vec x),
\end{equation}
and the value $\hat v^{(k)}_{b} (\vec x)$ to the normalized KP--II wave function at both $P_b$ and $Q_b$ at the time $\vec x$:  
\begin{equation}\label{eq:psi_PQ_1}
\hat \psi (P_b, \vec x) \equiv \hat \psi (Q_b, \vec x) = \hat v^{(k)}_{b} (\vec x).
\end{equation}
\end{enumerate}

Finally we extend the normalized KP wave function on each component $\Gamma_w, \Gamma_b$ as follows: 
\begin{enumerate}
\item we extend it to a constant function with respect to the spectral parameter on each rational component corresponding either to a black vertex or to a bivalent white vertex;
\item we extend it to a degree one meromorphic function in the spectral parameter on each rational component $\Gamma_w$ corresponding to a trivalent white vertex $w$. 
\end{enumerate}
By construction we obtain the desired KP divisor which is contained in the union of the ovals and is given by the union of the Sato divisor and  of the pole divisor at the components $\Gamma_w$. More precisely:
 
\begin{construction}\textbf{The KP wave function on $\Gamma$}\label{constr:KP}
Let $(\mathcal K, [A])$ be the given soliton data where $\mathcal K =\{ \kappa_1 < \cdots < \kappa_n\}$ and $[A]\in \S$. Let $\gamma^{(S)}_l (\vec x)$, $l\in [k]$, be the local coordinates of the Sato divisor at time $\vec x$, {\sl i.e.} the solutions to the characteristic equation $P(\zeta, \vec x)=0$, with $P$ as in (\ref{eq:char_eq}). 
Let $\Gamma\equiv \Gamma (\mathcal G)$ be the reducible $\mathtt M$--curve whose dual graph $\mathcal G=(\mathcal B\cup \mathcal W, \mathcal E)$ is a reduced trivalent planar bipartite graph in the disk with black boundary vertices representing the irreducible positroid cell $\S\subset \GTNN$. 
Let $\vec x_0$, $v_b (\vec x)$ and $\hat v_b (\vec x)$ be as in Corollary \ref{cor:norm_v},
and let $ \psi$ and $\hat \psi$ be defined as in (\ref{eq:bQ})--(\ref{eq:psi_PQ_1}) at the double points of $\Gamma$.

We then define the normalized KP--II wave function on $\Gamma$ as follows:
\begin{enumerate}
\item On $\Gamma_0$ it coincides with the normalized wave function in (\ref{eq:KP_norm_rel}), that is for $P\in\Gamma_0\backslash\{ P_0\}$ and for any $\vec x$, in the natural local coordinate $\zeta$ such that $\zeta^{-1} (P_0)=0$,
\[
\hat \psi (\zeta(P) , \vec x) = \frac{\prod_{l=1}^k (\zeta-\gamma^{(S)}_l (\vec x))}{\prod_{l=1}^k (\zeta-\gamma^{(S)}_l (\vec x_0))}\exp( \zeta (x-x_0) + \zeta^2 (y-y_0) + \zeta^3 (t-t_0) + \zeta^4 (x_{4}-x_{4,0}) + \cdots);
\]
\item On each rational component $\Gamma_b$ of $\Gamma$ represented by an internal black vertex $b$, the normalized wave function $\hat \psi$ takes the same value at all marked points $P_l\in \Gamma_b$. Thus on $\Gamma_b$ we extend the normalized wave function $\hat \psi$ to a function constant with respect to the spectral parameter:
\[
\hat \psi (P, \vec x) \equiv \hat \psi (P_l, \vec x), \quad\quad \forall P\in \Gamma_b,\;\; \forall \vec x;
\]
\item On each rational component $\Gamma_w$ of $\Gamma$ represented by a bivalent white vertex $w$, let $Q_i\in \Gamma_w$, $i=1,2$, be the marked points. Then $\hat \psi (Q_2, \vec x) = \hat \psi (Q_1, \vec x)$, for all $\vec x$.
Thus on $\Gamma_w$ we extend the normalized wave function $\hat \psi$ to a function constant with respect to the spectral parameter:
\[
\hat \psi (P, \vec x) \equiv \hat \psi (Q_1, \vec x), \quad\quad \forall P\in \Gamma_w,\;\; \forall \vec x;
\]
\item On each rational component $\Gamma_w$ of $\Gamma$ represented by a trivalent white vertex $w$, let $Q_i\in \Gamma_w$, $i\in [3]$, be its marked points, where we label $Q_{i}$ in increasing order clockwise. Then there is a unique point $P_w\in \Gamma_w$ such that $\hat \psi$ is extended to a degree one meromorphic function in the spectral parameter on $\Gamma_w$ with pole divisor contained in $\{ P_w \}$. A representation of $\hat \psi$ is 
\begin{equation}\label{eq:wave_white}
\hat \psi (\zeta (P), \vec x) = \frac{\hat \psi (Q_3, \vec x) \, \zeta \, -\, \hat \psi (Q_1, \vec x)\,  \gamma_w }{\zeta \, -\, \gamma_w},\quad\quad \forall P\in \Gamma_w, \;\; \forall \vec x,
\end{equation}
where $\zeta$ is the coordinate on $\Gamma_w$ such that $\zeta(Q_1)=0$, $\zeta(Q_2)=1$ and $\zeta(Q_3)=\infty$, $\gamma_w $ is the local coordinate of the divisor point $P_w$: $\zeta (P_w)= \gamma_w$, and 
\[
\gamma_w \, \hat \psi (Q_1, \vec x) \, + \, (1 \, -\, \gamma_w) \, \hat \psi (Q_2, \vec x) \, -\, \hat \psi (Q_3, \vec x) \, =\, 0, \quad\quad \forall \vec x.
\]
\end{enumerate}
\end{construction}

\begin{lemma}\textbf{Explicit expression of the divisor coordinates using Kasteleyn system of relations}
None of the divisor points in Contruction \ref{constr:KP} coincides with a double point of $\Gamma$.
Moreover, the local coordinate $\gamma_w$ of the divisor point $P_w\in \Gamma_w$ in (\ref{eq:wave_white}) takes the value
\begin{equation}\label{eq:KP_div_white}
\gamma_w \equiv \zeta( P_w)	\, = \, \frac{K^{\sigma,wt}_{b_1 w} \,  \psi (Q_1, \vec x_0)}{K^{\sigma,wt}_{b_1 w} \,  \psi (Q_1, \vec x_0)\, +\, K^{\sigma,wt}_{b_2 w}  \psi (Q_2, \vec x_0)},
\end{equation}
where $K^{\sigma,wt}_{b_i w}$, $i\in [3]$, are the three non zero entries of the Kasteleyn matrix at the $w$-th column used in Proposition \ref{prop:KP_double} to contruct the wave function $\psi$, that is 
\begin{equation}\label{eq:equ}
R_w (v^{(k)})\,  \equiv \, \sum_{i=1}^3\,  K^{\sigma,wt}_{b_i w} \, v^{(k)}_{b_i} (\vec x) \, = 	\, \sum_{i=1}^3 \, K^{\sigma,wt}_{b_i w}	\, 
 \psi (Q_i, \vec x)  \, =		\, 0.
\end{equation}
\end{lemma}

The proof is straightforward since one obtains (\ref{eq:equ}) substituting (\ref{eq:KP_div_white}) into (\ref{eq:wave_white}) for $\zeta(Q)=\zeta(Q_2)=1$.

Next the KP divisor is defined as the union of the Sato divisor and of the divisor points on the rational components represented by the trivalent white vertices of the graph at the normalization time $\vec x_0$.

\begin{definition}\label{def:DKP}{\textbf{The KP-II divisor on  $\Gamma$.}} 
Let the soliton data be $(\mathcal K, [A])$, $[A]\in \S$, with $\S\subset \GTNN$ a $g$--dimensional irreducible positroid cell, and let $\mathcal G$ be a reduced planar bipartite graph in the disk representing $\S$. Let $\Gamma$ be the reducible curve whose dual graph is $\mathcal G$ and let $\vec x_0$ be as in the above construction. Then the KP-II divisor $\DKP$ is the sum of the following $g$ simple poles,
\begin{enumerate}
\item The $k$ poles on $\Gamma_0$ coinciding with the Sato divisor  at $\vec x= \vec x_0$: $\DS =\{ P^{(S)}_1 (\vec x_0),\dots , P^{(S)}_k(\vec x_0)\}$;
\item The $g-k$ poles $P_{w_l}\equiv P_{w_l}(\vec x_0)\in \Gamma_{w_l}$ uniquely identified by the condition that, in the local coordinate defined above $\zeta(P_{w_l}) =\gamma_{w_l} (\vec x_0)$,
where $w_l$, $l\in [g-k]$, are the trivalent white vertices of $\mathcal G$ and $\gamma_{w_l}$ is as in (\ref{eq:KP_div_white}). 
\end{enumerate}
\end{definition} 

\smallskip

Then $\hat \psi$ is the desired extension to  the reducible spectral curve $\Gamma$ of the wave function on $\Gamma_0$ arising in the spectral problem for the KP-II soliton data $(\mathcal K, [A])$ .

\begin{theorem}\label{theo:KPeffvac} 
Let the data $(\mathcal K, [A], \mathcal G, \vec x_0)$, and $\hat \psi$ and $\DKP$ on $\Gamma$ be as in Construction \ref{constr:KP} and Definition \ref{def:DKP}, where we assume that a finite number of KP time variables may change. Then $\hat \psi$ is the KP wave function associated to the soliton data $(\mathcal K, [A])$ which extends $\hat \psi$ in (\ref{eq:KP_norm_rel}) from $\Gamma_0$ to $\Gamma=\Gamma(\mathcal G)$ and is uniquely identified by the normalization condition $\hat \psi (P, \vec x_0)=1$ at all points $P\in \Gamma\backslash \{P_0\} $. Moreover $\hat \psi$ has the following properties on $\Gamma\backslash\{P_0\}$:
\begin{enumerate}
\item $\hat \psi$ is analytic in $\vec x$;
\item $\hat \psi$ takes the same value at pairs of glued points $P,Q\in \Gamma$, $\hat \psi(P, \vec x) = \hat \psi(Q, \vec x)$, for all $\vec x$;
\item $\hat \psi$ is meromorphic in $P\in \Gamma \backslash \{P_0\}$. More precisely, $\hat \psi(\zeta, \vec x)$ is either constant or meromorphic of degree one w.r.t. to the spectral parameter on each rational component of $\Gamma$ corresponding to a trivalent white vertex of $\mathcal G$. $\hat \psi(\zeta, \vec x)$ is constant w.r.t. to the spectral coordinate $\zeta$ on each other component corresponding to an internal vertex;
\item ${\hat \psi} (\zeta(P), \vec x)$ is real for real values of both the spectral coordinate $\zeta$ and the KP--times $\vec x$ on $\Gamma$;
\item $\DKP$ is the pole divisor of $\hat \psi$ for all $\vec x$: $\DKP+(\hat\psi(P,\vec x))\ge 0$ for all $\vec x$;
\item The KP--II divisor $\DKP$ is contained in the union of the ovals of $\Gamma=\Gamma(\mathcal G)$ and depends only on the soliton data and the normalization time $\vec x_0$;
\item\label{item:6} None of the divisor points in $\DKP$ coincides with any of the double points of the curve $\Gamma$.
\end{enumerate}
\end{theorem}

\begin{figure}
  \centering
  \includegraphics[width=0.4\textwidth]{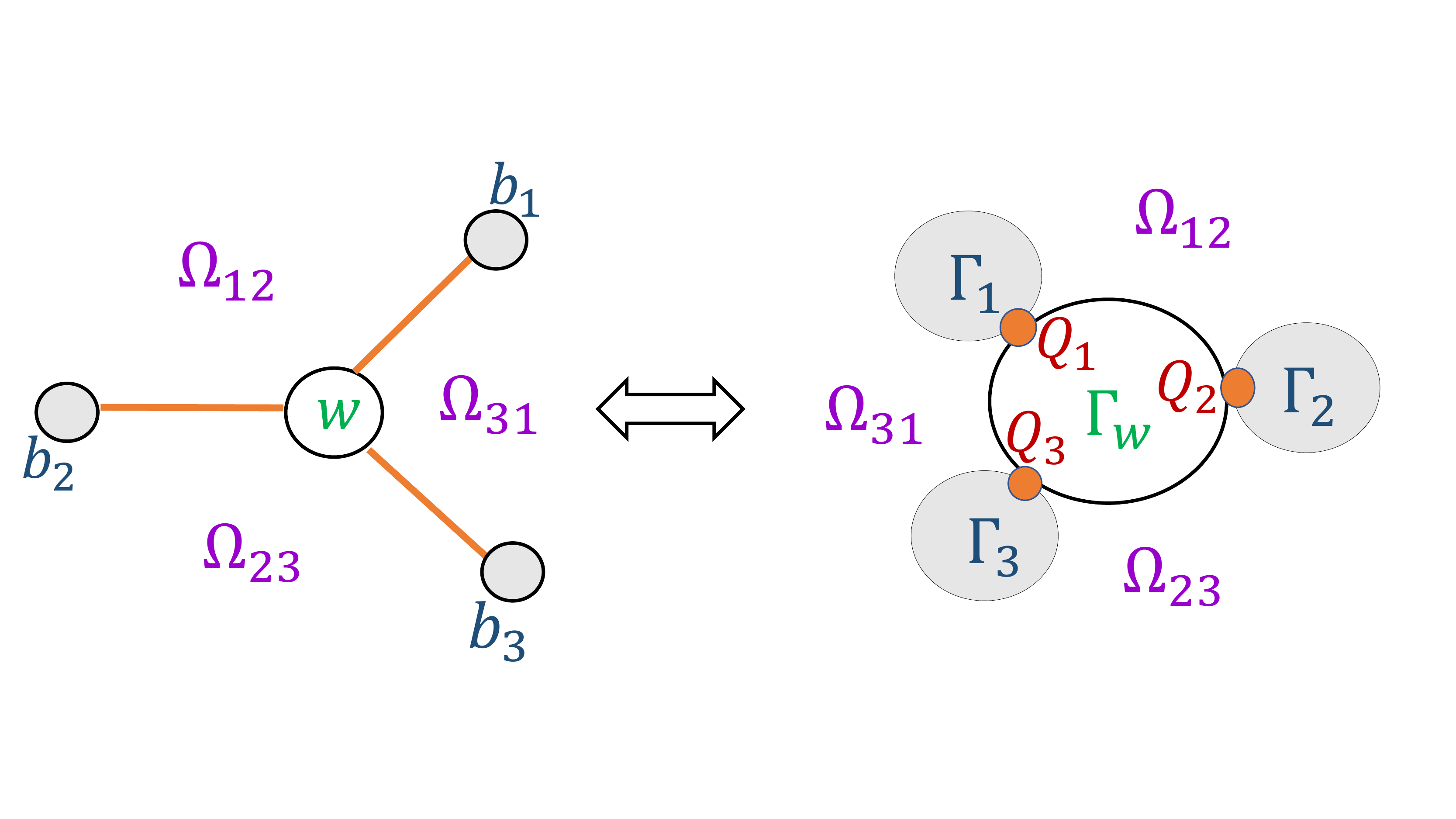}
	\vspace{-1. truecm}
  \caption{\small{\sl The correspondence between faces at the white vertex $w$ (left) and ovals bounded by $\Gamma_w$ (right) under the assumption that the curve is constructed reflecting the graph w.r.t. a vertical ray. Objects paired by the duality relation between the graph and the curve share the same color: vertices correspond to rational components, edges to double points, faces to ovals.}}
	\label{fig:corr_V_G}
\end{figure}

We remark that Item (\ref{item:6}) in Theorem \ref{theo:KPeffvac} follows from the condition $\psi (P, \vec x_0) \not =0$ at the double points.

Any trivalent white vertex $w$  bounds three faces; therefore the corresponding rational component $\Gamma_w$ bounds three ovals (see Figure \ref{fig:corr_V_G}). Next Lemma provides a simple criterion to detect the oval containing the divisor point $P_w\in \Gamma_w$.

\begin{lemma} \textbf{The position of the divisor points in the ovals} \label{lem:pos_gen} 
Let us define
\begin{equation}\label{eq:tilde_half}
\tilde z_{w, b_i} \equiv K^{\sigma,wt}_{b_{i} w} \,  \psi (P_{i}, \vec x_0), \quad \quad i\in [3],
\end{equation}
where $w$ is a white vertex of $\mathcal G$ and notations are as in Construction \ref{constr:KP}. Let $\Omega_{ij}$ denote both the face in $\mathcal G$ bounded by the edges $\overline{b_i w}$ and $\overline{b_j w}$, and the oval in $\Gamma=\Gamma(\mathcal G)$ to which the marked points $Q_i$ and $Q_j$ belong, $i,j \in [3]$.
Then the KP divisor point $P_w\in \Gamma_w$ belongs to $\Omega_{ij}$ if and only if $\tilde z_{w,b_i}$ and $\tilde z_{w,b_j}$ share the same sign
\[
P_w \in \Omega_{ij} \quad\quad \iff \quad\quad \tilde z_{w,b_i}\tilde z_{w,b_j}> 0.
\]
\end{lemma}

The proof immediately follows comparing the local coordinate of the divisor point $\zeta (P_w)$ in (\ref{eq:KP_div_white}) with those of the marked points, $\zeta(Q_1)=0$, $\zeta(Q_2)=1$ and $\zeta(Q_3)=\infty$ on $\Gamma_w$.

\smallskip

Finally, we prove that the divisor $\DKP$ satisfies the reality and regularity conditions settled in \cite{DN}:

\begin{theorem}\label{theo:pos_div}\textbf{Number of divisor points in the ovals}
There is exactly one divisor point in each finite oval $\Omega_s$, $s\in [g]$, and no divisor point in the infinite oval $\Omega_0$.
\end{theorem}

\begin{proof}
To simplify notations we use the same symbol $\Omega$ to denote both the face of the graph $\mathcal G$ and the oval of the curve $\Gamma = \Gamma(\mathcal G)$. Let $|\Omega|$ denote the number of edges bounding $\Omega$ and let $\tilde z_{e}$ be as in (\ref{eq:tilde_half})
\[
\tilde z_{e} \equiv K^{\sigma,wt}_{b w} \,  \psi (P_{b}, \vec x_0), \quad \mbox{ if } e = \overline{bw}.
\]
Then
\begin{equation}\label{eq:tilde_z}
\prod_{e\in \partial \Omega} \tilde z_e = \left\{ \begin{array}{ll}
\displaystyle
(-1)^{\frac{|\Omega|}{2}-1}\; \prod_{e\in \partial \Omega} t_e \; \prod_{b\in \partial \Omega} \big(  \psi (P_{b}, \vec x_0) \big)^2,
& 	\quad \mbox{ if } \Omega \mbox{ internal};\\
\noalign{\medskip}
\displaystyle
(-1)^{\frac{|\Omega|}{2}-1} \; \prod_{e\in \partial \Omega} t_e \; \prod_{\substack{ b\in \partial \Omega\\ b \mbox{\tiny{ internal}}}} \big(  \psi (P_{b}, \vec x_0) \big)^2 \; \prod_{\substack{ b_j\in \partial \Omega\\ b_j \mbox{\tiny{ b.ry vertex }}}}  \psi (\kappa_j, \vec x_0) , & \quad \mbox{ if } \Omega \mbox{ finite external}.
\end{array}
\right.
\end{equation}

Let $\nu_\Omega$ denote the total number of divisor points in the oval $\Omega$ associated to the white vertices $w$ bounding the face $\Omega$. 
From Lemma \ref{lem:pos_gen}, the divisor point $P_w \in \Gamma_w \cap \Omega$ if and only if $\tilde z_{e_{i_1}}\tilde z_{e_{i_2}}>0$, where the edges $e_{i_1}=\overline{b_{i_1}w}$ and $e_{i_2}=\overline{b_{i_2}w}$ bound $\Omega$ at $w$.

Next, let $c_{\Omega}$ denote the number of white vertices $w\in \partial \Omega$ such that the product $\tilde z_{e_i}\tilde z_{e_{i+1}}<0$ at the consecutive edges $e_i=\overline{b_iw},e_{i+1}=\overline{b_{i+1}w}\in \partial \Omega$. Obviously
\begin{equation}\label{eq:nu_int}
\nu_{\Omega} = \frac{|\Omega|}{2} -c_{\Omega}.
\end{equation}
If $\Omega$ is an internal face, (\ref{eq:tilde_z}) implies that 
\begin{equation}\label{eq:c_int}
c_{\Omega} = \frac{|\Omega|}{2} -1, \mod 2,
\end{equation}
since edge weights are positive. Therefore, $\nu_\Omega$, the total number of divisor points belonging to the internal oval $\Omega$ is odd.

If $\Omega$ is a finite external face, then the total number of divisor points on the corresponding face equals $\nu_\Omega+\rho_{\Omega}$, where $\nu_\Omega$ is the number of divisor points belonging to $\bigcup\limits_{w\in \mathcal W}\Gamma_w\cap\Omega$, and $\rho_{\Omega}$ is the number of Sato divisor points belonging to $\Omega\cap \Gamma_0$. By definition there is an odd (respectively even) number of Sato divisor points in
$]\kappa_j, \kappa_{j+1}[$ if $ \psi (P_{b_j}, \vec x_0)  \psi (P_{b_{j+1}}, \vec x_0) \, < \, 0$ ( respectively $>0$).
Therefore $\rho_\omega$ has the same parity as $\prod_{ P_b\in \partial \Omega\cap \Gamma_0}  \psi (P_{b}, \vec x_0)$, whereas $\nu_\Omega$ satisfies (\ref{eq:nu_int}). In this case (\ref{eq:tilde_z}) implies that 
\begin{equation}\label{eq:c_int_{53}}
c_{\Omega} = \frac{|\Omega|}{2} -1 +\rho_{\Omega}, \mod 2;
\end{equation}
therefore, $\nu_{\Omega}+\rho_{\Omega}$, the total number of divisor points in the finite external oval $\Omega$, is odd.

This ends the proof since the cardinality of the KP divisor coincides with the number of finite faces of the graph.
\end{proof}

Let us apply the construction of the KP wave function and of the divisor to the example shown in Figure \ref{fig:ex_Gr26_sys_rel} (see also Example \ref{ex:6}).

\begin{figure}
  \centering
  \includegraphics[width=0.37\textwidth]{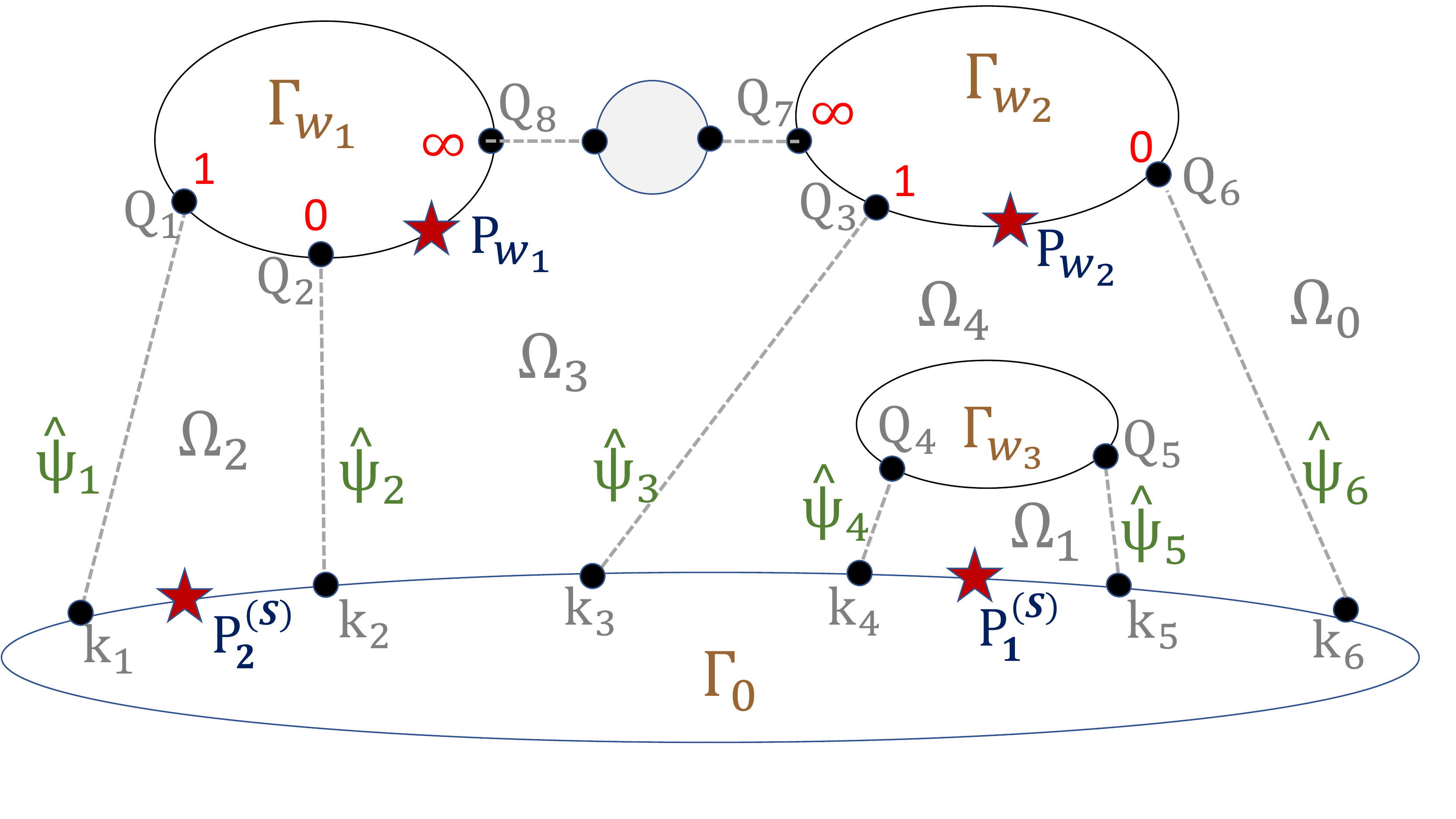}
	\hspace{.8 truecm}
	\includegraphics[width=0.37\textwidth]{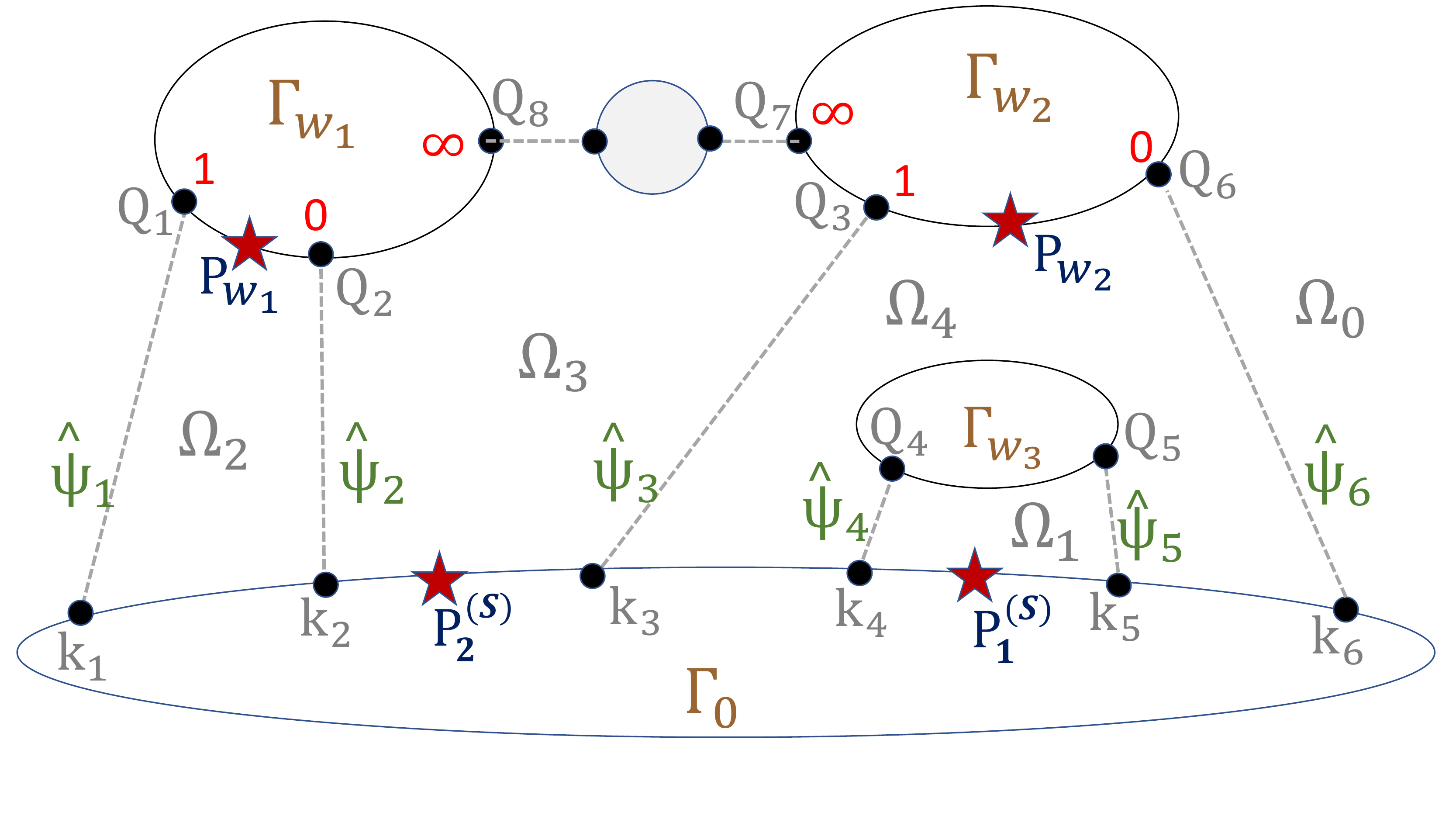}
	\includegraphics[width=0.37\textwidth]{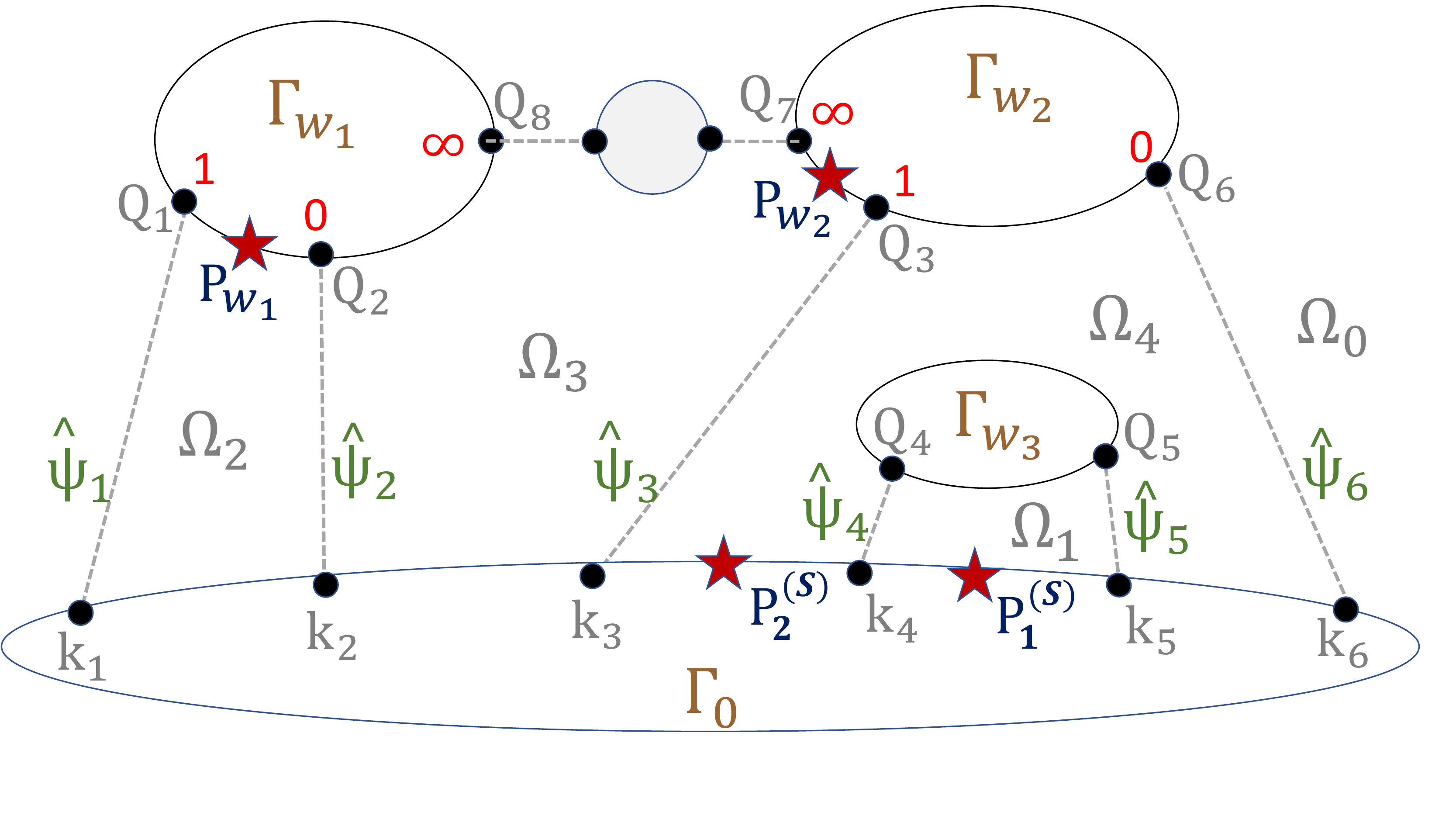}
	\hspace{.8 truecm}
	\includegraphics[width=0.37\textwidth]{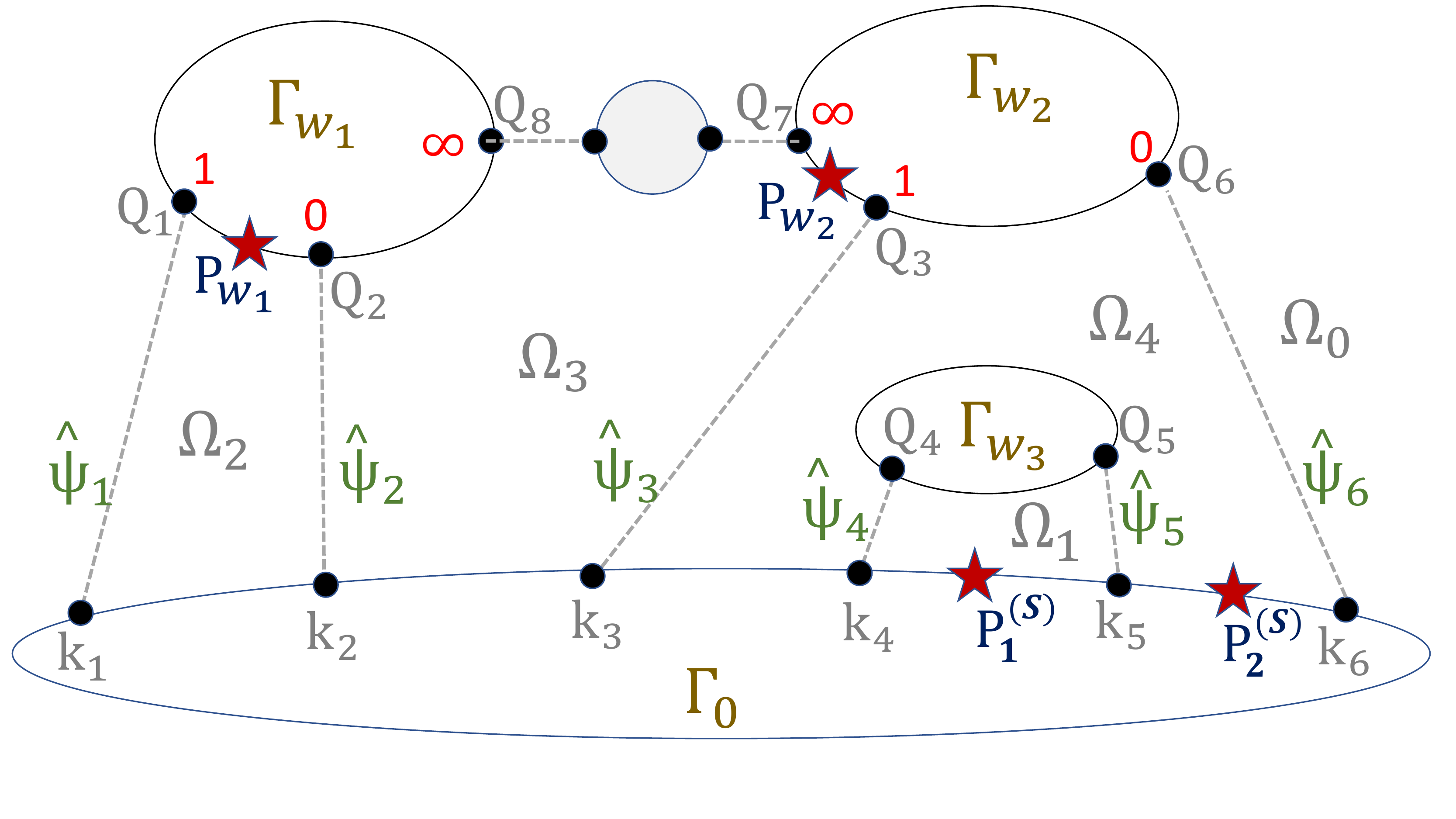}
	\vspace{-.3 truecm}
  \caption{\small{\sl The four possible  configurations of real regular KP divisors on the reducible rational curve whose dual graph is in Figure \ref{fig:curve_graph}, for the generic soliton data $({\mathcal K}, [A])$ of Example \ref{ex:7} (see also Figure \ref{fig:ex_Gr26_sys_rel}). Double points are represented by dashed lines and divisor points by stars; we mark the local coordinates at the marked points of $\Gamma_{w_1}$ and $\Gamma_{w_2}$ in red.}}
	\label{fig:config}
\end{figure}

\begin{example}\label{ex:7}
Given the soliton data $(\mathcal K = \{\kappa_1< \cdots < \kappa_6\}, [A])$ with $[A]\in \S\subset Gr^{\mbox{\tiny TNN}}(2,6)$ as in Example \ref{ex:6}, a basis of heat hierarchy solutions generating the KP--II multi--line soliton solution is
\[
f_1(\vec x) = e^{\theta_1(\vec x)}+t_{21}^{-1} e^{\theta_2(\vec x)}+t_{21}^{-1} t_{11}e^{\theta_3(\vec x)}-t_{21}^{-1} t_{11}t_{72} e^{\theta_6(\vec x)},\quad\quad
f_2(\vec x) = e^{\theta_4(\vec x)}+t_{53}^{-1}  e^{\theta_5(\vec x)}.
\]
The dressing operator
\[
\mathfrak D = \partial_x^2 -\mathfrak{w}_1 (\vec x) \partial_x -\mathfrak{w}_2 (\vec x),
\]
satisfies ${\mathfrak D} (f_i)\equiv 0$. Then the KP--II multi--line soliton solution is
\[
u(\vec x) = 2\partial_x^2 \log (\tau (\vec x)),
\]
where $\tau(\vec x) = f_1 (\vec x) \partial_x f_2 (\vec x) -f_2 (\vec x) \partial_x f_1 (\vec x)$.

The Sato divisor is $\DS = \{ P^{(S)}_1, P^{(S)}_2\}$ where $\gamma^{(S)}_i \equiv \zeta (P^{(S)}_i)$, $i\in [2]$, are the roots of the characteristic equation at time $\vec x_0$,
\[
(\zeta -\gamma^{(S)}_1) (\zeta -\gamma^{(S)}_1) \equiv \zeta^2 -\mathfrak{w}_1 (\vec x_0) \zeta -\mathfrak{w}_2 (\vec x_0).
\]
On $\Gamma_0\backslash \{P_0\}$ the auxiliary wave function $ \psi$ takes the form
\[
 \psi (\zeta, \vec x) \, = \, (\zeta^2 -\mathfrak{w}_1 (\vec x) \zeta -\mathfrak{w}_2 (\vec x)) \exp(\zeta x+\zeta^2 y +\zeta^3 t +\cdots).
\]

The relation between the graph and the reducible $\mathtt M$--curve is shown in Figure \ref{fig:curve_graph}.
By construction for all $\vec x$
\[
 \psi (Q_j, \vec x) = \psi (\kappa_j, \vec x) , \quad j\in [6], \quad\quad   \psi (Q_7, \vec x) = \psi (Q_8, \vec x).
\]
$ \psi$ takes opposite signs at the marked points $\kappa_4$ and $\kappa_5$, since at the corresponding white vertex we have
\[
 \psi (\kappa_4, \vec x) \, + 	\, t_{72} 	\,  \psi (\kappa_5, \vec x) \equiv 0, \quad\quad \forall \vec x.
\]
Therefore one of the two Sato divisor points, say $P^{(S)}_1$, has local coordinate in $[\kappa_4, \kappa_5]$ and, necessarily, $P^{(S)}_2$ has local coordinate in $[\kappa_1, \kappa_4]\cup[\kappa_5,\kappa_6]$.

At the trivalent white vertex $w_1$ we have the relation
\[
t_{11} \,  \psi (Q_8, \vec x)\,  +\,  t_{21} \,  \psi (\kappa_1, \vec x)\,  + \,  \psi (\kappa_2, \vec x) \, \equiv \, 0.
\]
If we assign local coordinates $\zeta(Q_2)=0$, $\zeta(Q_1)=1$ and $\zeta(Q_8)=\infty$ on $\Gamma_{w_1}$ (see also Figure \ref{fig:config}), the divisor point $P_{w_1}$ on such rational component has local coordinate
\[
\gamma_{w_1} \, = \, \zeta (P_{w_1}) \, = \, \frac{ \psi (\kappa_2, \vec x_0)}{t_{21} \,  \psi (\kappa_1, \vec x_0) \, + \,  \psi (\kappa_2, \vec x_0)}.
\]
Similarly, at the trivalent white vertex $w_2$ we have the relation
\[
 \psi (\kappa_3, \vec x)\, - \,  \psi (Q_7, \vec x) \, - \, t_{72} \,  \psi (\kappa_6, \vec x) \,  \equiv \, 0.
\]
If we assign local coordinates $\zeta(Q_6)=0$, $\zeta(Q_3)=1$ and $\zeta(Q_7)=\infty$ on $\Gamma_{w_2}$, the divisor point $P_{w_2}$ has local coordinate
\[
\gamma_{w_2} \, = \, \zeta (P_{w_2}) \, = \, \frac{ \psi (\kappa_3, \vec x_0)}{ \psi (\kappa_3, \vec x_0) \, -\,  t_{72} \,  \psi (\kappa_6, \vec x_0)}.
\]
Finally let us compute the position of the divisor. For simplicity we omit the time dependence, using the following abridged notation $ \psi_j = \psi (\kappa_j, \vec x_0)$, and use notation $]a,b[$ for open intervals. By construction $ \psi_1,  \psi_6\ge 0$ and $ \psi_1 +t_{21}^{-1}   \psi_2 + t_{21}^{-1} t_{11}  \psi_3 -t_{21}^{-1} t_{11}t_{72}  \psi_6=0$. Then the following divisor configurations occur for generic soliton data $(\mathcal K, [A])$:
\begin{enumerate}
\item If $\gamma^{(S)}_2 \in ]\kappa_1,\kappa_2[$, that is $P^{(S)}_2\in \Omega_2$, then $\gamma_{w_1} <0$, $\gamma_{w_2} \in ]0,1[$, {\sl i.e} $P_{w_1}\in\Omega_3$, $P_{w_2}\in\Omega_4$. This configuration is represented in Figure \ref{fig:config} [top,left];
\item If $\gamma^{(S)}_2 \in ]\kappa_2,\kappa_3[$, that is $P^{(S)}_2\in \Omega_3$, then $\gamma_{w_1}\in ]0,1[$, $\gamma_{w_2} \in ]0,1[$, {\sl i.e}  $P_{w_1}\in\Omega_2$ and $P_{w_2}\in\Omega_4$. This configuration is represented in Figure \ref{fig:config} [top,right];
\item If $\gamma^{(S)}_2 \in ]\kappa_3,\kappa_4[\cup ]\kappa_5,\kappa_6[$, that is $P^{(S)}_2\in \Omega_4$, then $\gamma_{w_1} \in ]0,1[$, $\gamma_{w_2} >1 $, {\sl i.e.} $P_{w_1}\in\Omega_2$ and $P_{w_2}\in\Omega_3$. This configuration is represented in Figure \ref{fig:config} [bottom].
\end{enumerate}
\end{example}
In \cite{AG5}, the construction of the KP--II wave function at the double points of the curve whose dual graph is $\mathcal G$ is performed solving Lam system of relations for the geometric signature $\epsilon_I$ imposing the boundary conditions
\[
z^{(g)}_{b_j} (\vec x) =  \psi (\kappa_j, \vec x), \quad\quad j\in \bar I.
\]
The un-normalized KP wave function is defined as follows for all $b \in \mathcal B$:
\[
\psi^{(g)} (P_b, \vec x) = z^{(g)}_{b} (\vec x).
\] 
From Theorem \ref{theo:kas_geo_sys} it follows that the normalized wave function 
\[
\hat \psi(P, \vec x) = \frac{{ \psi}^{(g)} (P, \vec x)}{{ \psi}^{(g)} (P, \vec x_0)}
\]
constructed in \cite{AG5} coincides with the one defined in Construction \ref{constr:KP}, and necessarily the KP--II divisors are the same.

\begin{proposition}\label{prop:KP_equiv}
The normalized wave function $\hat \psi(P,\vec x)$ and the KP--II divisor $\DKP$ in Construction \ref{constr:KP} and Definition \ref{def:DKP} coincide with those defined in \cite{AG5} using Lam system of relations for the geometric signatures on $\Gamma(\mathcal G)$ for any given soliton data $(\mathcal K, [A])$ and reference time $\vec x_0$.
\end{proposition}

\bibliographystyle{alpha}

\end{document}